\documentclass[
onecolumn, aps, pra,
tightenlines,
11pt,
longbibliography,
showpacs,
nofootinbib,
notitlepage,
superscriptaddress]{revtex4-2}

\usepackage{subfiles}
\usepackage{amsmath, amssymb, amsthm, mathtools}
\usepackage{mathrsfs}
\usepackage{braket}
\usepackage{booktabs}
\usepackage{graphicx}
\usepackage{quantikz}
\usetikzlibrary{arrows.meta, positioning}
\usepackage{subcaption}
\usepackage{ragged2e}
\usepackage{enumitem}

\makeatletter
\long\def\revtex@flush@makecaption#1#2{%
  \par
  \vskip\abovecaptionskip
  \begingroup
    \footnotesize\rmfamily
    \samepage
    \flushing
    \let\footnote\@footnotemark@gobble
    \@make@capt@title{#1}{#2}\par
  \endgroup
  \vskip\belowcaptionskip
}%
\AtBeginDocument{\let\@makecaption\revtex@flush@makecaption}%

\makeatother

\newcommand{\eps}{\varepsilon}
\newcommand{\osc}{\operatorname{osc}}
\newcommand{\Ncl}{N_{\mathrm{cl}}}
\newcommand{\Nq}{N_{\mathrm{q}}}
\newcommand{\R}{\mathbb{R}}

\newcommand{\T}{\mathbb{T}}
\newcommand{\Z}{\mathbb{Z}}

\newcommand{\vol}{\operatorname{vol}}

\newcommand{\onenorm}[1]{\Vert #1\Vert_1}
\newcommand{\N}{\mathbb{N}}

\newcommand{\Om}{\Omega_0}

\usepackage{titlesec}
\titleformat{\paragraph}[runin]{\normalfont\itshape}{}{0pt}{}[.]
\titlespacing*{\paragraph}{0pt}{\smallskipamount}{1em}

\usepackage[colorlinks, allcolors=black]{hyperref}
\usepackage[capitalize, nameinlink]{cleveref}

\usepackage{setspace}
\usepackage{float}
\makeatletter
\let\float@end\float@end@ltx
\let\newfloat\newfloat@ltx
\makeatother

\usepackage{algorithm, algcompatible}
\usepackage{algpseudocode}

\makeatletter
\renewcommand{\fnum@algorithm}{\ALG@name~\thealgorithm}
\makeatother

\allowdisplaybreaks

\newtheorem{lem}{Lemma}
\newtheorem{defn}{Definition}
\newtheorem{rmk}{Remark}
\newtheorem{cor}{Corollary}
\newtheorem{prop}{Proposition}
\newtheorem{thm}{Theorem}

\newtheorem*{thm*}{Theorem}

\Crefname{prop}{Proposition}{Propositions}
\Crefname{thm}{Theorem}{Theorems}
\Crefname{lem}{Lemma}{Lemmas}
\Crefname{defn}{Definition}{Definitions}
\Crefname{rmk}{Remark}{Remarks}
\Crefname{exm}{Example}{Examples}
\Crefname{alg-line}{Line}{Lines}
\crefname{alg}{Algorithm}{Algorithms}
\Crefname{cor}{Corollary}{Corollaries}
\Crefname{appendix}{Appendix}{Appendices}

\DeclareMathOperator{\polylog}{polylog}

\DeclareMathOperator{\diag}{diag}

\renewcommand{\vec}[1]{\overrightarrow{#1}}
\renewcommand{\epsilon}{\varepsilon}
\renewcommand{\tilde}[1]{\widetilde{#1}}
\renewcommand{\hat}[1]{\widehat{#1}}

\newcommand{\norm}[1]{\left\lVert#1\right\rVert}

\mathcode`l="8000
\begingroup
\makeatletter
\lccode`\~=`\l
\DeclareMathSymbol{\lsb@l}{\mathalpha}{letters}{`l}
\lowercase{\gdef~{\ifnum\the\mathgroup=\m@ne  \ell \else \lsb@l \fi}}%
\endgroup

\begin{document}

\title{Provable Quantum--Classical Separation for Continuous Gibbs Sampling}

\author{Enrico~Olivucci}
\altaffiliation{Equal contribution authors.}
\affiliation{Irréversible Inc., Sherbrooke, Québec, Canada}

\author{Mariia~Sobchuk}
\altaffiliation{Equal contribution authors.}
\affiliation{Institute for Quantum Computing, University of Waterloo, Waterloo, Ontario N2L 3G1, Canada}
\affiliation{Department of Physics and Astronomy, University of Waterloo, Waterloo, Ontario N2L 3G1, Canada}

\author{Sehmimul~Hoque}
\altaffiliation{Equal contribution authors.}
\affiliation{Institute for Quantum Computing, University of Waterloo, Waterloo, Ontario N2L 3G1, Canada}
\affiliation{Department of Physics and Astronomy, University of Waterloo, Waterloo, Ontario N2L 3G1, Canada}
\affiliation{Perimeter Institute for Theoretical Physics, Waterloo, Ontario N2L 2Y5, Canada}

\author{Jeffrey~Hnybida}
\affiliation{Irréversible Inc., Sherbrooke, Québec, Canada}

\author{Kyungho W.~Kim}
\affiliation{Institute for Quantum Computing, University of Waterloo, Waterloo, Ontario N2L 3G1, Canada}
\affiliation{Department of Physics and Astronomy, University of Waterloo, Waterloo, Ontario N2L 3G1, Canada}

\author{Ala~Shayeghi}
\affiliation{Institute for Quantum Computing, University of Waterloo, Waterloo, Ontario N2L 3G1, Canada}
\affiliation{National Research Council Canada, Waterloo, Ontario, N2L 3G1, Canada}

\author{Pooya~Ronagh}
\email[Corresponding author: ]{pooya.ronagh@uwaterloo.ca}
\affiliation{Institute for Quantum Computing, University of Waterloo, Waterloo, Ontario N2L 3G1, Canada}
\affiliation{Department of Physics and Astronomy, University of Waterloo, Waterloo, Ontario N2L 3G1, Canada}
\affiliation{Perimeter Institute for Theoretical Physics, Waterloo, Ontario N2L 2Y5, Canada}
\affiliation{Microsoft, Redmond, WA 98052, USA}

\begin{abstract}
We prove the first quantum--classical separation for a sampling problem over a continuous domain. For a class of Gibbs states $p\propto e^{-\beta E}$ on the torus $\mathbb{T}^d$ with smooth ($s$-Gevrey) potential and barrier amplitude $\alpha=e^{\beta\Delta}$, where $\Delta = \max E-\min E$, every classical algorithm querying the value, gradient, or any higher-order derivatives of the log-density requires $\Omega(\alpha)$ queries to sample at constant accuracy in total variation distance, while a quantum algorithm based on quantum singular value thresholding and temperature annealing samples with $\tilde{O}\left(\sqrt{\alpha}\right)$ queries to an oracle for the gradient. The advantage is quadratic in the barrier amplitude, which becomes exponential in the dimension, $e^{\Omega(d)}$, at low temperature. The classical bound is information-theoretic, holding for every classical algorithm with query access to the Gibbs potential and its derivatives at any order.
\end{abstract}

\maketitle

\tableofcontents

\section{Introduction}
\label{sec:introduction}

Provable separations between quantum and classical computation are rare, and the canonical examples concern discrete (and often contrived) tasks: Grover's unstructured search~\cite{grover1996search}, the hidden-structure problems of Bernstein--Vazirani and Simon~\cite{bernstein1997quantum,simon1997power}, promise problems such as Deutsch--Jozsa and the collision problem~\cite{deutsch1992rapid,brassard1998collision,shi2002collision}, the quantum-walk speedups for welded trees and element distinctness~\cite{childs2003welded,ambainis2007distinctness}, and the forrelation problem~\cite{aaronson2015forrelation,raz2018oracle}. On continuous domains, quantum--classical query separations have been established for numerical integration~\cite{novak2001integration}, function approximation~\cite{heinrich2004approximation}, differential equations and eigenvalue problems~\cite{kacewicz2006initial,heinrich2006elliptic,papageorgiou2005sturm}, and continuous optimization~\cite{gocwin2005maximum,leng2025optimization}.

Here we seek a separation for continuous-domain Gibbs sampling against all classical algorithms with local-oracle access. Several recent papers have proposed quantum Gibbs samplers for continuous potentials~\cite{motamedi2024gibbs,leng2026operator,ozgul2024stochastic,childs2022logconcave} (also see~\cite{terhal2000equilibration,poulin2009sampling,chowdhury2017gibbs} for Gibbs samplers on discrete domains). However, each result has been benchmarked against the best classical competitors known, and such comparisons cannot rule out faster algorithms that may be devised in the future. A guaranteed separation requires a lower bound against all possible classical algorithms in the specified oracle model.

\begin{figure}[b]
\begin{minipage}[t]{0.5\textwidth}
\includegraphics[width=\linewidth, trim=80pt 90pt 80pt 145pt, clip]{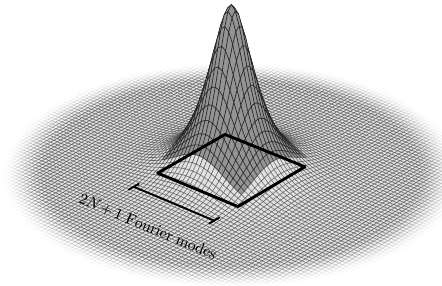}
\subcaption{}
\end{minipage}\hfill
\begin{minipage}[t]{0.5\textwidth}
\includegraphics[width=\linewidth, trim=23pt 36pt 14pt 137pt, clip]{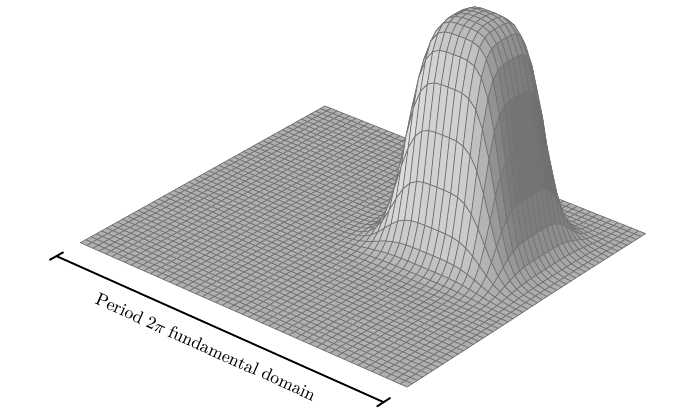}
\subcaption{}
\end{minipage}
\caption{The role of Gevrey smoothness in our separation result. (a) The schematic of the decay of the Fourier modes of a periodic smooth function. The QSVT-based quantum algorithm requires truncating the Fourier transform of the Gibbs density to a threshold $N$. Mere smoothness (i.e., the $(s=\infty)$-Gevrey class) is too weak to guarantee $N= \polylog(1/\epsilon)$, required for our quantum upper bound result. (b) The schematic of a bump function defined on a torus. Such functions hide their center of mass in a small area (patch) on the torus, and queries of any order to points outside the patch provide no information to help the search for the patch. These functions are not analytic (because any analytic function that is constant on an open interval, is constant everywhere on the domain). Therefore analyticity (i.e., the $(s=1)$-Gevrey class) is too strong to guarantee existence of hard instances required for our classical lower bound result.}
\label{fig:hero-plot}
\vspace{-1pt}
\end{figure}

We prove such a separation. Sampling from a Gibbs distribution $p\propto e^{-\beta E}$ of a potential $E$ at inverse temperature $\beta > 0$ is a recurring primitive of computational science: it underlies Bayesian inference and generative modeling~\cite{du2019implicit,ho2020denoising,song2021train}, thermal averages in statistical physics~\cite{ruelle1999statmech}, and the elementary step of annealing-type optimization. On a continuous domain, its difficulty is governed less by the dimension alone than by the landscape of the potential: wells separated by barriers concentrate probability mass in regions that local exploration struggles to find. The natural metric for that difficulty is therefore the barrier amplitude $\alpha=e^{\beta\Delta}$, with $\Delta\coloneqq \max E-\min E$, i.e., the worst-case ratio between density values, which, by the Holley--Stroock theorem~\cite{holley1987logsobolev,schlichting2012eyringkramers}, upper-bounds the Poincar\'e constant of the Gibbs measure. Kramers-type metastability analysis shows that classical algorithms such as Langevin dynamics and related Markov chain Monte Carlo methods, while fast in the log-concave setting~\cite{dwivedi2018logconcave}, take time exponential in $\Delta$ to move between wells~\cite{bovier2004metastability,berglund2013kramers}.

This paper asks whether quantum computation provably reduces the complexity of Gibbs sampling and answers in the affirmative: for Gibbs distributions of controlled smoothness on the torus $\T^d=\R^d/ (2\pi \Z)^d$, quantum sampling is quadratically faster in $\alpha$, backed by an unconditional, information-theoretic classical lower bound. The controlled-smoothness criterion requires the gradient of the Gibbs potential $\beta\nabla E$ to be \emph{$s$-Gevrey} for some $s > 1$. The Gevrey parameter $s$ quantifies the amount of smoothness, ranging from real-analytic ($s=1$) to merely smooth ($s=\infty$)~\cite{rodino1993linear}. See \cref{fig:hero-plot} for an intuitive explanation of the critical role of Gevrey smoothness for the separation result proven in this paper. We focus on the subclass of $s$-Gevrey functions for which $s \neq 1, \infty$ and whose barrier amplitude is at most $\alpha$: the class $\mathscr{W}^s(\alpha, \xi , \rho)$ as defined in \cref{sec:main-result} (see \cref{def:potential-class}).

On this class, at constant total-variation accuracy, every classical algorithm querying \emph{any local oracle} of the potential $E$ that returns the value and derivatives of $E$ of any order at the query point $x$ must perform $\Ncl=\Omega(\alpha)$ queries. In comparison, a quantum algorithm with coherent access to the first-order oracle samples with $\Nq=\tilde{O}(\sqrt{\alpha})$\footnote{Throughout the paper for the asymptotic order $f(x)=O(g(x))$, and equivalently $g(x) = \Omega(f(x))$, the notation $\tilde{O}$ and $\tilde{\Omega}$ refer to dropping polylogarithmic dependencies from the leading order, namely $O(f(x)\log f(x)) =\tilde{O}(f(x))$ and $\Omega(f(x)/\log f(x)) =\tilde{\Omega}(f(x))$.} queries, up to factors polynomial in the dimension and the inverse Gevrey radius, and logarithmic in precision. The query complexity ratio is $\Ncl/\Nq = \tilde{\Omega}(\alpha^{1/2-1/d})$ up to such factors (\cref{thm:main-result}), approaching a quadratic quantum advantage in the barrier amplitude as $d$ grows. It becomes $e^{\Omega(d)}$ once the barrier amplitude is exponential in the dimension, $\log\alpha = \Theta(d)$, for fixed $s$ and any accuracy with $\log\log\tfrac1\eps=o(d)$. This regime is reached when $\beta\Delta=\Omega(d)$, for example at $\beta=\Omega(d)$ with $\Delta=\Theta(1)$. To our knowledge, this is the first provable quantum--classical query separation for continuous-domain Gibbs sampling against arbitrary classical local-oracle algorithms.

The lower bound is an adversarial argument of the \emph{hide-and-seek} type in the tradition of the local-oracle lower bounds of information-based complexity and convex optimization~\cite{nemirovski1983problem,traub1988information}: the adversary hides a small patch of anomalous probability mass behind a perfectly flat landscape, at a uniformly random position, and locality forces every oracle to return the same instance-independent answer until a query lands on the patch. Effectively, each query probes only one out of $M \in O(\alpha)$ patches. The argument developed in \cref{sec:lower-bound} holds for every local oracle and identifies exactly where regularity matters: hiding probability mass in a flat landscape is impossible in the analytic case $s = 1$ but possible for every $s > 1$. Query lower bounds for the sampling task itself are scarce: prior results with sample accuracy measured in total variation concern Gibbs states on $\R^d$ either for log-concave targets~\cite{chewi2022oned,chewi2024query} or, in the non-log-concave regime, for Gibbs potentials with $L$-Lipschitz gradient and second moment $\mathbb{E}_p[\norm{x}^2] \leq M$ \cite{hezhang2025tvbounds}, and with access to a first-order oracle (cf.~\cite{talwar2019separations} for computational separations between sampling and optimization). The hard instances in~\cite{hezhang2025tvbounds} are in $\mathscr C^2(\R^d)$ with $\norm{\nabla^2 V(x)} \leq L$, so their lower bound holds on this smaller class too. To our knowledge, our lower bounds proven in this paper are the first in total variation for sampling periodic Gibbs states, defined on the compact domain $\T^d$, the first stated uniformly over arbitrary local oracles and, in the non-log-concave regime, the first proven with hard instances that are $\mathscr C^{\infty}$, as required for our quantum--classical separation result.

The upper bound builds on our revamped version of the original quantum singular value thresholding (QSVTh) Gibbs sampler of \cite{leng2026operator}, using finer-grid Fourier factors to control aliasing. The former samples from the Gibbs distribution $p \propto e^{-\beta E}$ by preparing the discretized Gibbs state via projection onto the kernel (assumed to exist and be unique) of a discretized Witten Laplacian, followed by continuous upsampling, but makes several simplifying assumptions that leave it short of an end-to-end, quadratic-speedup Gibbs sampler; we remove or prove each of them (see \cref{sec:quantum-gibbs-sampling} for the complete list). In particular, for our discretization we prove ground-state accuracy and the spectral lower bound needed for filtering, in place of \cite[Assumption~11]{leng2026operator}: comparison with the Galerkin restriction of the continuous Witten Laplacian, using Gevrey Fourier tails, gives the discrete gap by Weyl's inequality and the min--max principle, and ground-state accuracy by an energy estimate. This allows the discretized Gibbs state to be prepared by projecting onto the ground state of the Witten Laplacian, followed by Fourier upsampling with a discretization size that remains polylogarithmic in inverse precision, where the last step again uses the Gevrey property of the distribution, following~\cite{ronagh2026breaking}. The warm start required by the filtering step is dispensed with by temperature annealing~\cite{somma2007quantum,ozgul2024stochastic}: a schedule of inverse temperatures with increments small enough that consecutive Gibbs states overlap, so that each filtered state warm-starts the next, thus running the algorithm end to end on a uniform superposition as input with gradient queries only (\cref{sec:quantum-gibbs-sampling}).

The paper is organized as follows. \cref{sec:main-result} states the main theorem together with the two bounds it combines. \cref{sec:lower-bound} proves the classical lower bound: the hide-and-seek argument (\cref{ssec:hide-seek}) and the Gevrey bump construction (\cref{ssec:bumps}), leaving the formal query model and the Gevrey estimates in \cref{app:model,app:gevrey}. \cref{sec:quantum-gibbs-sampling} presents the quantum algorithm and its end-to-end query complexity, with the discretization (\cref{app:discretization}) and annealing (\cref{app:qsvt-anneal}) analysis in the appendix. Finally, \cref{sec:separation} works out the choice of class parameters (\cref{ssec:main-result}) and upper bounds the quantum query complexity at the class parameters of the lower bound, proving the separation.

\section{Statement of the Main Result}
\label{sec:main-result}

\paragraph{Setting}

We consider quantum and classical sampling algorithms for the Gibbs state of $d$-dimensional $2\pi$-periodic potentials $E: \R^d \to \R$ with density $p=e^{-\beta E}/Z$. To be precise, these are distributions on the flat torus $\T^d=\R^d/(2\pi\Z)^d$, i.e. the hypercube $[0,2\pi]^d$ with opposite faces identified. Here the normalizing constant $Z$ is the partition function, and $\beta>0$ is the inverse temperature:
\begin{equation}
Z= \int_{\T^d} e^{-\beta E(x)} dx := \int_{[0,2\pi)^d} e^{-\beta E(x)} dx,
\end{equation}
and, more generally, $p(A) = \int_A p(x) dx$ for every (measurable) set $A \subseteq [0, 2\pi)^d$. We blur the distinction between periodic functions on the Euclidean domain and their quotient counterparts on the compact torus, and likewise between sets $A \subseteq [0, 2\pi)^d$ and their counterparts in $\T^d$; a more formal notation is followed in \cref{app:lb-sec} for the interested reader. Throughout, the difference of two distributions $p$ and $q$ is measured in total variation distance,
\begin{equation}
\label{eq:tv-def}
    \mathrm{TV}(p,q)=\sup_{A\subseteq\T^d}\vert p(A)-q(A)\vert.
\end{equation}

\paragraph{Instance class}

To prove a separation result, a classical complexity lower bound will be derived for a class of Gibbs states $\mathscr{P}$ with certain type of smoothness. We require such smoothness to be strong enough for the quantum algorithm to succeed and weak enough for hiding probability mass to enable the lower bound result (see \cref{fig:hero-plot}). For a multi-index $a\in\N^d$ write $a!\equiv\prod_j a_j!$, $\partial^{a}\equiv\prod_j(\partial/\partial x_j)^{a_j}$, $\onenorm{a}=\sum_j a_j$.

\begin{defn}[Nonzero-order Gevrey class]
\label{def:gevrey-seminorm}
For $s \geq 1$, $\rho>0$ and a $2\pi$-periodic smooth function $f\in\mathscr{C}^{\infty}(\R^d)$, define
\begin{equation}
\label{eq:seminorm}
[f]_{s,\rho}\coloneqq\sup_{a\in\N^d\setminus\{0\}}\ \sup_{x\in{[0,2\pi)^d}}\frac{\vert\partial^{a} f(x)\vert\,\rho^{\onenorm{a}}}{(a!)^{s}},
\end{equation}
and write $f\in\dot{\mathscr{G}}^{s}(\xi,\rho)$ if $[f]_{s,\rho}\leq{\xi}$, that is every derivative obeys $\vert\partial^a f\vert\leq{\xi}\,(a!)^{s}\,\rho^{-\onenorm{a}}$.
\end{defn}

At $s = 1$ this derivative growth is that of \emph{real-analytic} functions, whose derivatives at a single point determine them throughout the connected torus by analytic continuation. In this case, a function constant on a nonempty open set is constant everywhere: no probability mass can be hidden in an otherwise flat landscape, and one order-$\infty$ query determines the density globally. The regime $s > 1$ relaxes the regularity and allows nonzero compactly supported functions, such as the classical mollifier $e^{-(1-x^2)^{-1/(s-1)}}$ on $(-1,1)$, extended by zero outside. Thus the upper-bound analysis includes $s = 1$, whereas our bump-based classical lower bound and separation require $s > 1$. Hiding probability density in a small area then comes at a price in the radius parameter $\rho$, which we will tune.

Unlike the usual definition of the Gevrey criteria, this variant ignores the order-zero term, matching the gauge freedom of Gibbs potentials, $E\mapsto E+\mathrm{const}$. The conventional Gevrey class $\mathscr{G}^{s}(C,r)$, which constrains the value as well, is recalled in \cref{def:gevrey-function}. Here $C$ and $r$ are counterparts to $\xi$ and $\rho$ in \cref{def:gevrey-seminorm}. The parameter $C$ bounds the function on the real domain, and $r$ controls derivative growth. For $s=1$, the derivative bounds guarantee Taylor convergence for complex increments of modulus less than $r$ in each coordinate, but the same bound $C$ need not hold there. In \cref{cor:class-reg-equiv} we show that \cref{def:gevrey-seminorm} is equivalent to Gevrey regularity of the gradient $\nabla f$ with slightly different Gevrey parameters.

The other parameter defining the instance class bounds the \emph{barrier amplitude}, the worst-case ratio of density values. For a potential $E$, write $\Delta\coloneqq\max E-\min E$ and require
\begin{equation}
e^{\beta\,\Delta}\leq\alpha.
\end{equation}
The class thus has barrier amplitude at most $\alpha$; the hard family below attains equality.

\begin{defn}[$s$-Gevrey potential class]
\label{def:potential-class}
For $s \geq 1$, $\alpha\geq1$ and $\xi,\rho>0$, the \emph{$s$-Gevrey potential class} is the family of Gibbs states
\begin{equation}
\label{eq:potential-class}
\mathscr{W}^{s}(\alpha,\xi,\rho)\coloneqq\left\{p=e^{-\beta E}/Z \bigm|  [\beta E]_{s,\rho}\leq \xi\ \text{and}\ \beta\Delta\leq \log\alpha \right\}.
\end{equation}
\end{defn}

The class grows with $\alpha$ and $\xi$ and shrinks with $\rho$: $\mathscr{W}^{s}(\alpha',{\xi'},{\rho'})\subseteq\mathscr{W}^{s}(\alpha,\xi,\rho)$ for $\alpha\geq\alpha'$, $\xi\geq\xi'$, $\rho\leq\rho'$. The smoothness bound also imposes a ceiling on $\beta\Delta$, not on a possibly loose upper bound $\log\alpha$: taking $a=e_j$ to be the $j$-th one-hot multi-index in \eqref{eq:seminorm} bounds the gradient, $\sup_j\Vert\partial_j(\beta E)\Vert_\infty\leq{\xi/\rho}$, and since two points of $\T^d$ are at $\ell^1$-distance at most $\pi d$, the fundamental theorem of calculus forces for every member:
\begin{equation}
\label{eq:osc-ceiling}
\beta \Delta \leq \min\{\log\alpha,\pi d\,\xi/\rho\}.
\end{equation}
For $0\leq\beta'\leq\beta$, the same potential at inverse temperature $\beta'$ remains in $\mathscr W^s(\alpha,\xi,\rho)$, since both the seminorm and the range of $\beta E$ scale by $\beta'/\beta$.

\paragraph{The quantum upper bound}

In \cref{sec:quantum-gibbs-sampling} we show that quantum singular value thresholding (QSVTh) with temperature annealing applied to the square root of the Witten Laplacian, with access to the coherent gradient oracle,
\begin{equation}
\label{eq:gradient-oracle}
O_{\nabla E}: \ket{x}\ket{b}\mapsto \ket{x}\ket{b\oplus \nabla E(x)},
\end{equation}
and its inverse, requires
\begin{equation}
\label{eq:quantum-ub}
\Nq \in \tilde{\mathcal O}\left(\sqrt{\alpha d} \left(1 + \frac{d + \xi}{\rho}
\left[ds + \log \frac{\alpha(1 + \xi/\rho)}{\epsilon}\right]^s\right) \log \frac1\epsilon\right)
\end{equation}
queries to the gradient oracle \eqref{eq:gradient-oracle} to sample from any Gibbs state $p\in\mathscr{W}^{s}(\alpha,\xi,\rho)$ at accuracy $\eps<1/2$ in total variation distance. We call such a quantum algorithm a first-order quantum algorithm (for its access to first-order derivatives). For fixed $s \geq 1$, the implicit constants in \eqref{eq:quantum-ub} are independent of $d, \alpha, \xi, \rho$ and $\epsilon$; the tilde suppresses only polylogarithmic factors in these parameters.

The bound combines the cost of QSVTh per annealing step (\cref{thm:warm-start-query-complexity}) with the length and success probability of the annealing schedule of \cref{alg:cap} (\cref{lem:success-probability}), including restarts from the uniform state on failure, consolidated into the total query complexity \eqref{eq:quantum-ub} in \cref{thm:main-theorem-quantum}. Membership in the $s$-Gevrey potential class $\mathscr{W}^{s}(\alpha,\xi,\rho)$ guarantees that the density function satisfies $p\in\mathscr G^{s}(\alpha, \rho/(d+\xi+1))$ and $\sqrt{p}\in\mathscr G^{s}(\sqrt\alpha, \rho/(d+\xi/2+1))$ (\cref{cor:class-reg}, \cref{app:gevrey}), both necessary for the filtering step of the quantum sampler.

\paragraph{The classical lower bound}

A classical randomized algorithm for sampling accesses the potential through an \emph{oracle of order $k$}, $k\in\N \cup \{\infty\}$: queried at a point $x\in\T^d$, it answers with the value and the derivatives, up to order $k$, of the log-density at $x$,
\begin{equation}
\label{eq:order-k-answer}
    \phi_k(p,x) = \bigl(\log p(x), \nabla\log p(x), \dots, \nabla^{\otimes k}\log p(x)\bigr),
\end{equation}
that is with $-\beta E(x)-\log Z$ and its derivatives (\cref{def:query-model}). The case $k=0$ is the zeroth-order (evaluation) oracle and $k=1$ the first-order (gradient) oracle of the sampling and optimization literature~\cite{chewi2023fisher}; the first-order entry $\nabla\log p$ is the \emph{score} driving Langevin dynamics and diffusion models in the context of machine learning.

A randomized algorithm $\mathcal{A}$ adaptively chooses queries and obtains oracle answers, using a random seed $\omega$ drawn independently of the instance, and after $n$ queries outputs a point $\hat X_n\in\T^d$, whose distribution over the seed is the \emph{output law} $\hat{p}_n$~(\cref{def:sampling-algorithm}). The first query depends on the seed alone: instance-dependent initialization, such as a warm start, is excluded. We use the number $N$ of queries to an oracle for the target such algorithm requires to guarantee an output law $X\sim\hat{p}_N$ that is $\eps$-accurate in total variation as the cost of the algorithm.

\begin{defn}[Classical query complexity]
\label{def:classical-query-complexity}
The query count of a classical randomized algorithm $\mathcal{A}$ on an instance class $\mathscr{P}$ at accuracy $\eps > 0$ is
\begin{equation}
\label{eq:query-cost}
    N_{\mathcal{A}}(\mathscr{P},\eps) = \inf \bigl\lbrace N \in \N \,\big\vert\, \sup_{p \in \mathscr{P}}\,\mathrm{TV}(\hat{p}_N,p)<\eps \bigr\rbrace,
\end{equation}
and the classical query complexity of sampling $\mathscr{P}$ is
\begin{equation}
\label{eq:query-complexity-def}
    \Ncl(\mathscr{P},\eps) = \inf_{\mathcal{A}}\,N_{\mathcal{A}}(\mathscr{P},\eps),\qquad \inf\varnothing=+\infty.
\end{equation}
\end{defn}
\noindent The query count \eqref{eq:query-cost} is a worst-case guarantee for any classical algorithm to provide $\eps$-accurate samples on every instance in the class $\mathscr P$. This definition matches \cite[Definition~5]{chewi2023fisher} with total variation distance in place of the Fisher information.

The bound derived in~\cref{thm:hide-and-seek-lower-bound} below is information-theoretic and valid for any adaptive sampler and for queries of every order, such as value, gradient, even the full tower of derivatives. In fact, it holds for any oracle whose answer is \emph{local}, i.e. at query point $x$ it is determined by the density on an arbitrarily small neighbourhood of $x$ (see~\cref{rem:locally-determined}). Consequently, \cref{thm:classical-lb} shows that on a suitably chosen subclass of \eqref{eq:potential-class} every classical algorithm performs a number of queries linear in the barrier amplitude $\alpha$. 

\paragraph{The main result}

The above-mentioned quantum upper bound and classical lower bound combine to infer our main result as stated below.

\begin{thm}[Separation]
\label{thm:separation}
\label{thm:main-result}
Let $s>1$, $d>2s$, fix an accuracy cap $\eps_0<\tfrac12$, and let the barrier amplitude satisfy $\alpha\geq1+4/(\tfrac12-\eps_0)$. On the class of Gibbs states $\mathscr{W}^{s}_{\alpha,d}\coloneqq\mathscr{W}^{s}(\alpha,{\xi_{\alpha,d}},{\rho_{\alpha,d}})$ of \eqref{eq:potential-class}, with
\begin{equation}
\label{eq:main-constants}
{\xi_{\alpha,d}}=\frac{1}{2},\qquad \frac{1}{{\rho_{\alpha,d}}}=\Theta\bigl(d^{1+s}\,\alpha^{1/d}\,\log^{s}\alpha\bigr),
\end{equation}
and with implicit constants depending only on $s$, the following holds for every accuracy $0<\eps\leq\eps_0$. The classical query complexity of sampling from $\mathscr{P}=\mathscr{W}^{s}_{\alpha,d}$, with an oracle of any order $k\in\N\cup\{\infty\}$ obeys
\begin{equation}
\label{eq:main-lb}
\Ncl \geq \frac{1-2\eps_0}{16}\,(\alpha-1),
\end{equation}
while there exists a first-order quantum algorithm with query complexity \eqref{eq:quantum-ub} at $\xi=\xi_{\alpha,d}$, $\rho=\rho_{\alpha,d}$. The ratio of the query complexities is
\begin{equation}
\label{eq:separation-asymptotics}
    \frac{\Ncl}{\Nq}=\tilde{\Omega}\left(\frac{\alpha^{c(d)}}{d^{\frac52+2s}\,
    \log^{s+1}\tfrac{1}\eps}\right),
\qquad
c(d)=\frac12-\frac1d,
\end{equation}
with implicit constants depending only on $s$ and $\eps_0$.
\end{thm}

\begin{rmk}
\label{rmk:exponential-separation}
In particular, in the regime $\log \alpha =\Omega(d)$, $\log \log \frac{1}\eps=o(d)$, reached at low temperatures, say when $\beta=\Omega(d)$, we have
\begin{equation}
\label{eq:separation-regime}
    \Ncl=e^{\Omega(d)}\, \Nq.
\end{equation}
\end{rmk}

\section{Upper Bound for Quantum Gibbs Sampling on a Torus}
\label{sec:quantum-gibbs-sampling}

In this section we develop an end-to-end quantum algorithm for Gibbs sampling on the class of potentials $\mathscr{W}^s(\alpha,\xi,\rho)$ defined in \cref{def:potential-class}. Sampling from $p(x)$ can be recast as finding a steady-state solution of the time-dependent, parabolic Fokker--Planck equation
\begin{equation}\label{eq:fp-equation}
\partial_t f(x,t) = \mathcal{L}(f(x,t)),
\end{equation}
where the Fokker--Planck operator $\mathcal{L}$ is defined by
\begin{equation}\label{eq:fp-operator}
\mathcal{L}(f(x))\coloneqq\beta^{-1}\nabla \cdot \left(e^{-\beta E}\nabla\left(e^{\beta E}f(x)\right)\right)=\nabla \cdot (\nabla E(x) f(x))+\beta^{-1}\nabla^2 f(x).
\end{equation}
equivalently, this steady state is the solution of the time-independent, elliptic Fokker--Planck equation
\begin{equation}\label{eq:fp-steady}
\mathcal{L}(f(x)) = 0.
\end{equation}
It is a classical fact that $p(x)$ is the unique steady state of the Fokker--Planck operator \cite{pavliotis2014stochastic}; \cite{motamedi2024gibbs} establishes the analogous statement for periodic potentials on a torus, the setting relevant here, while \cite{leng2026operator} uses the same fact in the Euclidean setting.

We use \cite{leng2026operator}'s quantum singular value thresholding (QSVTh) method for solving the time-independent equation \eqref{eq:fp-steady} to achieve a rigorous quadratic speedup in Gibbs sampling. \cite{motamedi2024gibbs} faces two obstructions to solving \eqref{eq:fp-steady} directly, and instead solves the time-dependent equation \eqref{eq:fp-equation} by forward time stepping until convergence to the Gibbs state. First, \eqref{eq:fp-steady} is homogeneous: one seeks a vector in the kernel of the discretized operator rather than the solution to a linear system with a nonzero right-hand side, and quantum linear system algorithms (QLSA) are not designed for this kernel-finding task. Second, the discretized Fokker--Planck operator has a large condition number, on the order of the Poincar\'e constant $C_{\mathrm{PI}}$ (the inverse spectral gap of the generator) itself. For these reasons \cite{motamedi2024gibbs} does not achieve a general, operator-level quadratic speedup, i.e., a complexity scaling as $\sqrt{C_{\mathrm{PI}}}$; such a speedup is obtained there only for Morse potentials, for which a favorable Poincar\'e constant is available.

\cite{leng2026operator} overcomes both of these obstructions in solving \eqref{eq:fp-steady}. First, it introduces the QSVTh method, which targets the kernel directly via quantum singular value thresholding rather than solving a linear system. Second, rather than discretizing $\mathcal{L}$, it exploits the square-root decomposition of the Witten Laplacian $\mathcal{H}:=-e^{\beta E/2}\mathcal{L}e^{-\beta E/2}$ as $\mathcal{H} = \sum^d_{j=1}L^\dagger_j L_j,$ where $L_{j}e^{-\beta E/2} = 0$ for all $j\in[1,\ldots,d]$.\footnote{In the following we use calligraphic font to denote continuous operators (such as $\mathcal{H}$), and reserve blackboard bold (such as $\mathbb{H}_N$, with subscript $N$ for the grid size), for the discretized ones.} For an integer Fourier cutoff $N\geq1$, we keep the state on the $N$-grid and evaluate each factor $L_j$ on a $2N$-grid after Fourier zero-padding, as defined in \eqref{eq:oversampled-factors}. This gives the stacked operator\footnote{The corresponding stack without this finer-grid modification is denoted $\mathbb{L}_N$ in \cite{leng2026operator} (Eqn.~44).}
\begin{equation}
    \mathbb{W}_N = [\mathbb{L}_{1,N}^\top,\dots,\mathbb{L}_{d,N}^\top]^\top,
\end{equation}
whose least right singular vector approximates the discretized state $\ket{\sqrt{p}_N}$. QSVTh is then used to prepare this vector \cite{leng2026operator}. Applying QSVTh to these components rather than to the Witten Laplacian directly gives a square-root dependence on the inverse gap, yielding the desired operator-level quadratic speedup.

However, \cite{leng2026operator} makes several simplifying assumptions that leave it short of an end-to-end, quadratic-speedup Gibbs sampler. We enumerate these assumptions together with how we remove or prove each of them.
\begin{enumerate}[itemsep=2pt, topsep=2pt, parsep=0pt]
\item \emph{Warm start.} The algorithm requires access to an initial state with constant overlap with the target Gibbs state. We remove this requirement using a temperature-annealing procedure (\cref{sec:annealing}; full statements and proofs are given in \cref{app:qsvt-anneal}).
\item \emph{Dissipativity.} \cite{leng2026operator} works on Euclidean space and assumes the potential is dissipative, e.g., grows at least quadratically away from the origin, so as to truncate the domain to a hypercube of side length $O(\log(d/\epsilon))$ around the origin at target precision $\epsilon$ (\cite[Section~D.1, footnote~7]{leng2026operator}). We instead work on the torus $\mathbb{T}^d$ with periodic boundary conditions and impose $s$-Gevrey smoothness on the potential (\cref{def:potential-class}), which controls the discretization error without any dissipativity assumption.
\item \emph{Ground-state proximity.} \cite{leng2026operator} assumes, rather than proves (their Assumption~11, item~1), that the ground state of the discretized Witten Laplacian is close to the discretized square-root Gibbs state. For our finer-grid factors, we prove this by an energy estimate in \cref{thm:overlap-guarantee}.
\item \emph{Spectral-gap tracking.} \cite{leng2026operator} similarly assumes (their Assumption~11, item~2) that the spectral gap of the discretized Witten Laplacian tracks that of the continuous operator. We prove the lower bound needed for filtering in \cref{lem:least-eigenvalues-hn}, using the Galerkin comparison of \cref{lem:galerkin-comparison}.
\end{enumerate}
In the rest of this section we present the components that fill the gaps identified above.

\subsection{Gibbs state preparation from a warm start}
\label{sec:warmstart}

We use QSVTh to project onto the ground states of $\mathbb{W}_N$ (its smallest-singular-value subspace) and prove that this subspace approximates the target state $\ket{\sqrt{p}_N}$, without assuming in advance that $\mathbb{W}_N$ has a one-dimensional kernel. Applying QSVTh to $\mathbb{W}_N$ requires a threshold lying strictly between the least and second-least singular value of $\mathbb{W}_N$, so that the filter amplifies the ground states and damps everything above it.

Consider the spectrum of $\mathbb{H}_N := \mathbb{W}_N^\dagger \mathbb{W}_N$, whose eigenvalues are the squared singular values of $\mathbb{W}_N$. We prove that its ground state is unique and approximates the grid Gibbs state in Euclidean norm when $N\in\polylog(1/\epsilon)$, with the other parameters fixed. We compare $\mathbb H_N$ with $\mathbb G_N=I_N^*\mathcal H I_N$, the Galerkin restriction to the Fourier modes represented by the $N$-grid. The continuous gap and the min--max principle bound its second eigenvalue, while a Fourier approximation to $\sqrt p$ gives a low-energy trial state. The finer-grid comparison then supplies the discrete gap and ground-state accuracy needed by QSVTh (\cref{sec:galerkin-comparison}). This allows us to obtain the following block-encoding query complexity.

\begin{lem}[\cref{cor:threshold-choice}]
\label{cor:gbgs-query-complexity}
Suppose $0<\tilde\epsilon<1/2$ is given. Provided $p$ belongs to the class $\mathscr{W}^{s}(\alpha,\xi,\rho)$ as defined in \cref{def:potential-class} and
\begin{align}
N \in \tilde\Omega\left(\frac{d + \xi}{\rho} \left[ds + \log \frac{\alpha(1 + \xi/\rho)}{\tilde\epsilon}\right]^s\right),
\end{align}
given a constant-overlap warm start, there is an algorithm preparing a state $\ket{\tilde{\sqrt{p}}}$ satisfying $\norm{\ket{\tilde{\sqrt{p}}}-\ket{\sqrt{p}_{N}}}_2<\tilde\epsilon$ on heralded success, with constant success probability, using
\begin{equation}
\mathcal O\left(\sqrt{\alpha d}(\pi N+\xi/2\rho)\log\left(\frac{1}{\tilde\epsilon}\right)\right)
\end{equation}
queries to the block encoding of $\mathbb{W}_N$.
\end{lem}

We denote the algorithm which prepares the approximate Gibbs state in \cref{cor:gbgs-query-complexity} as Grid-based Gibbs sampling (GbGS). For Gibbs distributions in~\cref{def:potential-class}, using \cite{ronagh2026breaking}, we only require that the prepared quantum state approximate the discretized Gibbs state well for successful sampling. We use the fact that the Gibbs distribution belongs to the $s$-Gevrey class to ensure that $N$ is polylogarithmic in inverse precision \cite{ronagh2026breaking}. The pseudocode for GbGS is provided in \cref{alg:GbGS}.

\begin{algorithm}[t]
\caption{Grid-based Gibbs sampling, GbGS($\beta$)}
\label{alg:GbGS}
\begin{algorithmic}[1]
\Require
Inverse temperature $\beta>0$; Target precision $\tilde\epsilon>0$; \newline
\hspace*{6mm} Block encoding of $\mathbb{W}_N$ with normalization $\pi N \sqrt{d/\beta}\le \gamma \le \sqrt{d/\beta}(\pi N+\xi/2\rho)$;
\newline
\hspace*{6mm} Warm-start state $\ket{\psi}$ on a fixed $N$-grid such that $|\langle\psi|\sqrt{p}_N\rangle|\in\Omega(1)$, for Gibbs state $p$ as per \cref{def:potential-class}.
\State Use a grid satisfying \cref{cor:gbgs-query-complexity}.
\State Apply QSVTh (\cref{prop:sv-filter}) approximating the following filter function:
\begin{align}
    f(x) = \begin{cases}
        1, & x \in \left[-\sqrt{\frac{1}{2\alpha\gamma^{2}\beta}},\sqrt{\frac{1}{2\alpha\gamma^{2}\beta}}\right],\\
        0, & x \in \left[-1, -\sqrt{\frac{3}{4\alpha\gamma^{2}\beta}}\right]\cup \left[\sqrt{\frac{3}{4\alpha\gamma^{2}\beta}}, 1\right].
    \end{cases}
\end{align}
\State Measure the filter's success flag to report success or failure.
\Ensure Normalized state $\ket{\tilde{\sqrt{p}}}$ satisfying $\norm{\ket{\tilde{\sqrt{p}}} - \ket{\sqrt{p}_{N}}}_{2}\le\tilde\epsilon$, on success.
\end{algorithmic}
\end{algorithm}

\begin{lem}[\cref{cor:gevrey-sampling}]
\label{lem:gevrey-sampling-n-choice}
Suppose $s \geq 1$ and the output $\ket{\tilde{\sqrt{p}}}$ of GbGS (\cref{alg:GbGS}) satisfies $\norm{\ket{\tilde{\sqrt{p}}} - \ket{\sqrt{p}_N}}_2 \leq \tilde\epsilon/4$ for some $0 < \tilde\epsilon < 1/2$. Then, provided
\begin{equation}
\label{eq:gevrey-sampling-n-choice}
N \in \tilde\Omega\left(\frac{d + \xi}{\rho} \left[ds + \log \frac{\alpha}{\tilde\epsilon}\right]^s\right),
\end{equation}
there is a quantum algorithm that returns a sample from a random variable $X\sim \eta$ such that $\rm{TV}(\eta, p) \le \tilde\epsilon$ using one copy of the state $\ket{\tilde{\sqrt{p}}}$, and an additional $d\cdot \polylog(dN/\tilde\epsilon)$ elementary gates.
\end{lem}

\subsection{End-to-end Gibbs sampling without a warm start}
\label{sec:annealing}

We now remove the warm-start requirement of GbGS (\cref{alg:GbGS}) using temperature annealing, building on \cite{somma2007quantum,ozgul2024stochastic}. On one common grid, we start from the uniform superposition at $\beta_0=0$ and use $n_{\mathrm{ann}}=\max\{1,\lceil(\log\alpha)^2/4\rceil\}$ steps of size $\delta=\beta/n_{\mathrm{ann}}$. This uses the class parameter $\alpha$ and does not require knowledge of the actual range $\Delta$. At each step we apply the GbGS filter at $\beta_k=k\delta$, with ground-state and filter errors at most $\tilde\epsilon=\epsilon/(64 n_{\mathrm{ann}})$. \cref{lem:success-probability} proves a finite success probability greater than $1/2$ and a constant overlap at every step, including weak barriers. \cref{alg:cap} shows one successful attempt.

\cref{cor:gbgs-query-complexity} determines the choice of $N$ guaranteeing desired precision at each annealing step, and choosing $N$ according to \cref{lem:gevrey-sampling-n-choice} ensures that continuous sampling after the last step approximates the final distribution well. Using these lemmas and restarting \cref{alg:cap} from the uniform state on failure, for at most $\lceil\log_2(2/\epsilon)\rceil$ attempts, gives the following overall complexity:

\begin{thm}
\label{thm:main-theorem-quantum}
Let $E:\mathbb{T}^{d}\rightarrow\mathbb{R}$ be a Gibbs potential in $\mathscr{W}^s(\alpha,\xi,\rho)$ as defined in \cref{def:potential-class} and $p(x)= e^{-\beta E(x)}/Z$ the corresponding Gibbs state. There exists a quantum algorithm that returns a sample from a random variable $X\sim\eta$ after
\begin{align}
\label{eq:main-query-complexity}
N_q \in \tilde{\mathcal O}\left(\sqrt{\alpha d} \left(1 + \frac{d + \xi}{\rho}
\left[ds + \log \frac{\alpha(1 + \xi/\rho)}{\epsilon}\right]^s\right) \log \frac1\epsilon\right)
\end{align}
queries to the quantum gradient oracle $O_{\nabla E}$ and its inverse $O_{\nabla E}^{\dagger}$, such that $\mathrm{TV}(\eta, p) \le \eps$.
\end{thm}
\noindent The proof is given in \cref{thm:main-theorem-quantum-app}. The $\sqrt{\alpha}$ scaling in \eqref{eq:main-query-complexity} is crucial for establishing the quantum--classical separation in the results that follow.

\begin{algorithm}[t]
\caption{End-to-end Gibbs sampling}
\label{alg:cap}
\begin{algorithmic}[1]
\Require Inverse temperature $\beta>0$; Target precision $0<\epsilon<1/2$; \newline
\hspace*{6mm} Class parameters $s,\alpha,\xi,\rho$ from \cref{def:potential-class}.
\State Set $n_{\mathrm{ann}}=\max\{1,\lceil(\log\alpha)^2/4\rceil\}$ and $\tilde\epsilon=\epsilon/(64 n_{\mathrm{ann}})$.
\State Choose $N$ satisfying \cref{cor:gbgs-query-complexity,lem:gevrey-sampling-n-choice} at precisions $\tilde\epsilon$ and $\epsilon/2$, respectively.
\State Initialize $\ket{\tilde{\sqrt p}}=(2N+1)^{-d/2}\sum_x\ket x$.
\For{$k=1,\dots,n_{\mathrm{ann}}$}
    \State Set $\beta_k=k\beta/n_{\mathrm{ann}}$.
    \State Apply the GbGS filter with ground-state and filter errors at most $\tilde\epsilon$.
\EndFor
\State Upsample and generate a continuous sample using \cref{lem:resolution-booster} at precision $\epsilon/2$.
\Ensure On success, a classical sample $X\sim\eta$ with $\mathrm{TV}(\eta,p)\leq\epsilon$, where $p(x)\propto e^{-\beta E(x)}$.
\end{algorithmic}
\end{algorithm}

\section{Lower Bound for Classical Gibbs Sampling on a Torus}
\label{sec:lower-bound}

In this section, we prove query-complexity lower bounds for sampling smooth Gibbs distributions on the torus $\T^d$. We establish that any classical algorithm sampling a natural class~(\cref{def:potential-class}) of targets $p\propto e^{-\beta E}$ must incur a number of queries growing linearly in the barrier amplitude $\alpha=e^{\beta\Delta}$ of that class (\cref{cor:lower-bounds-gevrey}). The algorithm learns about the target only through finite-order queries as in \eqref{eq:order-k-answer} and its cost is the number of queries~\cref{def:classical-query-complexity} needed to produce one sample $\eps$-close to $p$ in total variation distance, on every instance of the class.

\subsection{The adversarial \emph{hide-and-seek} argument}
\label{ssec:hide-seek}

The hard instances are \emph{bump functions}: distributions that are flat outside a small patch and concentrate an anomalous probability mass inside it. A patch $\Om\subset\T^d$ is a closed set. A bump on $\Om$ with concentration $\delta\in(0,1)$ is a distribution with density $p$ such that
\begin{equation}
\label{eq:offpatch-density}
    p(\Om) =\delta >\frac{\vol(\Om)}{(2\pi)^d} \quad \mathrm{and} \quad p(x) = \kappa,\quad \forall\,x \in \T^d \setminus \Om,
\end{equation}
\noindent
with $\kappa=(1-\delta)/\bigl((2\pi)^d-\vol(\Om)\bigr)$  (\cref{def:bump}). For a query point $x \in \T^d \setminus \Om$, the log-density is constant $\log p =\log\kappa$ on an open neighbourhood of $x$, hence all its derivatives in $x$ vanish. Therefore, the oracle \eqref{eq:order-k-answer} returns the \emph{flat answer}
\begin{equation}
\label{eq:null-response}
    \phi_{\mathrm{null}} = (\log\kappa, 0, 0, \dots),
\end{equation}
the same for every bump with background value $\kappa$.

Consider the \emph{null run} of an algorithm $\mathcal{A}$, that is the process obtained by feeding $\mathcal{A}$ the flat answer~\eqref{eq:null-response} for each query corresponding to a random seed $\omega$ independent of the instance. On any bump instance, any algorithm with access model~\eqref{eq:order-k-answer} gets exactly such answers until one of its queries lands in the patch, revealing its position. Therefore, a run of $\mathcal{A}$ coincides with the null run up until the \emph{first hit} (\cref{lem:null-coupling}). The requirement of accurate sampling force the hard instances to have mass $\approx\delta$ on the patch, which is essentially equilvalent to \emph{finding} the patch, mapping accurate sampling to a search problem~\cite{chewi2023fisher}. Adopting the adversarial strategy to hide the patch at a uniformly random position provides a quantitative lower bound on the query cost:
\begin{thm}[Hide-and-seek lower bound on $\T^d$]
\label{thm:hide-and-seek-lower-bound}
Let $p_0$ be a bump \eqref{eq:offpatch-density} on a closed patch $\Om\subset\T^d$ of volume $v_0\in(0,(2\pi)^d)$ and concentration $\delta$. Let $\mathscr{P}$ be any instance class accessed by the oracle and containing all shifts $p_g(A)\coloneqq p_0(A-g)$, $g\in\T^d$. For every oracle order $k\in\N\cup\{\infty\}$ and every accuracy $\eps<\delta$, the query complexity of \cref{def:classical-query-complexity} satisfies
\begin{equation}
\label{eq:compact-lower-bound}
    \Ncl(\mathscr{P},\eps) \geq (2\pi)^d\,\frac{\delta-\eps}{v_0}-1.
\end{equation}
\end{thm}

Read $M\coloneqq(2\pi)^d/v_0$ as the number of hiding places; so, the bound is $\Ncl\geq M(\delta-\eps)-1$. In the following we outline the proof, leaving its detailed presentation for \cref{prf:hide-and-seek-lower-bound}. For each shift $g \in \T^d$, with patch $\Omega_g=g+\Om$ and $\tau_g$ the first time the null run queries inside $\Omega_g$, accuracy requirements give
\begin{equation}
\label{eq:sketch-decomposition}
    \eps \geq \delta-\hat{p}_{g,N}(\Omega_g) \geq \delta-\sum_{n=1}^{N}\mathbb{P}\bigl(X^{\mathrm{null}}_n\in\Omega_g\bigr)-\mathbb{P}\bigl(\hat X^{\mathrm{null}}_N\in\Omega_g\bigr).
\end{equation}
The two terms on the right-hand side are a proxy on the probability $\hat{p}_{g,N}(\Omega_g)$ of the output sitting on the hidden patch: either some query found it (the sum term: the null run hits the patch), or by blind luck (the second term: the entire run is null, yet the output hits the patch). Now, we average over a random shift $g \sim \mathrm{Unif}(\T^d)$. Since the null run does not depend on $g$, and any $x \in \T^d$ is covered by a fraction $v_0/(2\pi)^d$ of the patches, by
\begin{equation}
\label{eq:coverage-identity}
    \vol\bigl(\{g\,\vert\,x\in g+\Om\}\bigr)=\vol(x-\Om)=v_0,
\end{equation}
each of the $N+1$ points $X^{\mathrm{null}}_1,\dots,X^{\mathrm{null}}_N,\hat X^{\mathrm{null}}_N$ is covered with probability $v_0/(2\pi)^d$. Accordingly the average of \eqref{eq:sketch-decomposition} evaluates to $\eps\geq\delta-(N+1)\,v_0/(2\pi)^d$, which is \eqref{eq:compact-lower-bound}.

We note that our construction differs from arguments based on packing $\T^d$ with patches in the style of~\cite{chewi2024query,chewi2023fisher}. Consider a grid of cubic cells of side $\ell$ that tile the torus, hence in number $M=(2\pi/\ell)^d \in \mathbb{N}$: a finite average over the $M$ disjoint bumps located each on a cell gives the same bound as \eqref{eq:compact-lower-bound}. Our choice of a uniformly random shift removes the integrality constraint on $2\pi/\ell$ allowing to choose this knob continuously. It is also the best the adversary can do: however a patch of volume $v_0$ is placed at random, the average over $x$ of the probability that $x$ is covered equals $v_0/(2\pi)^d$, so some point is covered with probability at least $v_0/(2\pi)^d$, and the choice of uniform random $g$ attains this bound.

\subsection{The hard instances: Gevrey bumps}
\label{ssec:bumps}

It remains to show how to construct bump distributions within the class $\mathscr{W}^{s}(\alpha,\xi,\rho)$. Here the choice $s>1$ is crucial. We propose a construction in three steps which uses as a building block the mollifier $\psi_s(x)=e^{-(1-x^{2})^{-1/(s-1)}}$, supported on $[-1,1]$, whose derivatives obey the $s$-Gevrey bound with a radius specified in \cref{prop:psi-gevrey-property}. First, $\psi_s$ is integrated to give a smoothened step-function. Second, products of steps give a smooth indicator function $\chi_{\nu,\ell}:\R^d\to[0,1]$ of the cell $[0,\ell]^d$: equal to $1$ on the inner plateau $[2\nu,\ell-2\nu]^d$, zero outside $[0,\ell]^d$, monotone across ramps of width $2\nu\leq\ell/2$ (\cref{def:mollified-profile} and \cref{lem:bump-profile}). Third and final, the indicator is shifted by $1$ and normalized, and the resulting density is a two-level profile built on $\chi_{\nu,\ell}$, shown in \cref{fig:bump-profile}: a mesa of height ratio $\alpha$ over a flat background. As $\nu \to 0^+$ the ramp steepens, so hiding the mesa in a smaller cell costs larger derivatives and pushes towards a small Gevrey radius $\rho$.

\begin{figure}[t]
  \begin{minipage}[t]{0.48\textwidth}
        \includegraphics[width=\linewidth]{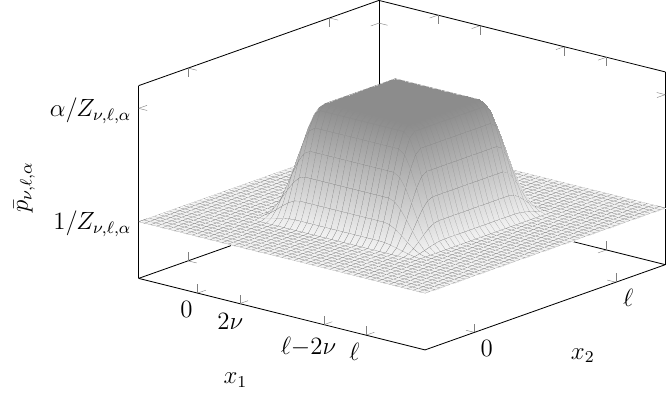}
  \end{minipage}\hfill
  \begin{minipage}[t]{0.48\textwidth}
        \includegraphics[width=\linewidth]{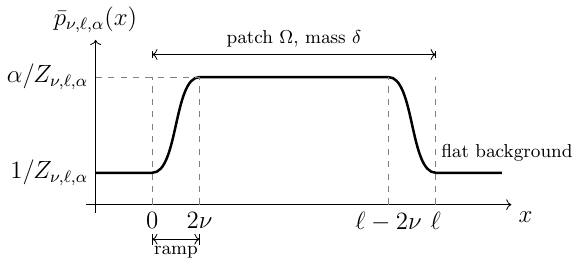}
  \end{minipage}
  \caption{\justifying The bump density $p_{\nu,\ell,\alpha}$ \eqref{eq:bump-torus-density}. Left: the two-dimensional case ($d=2$), a mesa of height $\alpha/Z_{\nu,\ell,\alpha}$ over the inner plateau $[2\nu,\ell-2\nu]^2$, falling to the flat background across ramps of width $2\nu$. Right: one-dimensional section. The barrier amplitude $\alpha$ is the plateau-to-background ratio, and the patch carries probability mass $\delta(\nu,\ell,\alpha)$ \eqref{eq:z-delta-ms}.
  }
\label{fig:bump-profile}
\end{figure}

\begin{defn}[Bump distribution $p_{\nu,\ell,\alpha}$ on $\mathbb{T}^d$]
\label{def:bump-torus}
For $\alpha>1$, $0<\ell<2\pi$, $0<\nu\leq\ell/4$, the bump distribution $p_{\nu,\ell,\alpha}$ is defined via its density
\begin{equation}
\label{eq:bump-torus-density}
p_{\nu,\ell,\alpha}(x) = \frac{1 + (\alpha-1)\,\chi_{\nu,\ell}(x)}{Z_{\nu,\ell,\alpha}},\quad x \in [0,2\pi)^d,\qquad
Z_{\nu,\ell,\alpha}=(2\pi)^d+(\alpha-1)\,\ell^d\,(1-2\varrho)^d,
\end{equation}
where $\varrho\coloneqq\nu/\ell$ and the normalization uses $\int_{[0,2\pi)^d}\chi_{\nu,\ell}=\ell^d(1-2\varrho)^d$ (see \cref{lem:fill}).
\end{defn}
\noindent
In what follows we use two reduced variables: the number of hiding places $M$ and the relative strength of the mesa $S$,
\begin{equation}
\label{eq:s-def}
    M\coloneqq\left(\frac{2\pi}\ell\right)^{d},\qquad S \coloneqq (\alpha-1)(1-2\varrho)^d.
\end{equation}
\begin{prop}[$s$-Gevrey bump family on $\T^d$]
\label{prop:bump-family}
Let $s>1$, $\alpha>1$, $0<\ell<2\pi$, $0<\nu\leq\ell/4$. The distribution $p_{\nu,\ell,\alpha}$ satisfies the following properties:
\begin{enumerate}
\item It is a bump \eqref{eq:offpatch-density} on the closed cell $\Om=[0,\ell]^d\subset\T^d$, with $\kappa=1/Z_{\nu,\ell,\alpha}$ and
\begin{equation}
\label{eq:z-delta-ms}
    Z_{\nu,\ell,\alpha}/(2\pi)^d=\Bigl(1+\frac{S}M\Bigr)\in[1,\alpha],\qquad \delta(\nu,\ell,\alpha)=p_{\nu,\ell,\alpha}(\Om)=\frac{1+S}{M+S}>\frac{\vol(\Om)}{(2\pi)^d};
\end{equation}
\item Its translates $p_g=p_{\nu,\ell,\alpha}(\cdot-g)$, $g\in\T^d$, are Gibbs states with $\mathscr{C}^\infty$ potentials satisfying $\beta\Delta=\log\alpha$.
\end{enumerate}
\end{prop}

The proof is provided in \cref{prf:bump-family}. For $\alpha\geq e$, tracking the Gevrey property of $\psi_s$ through the construction of the profile $p_{\nu,\ell,\alpha}$, via composition with the logarithm that defines the scaled potential, $\beta E_{\nu,\ell,\alpha} = -\log p_{\nu,\ell,\alpha}$ up to an additive constant, gives the class membership $\beta E_{\nu,\ell,\alpha}\in\dot{\mathscr{G}}^{s}(\xi_E,\rho_E)$ for
\begin{equation}
\label{eq:bump-gevrey-constants}
\xi_E =\frac{1}{2},
\qquad \text{and} \qquad
\rho_E(\nu,\ell,\alpha) =\frac{\mathfrak{r}_s\,\nu}{2\,d^{s}\,(\log\alpha)^{s}},
\end{equation}
where $\mathfrak{r}_s\in(0,\tfrac{1}{32}]$ depends only on $s$ \eqref{eq:constants} (\cref{thm:bump-potential-gevrey-seminorm} and \cref{rem:rho-s-bounds}). Note that the radius is linear in the ramp width $\nu$. Combining \eqref{eq:bump-gevrey-constants} with \cref{thm:hide-and-seek-lower-bound} at $v_0=\ell^d$ yields classical lower bounds for the Gibbs class \eqref{eq:potential-class} with the constants $\xi\geq \xi_E$ and $\rho\leq \rho_E$.

\begin{cor}[Classical sampling lower bounds over $\mathscr{W}^{s}(\alpha,\xi,\rho)$]
\label{cor:lower-bounds-gevrey}
Let $s>1$, $\ell<2\pi$, $0<\nu\leq\ell/4$, and $\alpha>e$. The query complexity of sampling the class $\mathscr{W}^{s}(\alpha,\xi,\rho)$, for any $\xi\geq \xi_E$ and $\rho\leq \rho_E(\nu,\ell,\alpha)$, through queries of any order $k\in\N\cup\{\infty\}$ and at any accuracy $\eps<\delta(\nu,\ell,\alpha)$, obeys
\begin{equation}
\label{eq:lower-bound-on-w}
    \Ncl(\mathscr{W}^{s}(\alpha,\xi,\rho),\eps) \geq \left( \tfrac{2\pi}\ell\right)^d\bigl(\delta(\nu,\ell,\alpha)-\eps\bigr)-1.
\end{equation}
\end{cor}

We refer the reader to its proof in \cref{prf:lower-bounds-gevrey}. Note that the hide-and-seek argument here does not use smoothness since the classical lower bound~\eqref{eq:compact-lower-bound} is regularity-agnostic; membership \eqref{eq:potential-class} with parameters $\alpha,\xi,\rho$ only certifies that the hard family is regular enough for the quantum algorithm of \cref{sec:quantum-gibbs-sampling} so that a quantum--classical separation can be established as shown in the next section.

\section{Quantum--Classical Separation}
\label{sec:separation}

The hide-and-seek bound of \cref{sec:lower-bound} was stated with the geometry $(\nu,\ell)$ of the hidden mesa left free (\cref{cor:lower-bounds-gevrey}). A separation is a matter of choosing this geometry carefully. In this section, we first fix the geometry, turning the flexible bound into one linear in the barrier amplitude $\alpha$ at explicit class parameters (\cref{thm:classical-lb}). We then evaluate the quantum upper bound \eqref{eq:quantum-ub} at those parameters to prove \cref{thm:separation}.

\subsection{The classical lower bound at the separation regime}
\label{ssec:main-result}

The first input is that no classical algorithm can beat a query count linear in $\alpha$ on the $s$-Gevrey class of \cref{def:potential-class} at a specific choice of constants.

\begin{thm}[Classical lower bound]
\label{thm:classical-lb}
Let $s>1$, $d>2s$, and fix a constant accuracy level $\eps_0<1/2$. For every barrier amplitude
\begin{equation}
\label{eq:alpha-hyp}
    \alpha \geq 1+\frac{4}{1/2- \eps_0},
\end{equation}
there are class parameters
\begin{equation}
\label{eq:class-constants}
{\xi_{\alpha,d}}=\frac{1}{2},\qquad \frac{1}{{\rho_{\alpha,d}}}=\Theta\left( {d^{1+s}\,\alpha^{1/d}\,\log^{s}\alpha}\right),
\end{equation}
such that, for every $0<\eps\leq\eps_0$, the query complexity (\cref{def:classical-query-complexity}) of sampling the class of Gibbs states $\mathscr{W}^{s}_{\alpha,d}\coloneqq\mathscr{W}^{s}(\alpha,{\xi_{\alpha,d}},{\rho_{\alpha,d}})$, using any classical algorithm with queries of any order $k\in\N\cup\{\infty\}$, obeys
\begin{equation}
\label{eq:ncl-main}
    \Ncl \coloneqq \Ncl(\mathscr{W}^{s}_{\alpha,d},\eps) \geq \frac{1-2\eps_0}{16}\,(\alpha-1).
\end{equation}
The implicit constants in \eqref{eq:class-constants} depend only on $s$.
\end{thm}

\noindent We first provide two remarks about the above statement.

\begin{rmk}[Accuracy regime] The theorem is stated at constant accuracy $\eps\leq\eps_0<\tfrac12$, the standard regime for query separations; within the family of hard instances of \cref{ssec:bumps} the $\alpha$-dependence of \eqref{eq:ncl-main} is the best possible, i.e. the largest one (see the discussion at the end of this subsection).
\end{rmk}

\begin{rmk}[Class strength] A lower bound over a smaller class is a stronger statement; therefore, we stated the theorem at the largest radius the construction certifies, $1/\rho_{\alpha,d}\asymp d^{1+s}\alpha^{1/d}\log^s\alpha$. For any larger class $\mathscr{W}^{s}(\alpha,\xi,\rho)$ with $\xi\geq\tfrac12$, $\rho\leq\rho_{\alpha,d}$, the same bound holds by monotonicity of the complexity with respect to inclusions.
\end{rmk}

We now explain the proof strategy. The bound~\eqref{eq:lower-bound-on-w} leaves freedom in the geometry of the bump: $(\nu,\ell)$, and equivalently $(M,\varrho)$, are not specified by a choice of $(d,\alpha)$ alone. As the quantum query complexity \eqref{eq:quantum-ub} inflates with the inverse of radius $\rho_E\propto \nu$~\eqref{eq:bump-gevrey-constants}, the choice of geometry is meant to maximize lower bounds for a fixed $\rho_E$. In the reduced variables \eqref{eq:s-def} the concentration is $\delta=(1+S)/(M+S)$, and two observations about monotonicity settle the choice.

First, fix $\varrho$: the product $M\delta=M(1+S)/(M+S)$ is increasing in $M$ but capped at $1+S$, and equals $(1+S)/2$ already at $M=S$. Any $M>S$ gains at most a factor $2$, while the radius $\rho\propto\varrho\,M^{-1/d}$ trades a strict increase in the upper bound for a bounded gain in the lower bound. Moreover, at constant accuracy the bound \eqref{eq:lower-bound-on-w} is vacuous for any $M\geq (1+S)/\eps$.

Second, fix $M$: the strength $S=(\alpha-1)(1-2\varrho)^{d}$ is capped at $\alpha-1$ for every $\varrho$. The choice $\varrho={1}/{2d}$ attains $S\geq(\alpha-1)/4$ \eqref{eq:s-star-bracket}; shrinking $\varrho$ further gains at most a factor $4$ in $S$, while the radius $\rho\propto\varrho\,M^{-1/d}$ falls linearly in $\varrho$.
The balanced choice attains both caps up to absolute factors: as many hiding places as the order of mesa strength, $M=\Theta(S)$, and ramps occupying a fraction $\varrho=\Theta(1/d)$ of each cell, just enough to prevent the decay of $S$ with $d$ without wasting radius:
\begin{equation}
\label{eq:parameter-choice}
    \varrho_*=\frac{1}{2d},\qquad
    M_*=S_*\coloneqq(\alpha-1)\Bigl(1-\frac 1d\Bigr)^{d},\qquad
    \ell_*={2\pi}\,M_*^{-1/d},\quad \nu_*=\varrho_*\,\ell_*.
\end{equation}
We can now present the proof of the above theorem.

\begin{proof}[Proof of \cref{thm:classical-lb}]
The proof proceeds in three steps.

\noindent
\emph{Step 1 (admissibility).} Since $d>2s>2$, $\varrho_*<\tfrac14$, so $\nu_*\leq\ell_*/4$. From $\log(1-x)\geq-(2\log2)x$ on $x \in [0,\tfrac{1}{2}]$, we have
\begin{equation}
\label{eq:s-star-bracket}
    \tfrac14(\alpha-1)\leq S_*=e^{d\log(1-1/d)}(\alpha-1)\leq\alpha-1,
\end{equation}
and the hypothesis \eqref{eq:alpha-hyp} gives $S_*\geq \tfrac14(\alpha-1)\geq\tfrac{2}{1-2\eps_0}>1$, hence $\ell_* = 2\pi S_*^{-1/d}<2\pi$.

\noindent
\emph{Step 2 (classical lower bound).} By \eqref{eq:z-delta-ms} at $M_*=S_*$, the concentration is
\begin{equation}
    \delta_* = \frac{1+S_*}{2S_*} = \frac12+\frac{1}{2S_*} > \frac12 > \eps_0 \geq \eps.
\end{equation}
So \cref{cor:lower-bounds-gevrey} applies with $\xi=\xi_E=\tfrac12$ and $\rho=\rho_{\alpha,d}\coloneqq \rho_E(\nu_*,\ell_*,\alpha)$, and a lower bound linear in $\alpha$ holds:
\begin{equation}
\label{eq:ncl-proof}
    \Ncl \geq M_*(\delta_*-\eps)-1 = S_*\Bigl(\frac12-\eps\Bigr)-\frac12 \geq \frac{1-2\eps_0}{2}\,S_*-\frac12 \geq \frac{1-2\eps_0}{4}\,S_* \geq \frac{1-2\eps_0}{16}\,(\alpha-1).
\end{equation}
The second-to-last inequality follows by $S_*\geq2/(1-2\eps_0)$, the last by \eqref{eq:s-star-bracket}. This is claim~\eqref{eq:ncl-main}, and $\eps<\delta_*$ satisfies the accuracy hypothesis of \cref{cor:lower-bounds-gevrey}.

\noindent
\emph{Step 3 (class parameters).} From \eqref{eq:bump-gevrey-constants} and $\nu_*=\tfrac{\pi}d\,S_*^{-1/d}$, we have
\begin{equation}
\label{eq:r-alpha-d}
    \frac{1}{{\rho_{\alpha,d}}} = \frac{2\,d^{s}(\log\alpha)^{s}}{{\mathfrak{r}_s}\,\nu_*} = {\frac{2}{\pi\,\mathfrak{r}_s}}\,d^{1+s}\,S_*^{1/d}\,\log^{s}\!\alpha =\Theta(d^{1+s}\,\alpha^{1/d}\,\log^{s}\!\alpha),
\end{equation}
since $S_*^{1/d}=(\alpha-1)^{1/d}(1-\tfrac 1d)=\Theta(\alpha^{1/d})$ by \eqref{eq:s-star-bracket}. This specifies the class parameters \eqref{eq:class-constants}.
\end{proof}

We note that within the family of bumps \eqref{eq:bump-torus-density} the choice \eqref{eq:parameter-choice} guarantees a lower bound of optimal strength: combining the two caps above, $M\delta\leq 1+S\leq\alpha$ for every admissible $(\nu,\ell)$, so the hide-and-seek bound $M(\delta-\eps)-1$ through Gevrey bumps of barrier amplitude $\alpha$ never exceeds $\alpha$, and \eqref{eq:parameter-choice} attains it up to the absolute constant $(1-2\eps_0)/16$.

\subsection{Proof of the separation theorem}
\label{sssec:optimization-lb}

\begin{proof}[Proof of \cref{thm:separation}]
The classical bound \eqref{eq:main-lb} at the parameters \eqref{eq:main-constants} is \cref{thm:classical-lb} verbatim, with \eqref{eq:class-constants} and \eqref{eq:ncl-main}; the quantum hypothesis is met by the sampler of \cref{sec:quantum-gibbs-sampling}, through \cref{thm:main-theorem-quantum} and the class regularity of \cref{cor:class-reg}. It remains to derive \eqref{eq:separation-asymptotics} and \eqref{eq:separation-regime}.

\noindent
\emph{Step 1 (quantum query complexity).} We substitute $\xi = \tfrac12$ and $\rho = \rho_{\alpha,d}$ into \eqref{eq:quantum-ub}. By \eqref{eq:r-alpha-d}, $\log \tfrac{1}{\rho_{\alpha,d}} \leq \tfrac1d \log \alpha + (1+s)\log d + s\log\log\alpha + \mathcal O(1) = \mathcal O(\log\alpha + \log d)$ for fixed $s$, hence
\begin{equation}
\log \frac{\alpha(1+1/(2\rho_{\alpha,d}))}{\epsilon}
= \mathcal O\left(\log\alpha + \log \tfrac1\epsilon + \log d\right).
\end{equation}
Using \eqref{eq:r-alpha-d} again, we have $(d+\tfrac12)/\rho_{\alpha,d} = \Theta(d^{2+s} \alpha^{1/d} \log^s\alpha)$, and therefore
\begin{equation}
\label{eq:step-4-ub}
\Nq \in \tilde{\mathcal O}\left(\alpha^{\frac12+\frac1d} d^{\frac52+s}
(\log\alpha)^s \left(d + \log\alpha + \log \tfrac1\epsilon\right)^s \log \tfrac1\epsilon\right).
\end{equation}

\noindent
\emph{Step 2 (query complexity ratio and separation regime).} Dividing the classical bound \eqref{eq:ncl-main} by \eqref{eq:step-4-ub}, and bounding the parenthesis by $\mathcal O(d\log(\alpha d)\log(1/\epsilon))$, gives \eqref{eq:separation-asymptotics} with $c(d) = \tfrac12 - \tfrac1d$. For fixed $s$ and $\log\alpha = \Omega(d)$, the numerator contributes $c(d)\log\alpha = \Omega(\log\alpha)$. The logarithm of the denominator, including the suppressed logarithmic factors, is $\mathcal O\left(\log d + \log\log\alpha + \log\log(1/\epsilon)\right)$. If $\log\log(1/\epsilon) = o(d)$, the logarithm of the ratio is therefore bounded below by
\begin{equation}
\label{eq:log-ratio-bound}
\Omega(\log\alpha) - o(\log\alpha) - o(d) - \mathcal O(\log d) = \Omega(d),
\end{equation}
completing the proof.
\end{proof}

\begin{rmk}
The Fourier bound makes the cutoff linear in the inverse Gevrey radius. Thus the factor $\alpha^{1/d}$ in \eqref{eq:class-constants} costs $1/d$ in the separation exponent $c(d)$. The Gevrey order, $s$, enters the polynomial dependence on $d$ and the logarithmic dependence on precision. A barrier amplitude exponential in dimension, $\log\alpha = \Theta(d)$, yields \eqref{eq:separation-regime} at class parameters $\xi_{\alpha,d} = \tfrac12$ and $\rho_{\alpha,d} = \Theta(d^{-(1+2s)})$.
\end{rmk}

\section*{Acknowledgement}
\label{sec:ack}

Authors thank Jiaqi~Leng and Zhiyan~Ding for useful technical conversations. This work is supported by NSERC Discovery grant RGPIN-2022-03339, and the Quantum Computing Challenge Program AQC-206 at the National Research Council of Canada (NRC). Authors further acknowledge the support of Perimeter Institute for Theoretical Physics, supported in part by the Government of Canada through the Department of Innovation, Science and Economic Development Canada (ISED), and by the Province of Ontario through the Ministry of Economic Development, Job Creation, and Trade.

\vspace{-2mm}
\bibliography{main}

\pagebreak
\appendix
\titleformat{\section}[block]{\centering\normalfont\small\bfseries}{\appendixname~\thesection:}{0.5em}{}
\crefalias{section}{appendix}
\crefalias{subsection}{appendix}
\crefalias{subsubsection}{appendix}
\counterwithin{thm}{section}
\counterwithin{rmk}{section}
\counterwithin{lem}{section}
\counterwithin{cor}{section}
\counterwithin{defn}{section}
\counterwithin{prop}{section}
\counterwithin{exm}{section}

\section{Details of the Upper Bound Results}

\subsection{Preliminaries}

Throughout this appendix, as in the main text, we work with $2\pi$-periodic functions $f:\mathbb{R}^d\to\mathbb{R}$, equivalently functions on the flat torus $\mathbb{T}^d=\mathbb{R}^d/(2\pi\mathbb{Z})^d$; we blur the distinction between a periodic function and its periodic lift on $\mathbb{R}^d$, using $\mathbb{T}^d$ or $\mathbb{R}^d$ interchangeably as convenient. We define $f_N\in \mathbb{R}^{(2N+1)}$ to be the vector of values of $f$ on the grid $\frac{2\pi}{(2N+1)}[-N,\dots,N]^d$, and the ket notation implies normalization $\ket{f_N}=f_N/\norm{f_N}$. Throughout, the Fourier cutoff $N \geq 1$ is an integer. We denote by $\diag(f_N)$ the diagonal matrix with the values of $f$ on the grid along its diagonal. We use calligraphic font to denote continuous operators (such as Fokker-Planck operator $\mathcal{L}$, \eqref{eq:fp-operator}), and reserve blackboard bold for discretized operators with the subscript indicating the grid size (such as $\mathbb{H}_N$ denoting a discretization of $\mathcal{H}$).

\subsubsection{Fourier series}

The \emph{Fourier transform} of $f$ is given by
\begin{equation}
\label{eq:fourier-expansion}
f(x)= \sum_{\omega\in\mathbb{Z}^d} \hat f_\omega \,
e^{i\langle \omega, x\rangle}
\end{equation}
with \emph{Fourier coefficients}
\begin{equation}
\label{eq:fourier-coefficients}
\hat f_\omega = \frac{1}{(2\pi)^d}\int_{\mathbb T^d}
f(x) e^{-i\langle \omega,x\rangle} \, dx.
\end{equation}
In contrast, the \emph{discrete} Fourier transform coefficients are
\begin{equation}
\label{eq:discrete-fourier-coefficients}
\tilde f_\omega= \frac{1}{(2N+1)^d}\sum_{x\in \frac{2\pi}{(2N+1)}[-N,\dots,N]^d} f(x) e^{-i \langle \omega, x\rangle}, \quad \forall\, \omega \in [-N,\dots,N]^d.
\end{equation}

Let $\hat f_N= (\hat f_\omega)_{\omega \in [-N,\dots,N]^d} \in \mathbb C^{(2N+1)^d}$ and consider the inclusion $\iota: \mathbb C^{(2N+1)^d} \hookrightarrow \ell^2(\mathbb Z^d)$ induced by $[-N,\dots,N]^d \hookrightarrow \mathbb Z^d$. We will reuse the notation $\hat f_N$ for $\iota \hat f_N$, and similarly extend the discrete Fourier coefficient vector $\tilde f_N$ by zero. By the \emph{tail} of the Fourier series we mean the truncation error
\begin{equation}
\label{eq:def-tail}
\mathcal{E}_N= \norm{\hat f - \hat f_N}_2:= \left(\sum_{\omega \not\in [-N,\dots,N]^d} |\hat f_\omega|^2\right)^{1/2}.
\end{equation}

Recall the definition of Gevrey class functions on $\T^d$~\cite{rodino1993linear}:
\begin{defn}[$s$-Gevrey function]
\label{def:gevrey-function}
For $s \geq 1$ and $r, C > 0$, write $f \in \mathscr{G}^{s}(C,r)$ if the $2\pi$-periodic function $f$ satisfies
\begin{equation}
\label{eq:gevrey-function-bound}
  f \in \mathscr{C}^{\infty}(\R^d),\quad
  \vert \partial^{a} f(x)\vert \leq (a!)^{s} C r^{-\Vert a \Vert_1},
\end{equation}
for all $a\in\N^d$ and $\forall\,x\in [0,2\pi)^d$. We omit the domain since $\T^d$ is used throughout; we write it explicitly, e.g., $\mathscr{G}^{s}(C,r,\R)$, only where a different domain is meant.
\end{defn}
\noindent For these functions truncation error is exponentially small in $N$ as shown by
\begin{prop}[{\cite[Proposition~A.2 and Corollary~A.1]{ronagh2026breaking}}]
\label{prop:dft-distance-bound}
Let $u \in \mathscr G^s(C,r)$ with $s \geq 1$. For $N \in \tilde\Omega((ds)^s/r)$, we have
\begin{equation}
\label{eq:fourier-tail-envelope}
\mathcal E_N \leq 2^d C \sqrt{\frac{4s}{3r^{1/s}} N^{d - 1/s}}\,e^{-\frac34 (r N)^{1/s}},
\end{equation}
and
\begin{equation}
\label{eq:fourier-error-exponential}
\norm{\tilde u_N - \hat u}_2 \leq C e^{-\frac12 (r N)^{1/s}}.
\end{equation}
\end{prop}

\subsubsection{Truncated Fourier differentiation}

Differentiating \eqref{eq:fourier-expansion} termwise, a single partial derivative reads
\begin{align}
\partial_j f(x) = \sum_{\omega\in\mathbb{Z}^d} i\omega_j\, \hat f_\omega\, e^{i\langle \omega,x\rangle}.
\end{align}
More generally, this shows that all derivatives are diagonal operators in the Fourier domain:
\begin{equation}
\label{eq:diff-from-fourier-continuous}
D^\alpha f = \mathcal F^{-1} \diag\left\{(i\omega)^\alpha: \omega \in \mathbb Z^d\right\} \mathcal F (f)
\end{equation}
where $\mathcal F$ denotes the Fourier transform as an operator.

We define the approximate Fourier derivative operator on the lattice $\frac{2\pi}{(2N+1)}[-N,\dots,N]^d$ as:
\begin{align}\label{eq:fourier-derivative-operator}
\tilde\partial_j f(x) = \sum_{\omega\in[-N,\dots,N]^d} i\omega_j\, \tilde{f}_\omega\, e^{i\langle \omega,x\rangle}, \qquad x\in \tfrac{2\pi}{2N+1}[-N,\dots,N]^d.
\end{align}
And more generally,
\begin{equation}
\label{eq:diff-from-fourier}
\tilde D^\alpha_N = F_N^{\dagger \otimes d} \diag\left\{(i\omega)^\alpha: \omega \in [-N,\dots,N]^d\right\} F_N^{\otimes d}.
\end{equation}
The key question is whether such a truncation of $\mathcal F$ and the diagonal operator up to a cutoff mode $N$ approximates $D^\alpha$ well.

\begin{thm}[{\cite[Corollary~A.2]{ronagh2026breaking}}]
\label{thm:high-precision-dft}
Let $f \in \mathscr G^s(C,r)$ be real-valued, with $s \geq 1$. For any multi-index $\alpha \in \mathbb Z^d_{\geq 0}$ we have
\begin{equation}
\norm{\tilde D^\alpha_N f_N - (D^\alpha f)_N}_2
\leq C e^{-\frac12 (r N)^{1/s}},
\end{equation}
provided $N \in \tilde{\Omega}((s(d + |\alpha|))^s/r)$.
\end{thm}

We now consider the approximation error of derivatives of composite and product functions as needed by what follows.

The exact product rule stated in \cite[Proposition~A.17(a)]{motamedi2024gibbs} requires a discretization error for finite-grid Fourier derivatives. We bound this error for Gevrey factors below.

\begin{lem}[Product rule error bound for discrete Fourier derivatives]
\label{lem:fourier-product-rule}
Let $u \in \mathscr G^s(C_u,r)$ and $v \in \mathscr G^s(C_v,r)$. For $N \in \tilde\Omega((ds)^s/r)$, we have
\begin{equation}
\label{eq:fourier-product-rule}
\tilde\partial_j [uv]_N
= (\tilde\partial_j u_N) \cdot v_N + u_N \cdot (\tilde\partial_j v_N) + \vec g,
\end{equation}
where $\norm{\vec g}_2 \leq 3 C_u C_v e^{-\frac12 (r N/2)^{1/s}}$.
\end{lem}

\begin{proof}
The classical Leibniz rule and $s \geq 1$ give
\begin{equation}
\norm{\partial^a(uv)}_\infty
\leq C_u C_v (a!)^s r^{-|a|} \sum_{b\leq a}\binom{a}{b}^{1-s}
\leq C_u C_v (a!)^s (2/r)^{|a|},
\end{equation}
so $uv \in \mathscr G^s(C_u C_v,r/2)$. For $f=u,v,uv$, we denote the differentiation residual by $\vec e_f := \tilde\partial_j f_N - [\partial_j f]_N$. The continuous product rule gives
\begin{equation}
\vec g = \vec e_{uv} - \vec e_u \cdot v_N - u_N \cdot \vec e_v.
\end{equation}
We apply \cref{thm:high-precision-dft} at radius $r/2$ to all three functions and use $\norm{u_N}_\infty \leq C_u$ and $\norm{v_N}_\infty \leq C_v$. The triangle inequality gives the stated bound, proving the claim.
\end{proof}

The error estimate depends on the Gevrey parameters of both factors; it is not a uniform operator-norm bound when one factor is an arbitrary grid vector.

\begin{cor}[Chain rule error bound for discrete Fourier derivatives]
\label{cor:chain-rule-error}
Let $u: \mathbb R \to \mathbb R$ and $v: \mathbb R^d \to \mathbb R$ be smooth functions such that $u\circ v \in \mathscr G^s(C,r)$, with $s \geq 1$. Then, provided $N \in \tilde\Omega((ds)^s/r)$,
\begin{equation}
\label{eq:chain-rule-error}
\tilde\partial_j [u(v(x))]_N = [(\partial_j v(x)) u'(v(x))]_N + \vec{g},
\end{equation}
where $\norm{\vec{g}}_2 \leq C e^{-\frac12 (r N)^{1/s}}$.
\end{cor}

\begin{proof}
We apply \cref{thm:high-precision-dft} to $u\circ v$ with $\alpha = e_j$ and use the classical chain rule $\partial_j(u\circ v)(x) = u'(v(x))\partial_j v(x)$, proving the claim.
\end{proof}

\subsubsection{Fourier interpolation}
\label{sec:fourier-interpolation}

To convert closeness of a discretized state to the Gibbs state, measured in the grid ($\ell^2$) norm, into closeness of the corresponding continuous function to $\sqrt p$ in $L^2$, as required by \cref{lem:resolution-booster,lem:fourier-interpolation-distance,thm:warm-start-query-complexity} below, we use the Fourier interpolation map underlying the upsampling procedure of \cite[Theorem~3.2]{motamedi2024gibbs}, also presented as a resolution booster in \cite[Lemma~9]{leng2026operator}. We write the map in our $2\pi$-periodic, $(2N+1)$-point grid convention. For $k=-N,\dots,N$, let $x_k:=\frac{2\pi k}{2N+1}$ and define
\begin{equation}
\psi_k(x) = \frac{1}{\sqrt{2\pi(2N+1)}}\sum^{N}_{j=-N} e^{ij(x-x_k)}, \quad x \in [0,2\pi).
\end{equation}
A direct computation gives the orthonormality $\int_0^{2\pi}\psi_k(x)\overline{\psi_l(x)}\,dx = \delta_{k,l}$ for $k,l=-N,\dots,N$.

For a quantum state $\ket{u} = \sum^{N}_{k_1,\dots,k_d = -N} u(k_1,\dots,k_d)\ket{k_1,\dots,k_d}$, we define the following interpolation map, which is an isometry embedding $\mathbb{C}^{(2N+1)^d}$ into $L^2(\T^d)$:
\begin{equation}
I_N\ket{u} = \sum^{N}_{k_1,\dots,k_d = -N} u(k_1,\dots,k_d)\psi_{k_1}(x_1)\dots \psi_{k_d}(x_d).
\end{equation}
Observe that
\begin{equation}\label{eq:interp-map}
I_N\ket{u} = \left(\frac{1}{2\pi(2N+1)}\right)^{d/2}\sum_{k, j\in[-N,\dots,N]^d} u(k_1,\dots,k_d)\,e^{-i\frac{2\pi}{2N+1} \langle j,k\rangle}\,e^{i \langle j,x\rangle}.
\end{equation}

\subsection{Gevrey regularity of the Gibbs density}

The Gibbs density is $p(x)=e^{-\beta E(x)}/Z$; since the normalization constant $Z$ is absorbed by ket normalization, the discretized and normalized $e^{-\beta E}$ coincides with the discretized Gibbs density, $\ket{[e^{-\beta E}]_N}=\ket{p_N}$. Note also that $\xi,\rho$ are invariant with respect to shifts $E \to E +\delta$ or simultaneous rescaling $\beta \to \kappa \beta$ and $E\to E/\kappa$. We now derive $s$-Gevrey membership of the density $p$ and of $\sqrt{p}$ implied by \cref{def:potential-class}, i.e. the seminorm bound $[\beta E]_{s,\rho}\leq \xi$ of the potential $E$. We take $s \geq 1$ in the upper-bound analysis; constants in its asymptotic estimates may depend on $s$.

\begin{lem}[Exponential of a Gevrey function]\label{lem:exp-prod}
Let $\Omega \subseteq \R^d$ be open,
$f \in \mathscr{C}^\infty(\Omega;\R)$, $s \ge 1$, $K, R > 0$, and
\begin{equation}\label{eq:hyp}
|\partial^b f| \le K \, (b!)^{s} R^{\onenorm{b}},
\quad\forall x \in \Omega,\quad
\forall\, b \in \N^d \setminus \{0\}.
\end{equation}
Then on $\Omega$ and for $K' \coloneqq d + K + 1$,
\begin{equation}
\label{eq:plain}
|\partial^a e^{f}|
\le (a!)^{s} (K' R)^{\onenorm{a}} \, e^{f},
\quad a \in \N^d.
\end{equation}
\end{lem}

\begin{proof}
Set the shorthand notation $g = e^{f} > 0$. \eqref{eq:plain} follows from the inequality
\begin{equation}\label{eq:induction}
|\partial^a g| \leq  (a!)^{s} (K' R)^{n} \, g
\quad \forall x\in \Omega,\quad
\forall\, a \in \N^d,\ \onenorm{a} = n,
\end{equation}
which we prove by induction on $n \geq 0$. Since $d \geq 1$,
\begin{equation}\label{eq:kappa}
K' - 1 = d + K, \qquad K' \ge K + 2 > 1 + K.
\end{equation}
The case $n = 0$ is verified by equality in \eqref{eq:induction}. For $n = 1$, $a = e_j$:
by \eqref{eq:hyp} at $b = e_j$ and using \eqref{eq:kappa}, it follows
\begin{equation}\label{eq:base}
|\partial_j g| = |\partial_j f| \, g \leq K R \, g
\le \frac{K}{1+K} \, (e_j!)^{s} (K' R) \, g,
\end{equation}
hence \eqref{eq:induction}. Now, for $n \geq 2$ assume the validity of \eqref{eq:induction} up to the case $n-1$. Let $a$ with $\onenorm{a} = n$, such that for index $j$, $a_j \geq 1$. Set $a' \coloneqq a - e_j$ and apply \eqref{eq:induction} at order $\onenorm{a'} = n - 1 \geq 1$.
By the Leibniz rule:
\begin{equation}\label{eq:leibniz}
\partial^a g
= \partial^{a'}\!\bigl( (\partial_j f)\, g \bigr)
= (\partial^a f)\, g
+ \sum_{b < a'} \binom{a'}b
\bigl(\partial^{b+e_j} f\bigr)\bigl(\partial^{a'-b} g\bigr).
\end{equation}
For $b < a'$:
\begin{equation}\label{eq:factorial}
\binom{a'}b \bigl( (b+e_j)! \, (a'-b)! \bigr)^{s}
= \binom{a'}b^{1-s} \bigl( (b_j+1)\, a'! \bigr)^{s}
< (a!)^{s}.
\end{equation}
Accordingly, the inductive hypothesis and \eqref{eq:hyp} give
\begin{equation}
\begin{aligned}
\label{eq:term}
    \binom{a'}b
\bigl|\bigl(\partial^{b+e_j} f\bigr)\bigl(\partial^{a'-b} g\bigr)\bigr|
&\leq K \binom{a'}b
\bigl( (b+e_j)! \, (a'-b)! \bigr)^{s} (K' R)^{\onenorm{a'-b}}R^{\onenorm{b}+1}g\\&
\leq K
\bigl( a! \bigr)^{s} (K' R)^{\onenorm{a}}K'^{-\onenorm{b}-1}g.
\end{aligned}
\end{equation}
Enlarging the range to $b \leq a'$, factorizing over
coordinates, and using
$\log(1-x) \geq \frac{x}{x-1}$ on $x\in [0,1)$,
at $x = K'^{-1}$,
\begin{equation}\label{eq:geometric}
\sum_{b < a'} K'^{-\onenorm{b}-1}
\le \sum_{b \leq a'} K'^{-\onenorm{b}-1}
= \frac{1}K' \prod_{i=1}^{d} \sum_{k=0}^{a'_i} K'^{-k}
\le \frac{\bigl(1 - K'^{-1}\bigr)^{-d}}K'
\le \frac{e^{d/(d+K)}}K'.
\end{equation}
From \eqref{eq:leibniz},
\begin{equation}
|\partial^a g|\leq (a!)^{s}(K' R)^n g \left[\frac{K}{K'^{n}} +K\frac{e^{d/(d+K)}}K' \right],
\end{equation}
and finally using that $K/K'^{n} \leq 1/K'$ for $n\geq 2$, and that
\begin{equation}
\label{eq:log-ineq-chain}
\log (1+x) \geq \frac{x}{x+1} \Longrightarrow e^{d/(d+K)}\leq (K+d)/K \Longrightarrow K e^{d/(d+K)}/K' \leq 1 - 1/K',
\end{equation}
it follows that
\begin{equation}
|\partial^a g|
 \leq  (a!)^{s}(K' R)^n g \left[\frac{1}K' +1-\frac{1}K' \right]
=(a!)^{s}(K' R)^n g,
\end{equation}
which proves the induction.
\end{proof}

\begin{cor}
\label{cor:exp-abs}
Let $\Omega \subseteq \R^d$ be open, $f \in \mathscr{C}^\infty(\Omega;\R)$, and assume hypothesis \eqref{eq:hyp}. If in addition $\sup_{x\in \Omega} f(x)\le F < \infty$ then,
\begin{equation}\label{eq:absolute}
|\partial^a e^{f}| \le e^{F} (a!)^{s} (K' R)^{\onenorm{a}},
\qquad\forall \,x\in \Omega,\quad a \in \N^d.
\end{equation}
If moreover $|f|\leq K$, i.e.\ \eqref{eq:hyp} holds also at $b=0$, then one may take $F=K$ and the exponential maps the Gevrey-class parameters $(K, R)$ to
$\bigl(e^{K}, (d+K+1)\,R\bigr)$.
\end{cor}

\begin{proof}
\eqref{eq:absolute} follows from \eqref{eq:plain} in \cref{lem:exp-prod} supplied with $e^{f}\leq e^{F}$. The case
$b = 0$ of \eqref{eq:hyp} reads $|f| \leq K$, under which $F=K$ is admissible in \eqref{eq:absolute}.
\end{proof}

\begin{cor}[Density regularity from class parameters]
\label{cor:class-reg}
Let $p = e^{-\beta E}/Z \in {\mathscr{W}^{s}(\alpha,\xi,\rho)}$, with $s \geq 1$. Then
\begin{equation}
\label{eq:general-route-radius}
    p \in \mathscr G^{s}\Bigl(\alpha,\ \frac{\rho}{d+\xi+1}\Bigr),
\quad \text{ and } \quad
\sqrt{p} \in \mathscr G^{s}\Bigl(\alpha^{1/2},\ \frac{\rho}{d+\xi/2+1}\Bigr).
\end{equation}
\end{cor}
\begin{proof}
$[\beta E]_{s,\rho}\leq \xi$ is \eqref{eq:hyp} for $f=-\beta E$ with $K=\xi$, $R=1/\rho$, and for $f=-\beta E/2$ with $K=\xi/2$, $R=1/\rho$. \cref{cor:exp-abs} with $F=\sup f$ gives, for all $a\in\N^d$,
\begin{equation}
    |\partial^{a}p|=\frac{|\partial^{a}e^{-\beta E}|}Z\leq(a!)^{s}\Bigl(\frac{d+\xi+1}\rho\Bigr)^{\onenorm{a}}\sup{p},
\end{equation}
and likewise for $\sqrt{p}$ with $d+\xi/2+1$ and $\sup\sqrt{p}$. Since $p$ integrates to $1$ over ${[0,2\pi)^d}$, {which has volume $(2\pi)^d$, its minimum is at most $(2\pi)^{-d}$: $\inf p\leq(2\pi)^{-d}\leq1$}, so $\sup{p}\leq e^{\beta\,\Delta}\inf{p}\leq\alpha$ and $\sup\sqrt{p}\leq\alpha^{1/2}$.
\end{proof}

\noindent
The full Gevrey class property $\mathscr{G}^{s}$ of \cref{cor:class-reg} plays role only in the quantum upper bound, which requires the regularity of the Gibbs density, used in the proof of \cref{lem:fourier-interpolation-distance}.

\begin{cor}[Equivalence of regularity classes]
\label{cor:class-reg-equiv}
The following implications hold:
\begin{equation}
\label{eq:gevrey-equiv-forward}
\beta E(x) \in \dot{\mathscr{G}}^s(\xi, \rho) \Longrightarrow \beta \nabla E(x) \in \mathscr{G}^s(\xi/ \rho, \rho/e^s),
\end{equation}
and vice versa
\begin{equation}
\label{eq:gevrey-equiv-backward}
\beta \nabla E(x) \in \mathscr{G}^s(\xi, \rho) \Longrightarrow \beta E \in \dot{\mathscr{G}}^s(\xi\rho , \rho).
\end{equation}
\end{cor}
\begin{proof}
Take $\beta E(x) \in \dot{\mathscr{G}}^s(\xi, \rho)$, that is
\begin{equation}
    |\beta \partial^a E(x)|\leq \xi (a!)^s \rho^{-\onenorm{a}},\quad \forall a \in \N^d \setminus \{0\},
\end{equation}
and for $a_i>0$, $i\in \{1,\dots,d\}$, write $a'=a-e_i$. The previous inequality gets rewritten as
\begin{equation}
    |\beta \partial^{a'}\partial_i E(x)|\leq \xi (a!)^s \rho^{-\onenorm{a}}= ((a'+e_i)!)^s \xi \rho^{-\onenorm{a'}-1} = (a'!)^s (a'_i+1)^s \frac{\xi}\rho \rho^{-\onenorm{a'}}.
\end{equation}
Use that $\log(1+x)\leq x$ to bound $(a'_i+1)^s\leq e^{s a'_i} \leq e^{s \onenorm{a'}}$, therefore the first implication is verified
\begin{equation}
    |\beta \partial^{a'}\partial_i E(x)|\leq  (a'!)^s \frac{\xi}\rho \left(\frac{\rho}{e^s}\right)^{-\onenorm{a'}}.
\end{equation}
Vice versa, for the second implication, pick $\beta \nabla E \in {\mathscr{G}}^s(\xi ,\rho)$, i.e. for any $i$, $\beta \partial_i E \in {\mathscr{G}}^s(\xi ,\rho)$,
\begin{equation}
    |\beta \partial^a \partial_i E(x)|\leq (a!)^s \xi \ \rho^{-\onenorm{a}},\quad \forall a \in \N^d.
\end{equation}
{Given $a'\in\N^d\setminus\{0\}$, pick $i$ with $a'_i\geq1$ and set $a\coloneqq a'-e_i$; the previous formula reads}
\begin{equation}
    |\beta \partial^{a'} E(x)|\leq ((a'-e_i)!)^s \xi \ \rho^{-\onenorm{a'}+1} = ((a')!)^s (a'_i)^{-s} \xi \rho \ \rho^{-\onenorm{a'}},
\end{equation}
and finally, using that $a'_i\geq 1$ and $s>0$, the second implication is verified:
\begin{equation}
     |\beta \partial^{a'} E(x)|\leq (a'!)^s (\xi \rho)  \rho^{-\onenorm{a'}},\quad \forall a' \in \N^d \setminus\{0\}.
\end{equation}
This completes the proof.
\end{proof}
\noindent
A consequence of~\cref{cor:class-reg-equiv} is that membership of the Gibbs state $p\in\mathscr{W}^s(\alpha,\xi,\rho)$ in \cref{def:potential-class} implies the $s$-Gevrey property of the log-gradient,
\begin{equation}
\label{eq:log-gradient-gevrey-class}
\nabla \log p \in \mathscr{G}^s(\xi', \rho'),\qquad \xi'=\xi/\rho,\quad\rho'=\rho/e^s,
\end{equation}
supplied with the range bound $\beta\Delta\leq\min\{\log\alpha,\pi d\,\xi'\}$.

We conclude this section by bounding the approximation error of the low-order derivatives of the Gibbs density.
\begin{cor}[Error in the chain rule for discrete Fourier derivatives for Gibbs state]
\label{cor:chain-rule-error-gibbs}
Let $e^{-k\beta E/2}$ be $s$-Gevrey and $2\pi$-periodic in all dimensions, with Gevrey constants $(\xi,\rho)$, $s \geq 1$, and $k \in \mathbb{R}$. Then, provided $N \in \tilde\Omega((ds)^s/\rho)$,
\begin{equation}
\label{eq:chain-rule-error-gibbs}
\tilde\partial_j [e^{-k\beta E/2}]_N=-\frac{k\beta}{2}[e^{-k\beta E/2}]_N [\partial_j E]_N+\vec{g},
\end{equation}
and
\begin{equation}
\label{eq:chain-rule-error-gibbs-second-deriv}
\tilde{\partial_j}^2 [e^{-k\beta E/2}]_N
=-\frac{k\beta}{2}{[\partial_j^2E]_N} [e^{-k\beta E/2}]_N+\frac{k^2\beta^2}{4}{[e^{-k\beta E/2}]}_N [(\partial_jE)^2]_N+\vec{g},
\end{equation}
where $\norm{\vec{g}}_2 \leq \xi e^{-\frac12 (\rho N)^{1/s}}$ in each identity.
\end{cor}
\begin{proof}
We obtain the first identity from \cref{cor:chain-rule-error} with $u= \exp$ and $v= -k\beta E/2$. For the second, we apply \cref{thm:high-precision-dft} with $\alpha= 2e_j$ and use
\begin{equation}
\partial_j^2(e^{-k\beta E/2}) = -\frac{k\beta}{2} (\partial_j^2 E) e^{-k\beta E/2} + \frac{k^2\beta^2}{4} (\partial_j E)^2 e^{-k\beta E/2}
\end{equation}
This proves both bounds.
\end{proof}

\subsection{Discretization of the Fokker-Planck operator}
\label{app:discretization}

In this section we discretize the Witten factors associated with the Fokker-Planck operator
\begin{align}
\label{eq:fp-operator-similarity-form}
\mathcal{L}(f)(x)=\beta^{-1}\nabla \cdot \left(e^{-\beta E}\nabla(e^{\beta E}f(x))\right),
\end{align}
equivalently written as
\begin{align}\label{eq:fp-operator-drift-diffusion-form}
\mathcal{L}(f)(x)=\nabla \cdot (\nabla E(x) f(x))+\beta^{-1}\nabla^2 f(x),
\end{align}
and study the resulting spectral properties, using the factorization of \cite{leng2026operator}.

\begin{itemize}
\item $\mathcal{H}=-e^{\beta E/2}\mathcal{L}e^{-\beta E/2}$ is a continuous operator known as the Witten Laplacian from \cite[Eqn.~4]{leng2026operator}.

\item $\mathbb{H}_N=\sum_{j}\mathbb{L}_{j,N}^* \mathbb{L}_{j,N}$ is the discretization of $\mathcal{H}=\sum_j L_j^* L_j$. To control aliasing, we evaluate the factors on the finer grid of cutoff $2N$. Let $\iota_N = I_{2N}^* I_N$ be the isometry obtained by a Fourier transform, zero-padding, and an inverse Fourier transform. We define
\begin{equation}
\label{eq:oversampled-factors}
\mathbb L_{j,N} = -\frac{i}{\sqrt\beta} \iota_N\tilde\partial_j - \frac{i\sqrt\beta}{2}[\partial_j E]_{2N} \iota_N,
\qquad \iota_N^* \iota_N = \mathbb I.
\end{equation}
Here $\tilde\partial_j$ acts on the $N$-grid, and $I_N$ is the Fourier interpolation isometry \eqref{eq:interp-map}. Thus the state space remains the $N$-grid, while each factor maps into the $2N$-grid.

\item ${\mathbb{W}_N} = [\mathbb{L}_{1,N}^T,\dots, \mathbb{L}_{d,N}^T]^T$ uses the factor stacking of \cite[Eqn.~44]{leng2026operator}, with the finer-grid factors \eqref{eq:oversampled-factors}. We rename it here to avoid a clash with our own notation $\mathbb{L}_N$ for the discretization of $\mathcal{L}$.
\end{itemize}

\subsubsection{Continuous spectral gap}

\begin{lem}[Lower bound on $\text{Gap}\left(\mathcal{L}^{\dagger}\right)$]
\label{lem:spectral-gap-bound}
Let $E:\mathbb{T}^{d}\rightarrow\mathbb{R}$. Let $\mathcal{L}^{\dagger} = -\nabla E\cdot \nabla + \beta^{-1}\Delta$ for inverse temperature $\beta$ according to \cite{leng2026operator}. Then the spectral gap of $\mathcal{-L}^{\dagger}$, denoted as $\text{Gap}\left(\mathcal{L}^{\dagger}\right)$, is lower bounded as:
\begin{equation}
\text{Gap}\left(\mathcal{L}^{\dagger}\right)\ge\beta^{-1}e^{-\beta \Delta}
\end{equation}
\end{lem}
\begin{proof}
First we define for $f:\mathbb{T}^{d}\rightarrow\mathbb{R}$, $g:\mathbb{T}^{d}\rightarrow\mathbb{R}$ and the Gibbs density $p$ that:
\begin{equation}
\label{eq:weighted-inner-product-def}
\braket{f,g}_{p} = \int_{\mathbb{T}^{d}}fg\,p\, dx,
\qquad
||f||_{p}^2 = \braket{f,f}_{p}
\end{equation}
The $\text{Gap}\left(\mathcal{L}^{\dagger}\right)$ is defined in \cite{leng2026operator} as:
\begin{equation}\label{eq:gap-l-dagger}
\text{Gap}\left(\mathcal{L}^{\dagger}\right) \coloneqq \inf_{f\notin\ker(\mathcal{L}^{\dagger})}\frac{\braket{f,\mathcal{-L}^{\dagger}f}_{p}}{\text{Var}_{p}(f)}
\end{equation}
where $\text{Var}_{p}(f)\coloneqq ||f-\int_{\mathbb{T}^{d}}f\,p\, dx||_{p}^{2}$. Using integration by parts and the fact that $\mathbb T^d$ has no boundary, we obtain \cite{pavliotis2014stochastic, kazdan2008partial}:
\begin{equation}
\label{eq:inner-product-simple}
\begin{aligned}
\braket{f,\mathcal{-L}^{\dagger}f}_{p}
&= \int_{\mathbb{T}^{d}} f\nabla E \cdot \nabla f\, p\, dx - \beta^{-1}\int_{\mathbb{T}^{d}} (\Delta f) f\,p\, dx\\
&= \int_{\mathbb{T}^{d}} f\nabla E \cdot \nabla f\, p\, dx + \beta^{-1}\int_{\mathbb{T}^{d}} \nabla f \cdot \nabla(f\,p)\, dx\\
&= \int_{\mathbb{T}^{d}} f\nabla E \cdot \nabla f\, p\, dx
+ \beta^{-1}\left[\int_{\mathbb{T}^{d}} \nabla f \cdot \nabla f\, p\, dx - \int_{\mathbb{T}^{d}}\beta f\nabla E \cdot \nabla f\, p\, dx\right]\\
&= \beta^{-1}\int_{\mathbb{T}^{d}} \nabla f \cdot \nabla f\, p\, dx.
\end{aligned}
\end{equation}
Hence from \eqref{eq:gap-l-dagger} and \eqref{eq:inner-product-simple}
\begin{equation}
\text{Gap}\left(\mathcal{L}^{\dagger}\right) = \beta^{-1}\inf_{f\notin\ker(\mathcal{L}^{\dagger})}\frac{\int_{\mathbb{T}^{d}} \nabla f \cdot \nabla f\, p\, dx}{\text{Var}_{p}(f)}
\end{equation}
Then applying Holley-Stroock perturbation principle \cite{schlichting2012eyringkramers} using the uniform distribution on the torus (which satisfies a Poincare inequality with Poincare constant $1$ \cite{tropp2021probability}) as the reference distribution and using $\min E\le E \le \max E$ we obtain:
\begin{equation}
\text{Gap}\left(\mathcal{L}^{\dagger}\right) = \beta^{-1}\inf_{f\notin\ker(\mathcal{L}^{\dagger})}\frac{\int_{\mathbb{T}^{d}} \nabla f \cdot \nabla f\, p\, dx}{\text{Var}_{p}(f)} \ge  \beta^{-1}e^{-\beta \Delta}
\end{equation}
which completes the proof.
\end{proof}

\subsubsection{Galerkin comparison and its spectral consequences}
\label{sec:galerkin-comparison}

We compare $\mathbb H_N$ with the Galerkin matrix $\mathbb G_N = I_N^*\mathcal H I_N$. Throughout, $\norm{\cdot}$ denotes the Euclidean norm for vectors and the operator norm for matrices. We write $P_N = I_N I_N^*$ for the orthogonal Fourier projection in $L^2(\mathbb T^d)$. For any matrix $M\in\mathbb{R}^{n\times n}$, we order its singular values increasingly as
\begin{equation}
\label{eq:singular-value-order}
\sigma_{\min}(M):=\sigma_0(M), \sigma_1(M), \sigma_2(M), \dots, \sigma_{\max}(M).
\end{equation}
For symmetric $M$, its eigenvalues are ordered analogously, with $\lambda_{\min}(M)=\lambda_0(M)$.

\begin{lem}[Galerkin comparison]
\label{lem:galerkin-comparison}
Let $p\in\mathscr W^s(\alpha,\xi,\rho)$, $\beta>0$, and $0<\epsilon<1/2$. For a sufficiently large $N$ with
\begin{align}
N \in \tilde\Omega\left(\frac{d + \xi}{\rho} \left[ds + \log \frac{\alpha(1 + \xi/\rho)}{\epsilon}\right]^s\right),
\end{align}
we have
\begin{align}
\label{eq:galerkin-comparison}
\beta\norm{\mathbb H_N - \mathbb G_N} \leq \epsilon,
\quad \text{and} \quad
\beta\lambda_0(\mathbb G_N) \leq \epsilon.
\end{align}
\end{lem}

\begin{proof}
We set $b_j = \beta\partial_j E/2$, $B = \sqrt d\,\xi/(2\rho)$, and $\tau_N = \bigl(\sum_j\norm{b_j-P_N b_j}_\infty^2\bigr)^{1/2}$. For $f=I_N v$ with $\norm v=1$, replacing $b_j$ by $P_N b_j$ makes $(\partial_j+P_N b_j)f$ a trigonometric polynomial of cutoff $2N$. Its squared $L^2$ norm is therefore exactly the grid quadratic form on the $2N$-grid. In both the continuous and grid norms the replacement changes the stacked factor by at most $\tau_N$, whereas the original stacked factor has norm at most $K_N=\sqrt d\,N+B$. Consequently,
\begin{equation}
\beta\norm{\mathbb H_N-\mathbb G_N}
\leq 2\tau_N(2K_N+\tau_N).
\end{equation}

We now recall the Fourier-tail estimate needed to choose $N$. For $f\in\mathscr G^s(C,r)$ and $\omega\ne0$, integration by parts in a coordinate of largest frequency, with $m=\lfloor(r\norm\omega_\infty)^{1/s}/e\rfloor$, gives
\begin{equation}
|\hat f_\omega|\leq C\inf_{m\geq0}\frac{(m!)^s}{(r\norm\omega_\infty)^m}
\leq C e^s e^{-(s/e)(r\norm\omega_\infty)^{1/s}}.
\end{equation}
For $q=0,1$, the weighted tail is bounded by
\begin{equation}
\sum_{\norm\omega_\infty>N}(1+\norm\omega_2)^q|\hat f_\omega|
\leq C e^s\sum_{\ell>N}2d(2\ell+1)^{d-1}(1+\sqrt d\,\ell)^q e^{-(s/e)(r\ell)^{1/s}}.
\end{equation}
For fixed $s$, summing this stretched-exponential tail and absorbing its polynomial prefactors makes it at most $C\eta$ at a cutoff $N \in \tilde\Omega(r^{-1}[s(d+q)+\log(1/\eta)]^s)$. The suppressed factors are polylogarithmic in $d$, $r+r^{-1}$, and $1/\eta$. This also bounds the uniform truncation error ($q=0$) and, after a factor $(2\pi)^{d/2}$, its $H^1$ norm ($q=1$). Here $\norm{f}_{H^1} = \bigl(\norm{f}_{L^2}^2 + \sum_{j=1}^d \norm{\partial_j f}_{L^2}^2\bigr)^{1/2}$ is the first-order Sobolev norm, combining the $L^2$ norms of the function and its first derivatives.

By \cref{cor:class-reg-equiv,cor:class-reg}, the functions $b_j$ and $\psi=\sqrt p$ have respective Gevrey amplitudes $\xi/(2\rho)$ and $\sqrt\alpha$, and share the smaller valid radius $r=\rho/[e^s(d+\xi+1)]$. Thus the stated cutoff makes $2\tau_N(2K_N+\tau_N)\leq\epsilon$. The additional factor $N$ inside $K_N$ costs only the polylogarithmic factors already suppressed.

For the second claim we use the normalized trial function $P_N\psi/\norm{P_N\psi}_{L^2}$. Writing $t_N=\norm{(1-P_N)\psi}_{H^1}$, the continuous identities $(\partial_j+b_j)\psi=0$ imply
\begin{equation}
\beta\lambda_0(\mathbb G_N)
\leq \frac{\sum_j\norm{(\partial_j+b_j)(1-P_N)\psi}_{L^2}^2}{\norm{P_N\psi}_{L^2}^2}
\leq \frac{(1+B)^2 t_N^2}{1-t_N^2}.
\end{equation}
The same tail estimate and cutoff give $t_N\leq\sqrt\epsilon/[2(1+B)]<1/2$, making the last bound at most $\epsilon$. This proves both claims.
\end{proof}

\begin{prop}[Weyl's inequality for singular values]
\label{prop:weyl-singular-value}
    Let $M$ and $D$ be matrices of the same dimension. Then
\begin{equation}
\label{eq:weyl-inequality}
|\sigma_k(M+D)-\sigma_k(M)|\le \sigma_{\max}(D).
\end{equation}
\end{prop}

\begin{lem}[The smallest two singular values of $\mathbb{H}_N$]
\label{lem:least-eigenvalues-hn}
Under the conclusions of \cref{lem:galerkin-comparison} with error $\epsilon > 0$, we have
\begin{equation}
\label{eq:hn-small-singular-value-bounds}
0\leq\sigma_{\min}(\mathbb H_N)\leq \frac{2\epsilon}{\beta},
\quad \text{and} \quad
\sigma_1(\mathbb H_N)\geq \frac{\alpha^{-1}-\epsilon}{\beta},
\end{equation}
and the following lower bound on the gap:
\begin{equation}
\label{eq:hn-spectral-gap-bounds}
\operatorname{Gap}(\mathbb H_N)\geq \frac{\alpha^{-1}-3\epsilon}{\beta}.
\end{equation}
\end{lem}

\begin{proof}
Multiplication by $\sqrt p$ is unitary from $L^2(p)$ to $L^2(dx)$ and conjugates $-\mathcal L^\dagger$ to $\mathcal H$. Hence \cref{lem:spectral-gap-bound} gives $\lambda_1(\mathcal H)\geq e^{-\beta\Delta}/\beta\geq1/(\alpha\beta)$. The min--max principle gives $\lambda_1(\mathbb G_N)\geq\lambda_1(\mathcal H)$. Both $\mathbb G_N$ and $\mathbb H_N$ are positive semi-definite, so \cref{prop:weyl-singular-value} and \cref{lem:galerkin-comparison} yield the first two bounds. Subtracting them proves the stated gap bound.
\end{proof}

\begin{thm}
\label{thm:overlap-guarantee}
Given $0<\tilde\epsilon<1/2$ and $p\in\mathscr W^s(\alpha,\xi,\rho)$, the choice
\begin{align}
\label{eq:overlap-guarantee-n-choice}
N \in \tilde\Omega\left(\frac{d + \xi}{\rho} \left[ds + \log \frac{\alpha(1 + \xi/\rho)}{\tilde\epsilon}\right]^s\right),
\end{align}
with sufficiently large constants guarantees a unique ground state $\ket g$ of $\mathbb H_N$, with
\begin{equation}
0\leq\lambda_0(\mathbb H_N)=\sigma_{\min}(\mathbb H_N)\leq\frac{\tilde\epsilon^2}{12\alpha\beta},
\quad \text{and} \quad
\norm{\ket{[e^{-\beta E/2}]_N}-\ket g}\leq\tilde\epsilon.
\end{equation}
\end{thm}
\begin{proof}
We apply \cref{lem:galerkin-comparison,lem:least-eigenvalues-hn} with $\epsilon=\tilde\epsilon^2/(24\alpha)$. The resulting positive gap proves uniqueness and gives the displayed eigenvalue bound. For $\psi=\sqrt p$, the continuous gap and the comparison estimate give
\begin{equation}
1-|\langle\psi,I_N g\rangle|^2
\leq \alpha\beta\langle I_N g,\mathcal H I_N g\rangle
\leq 3\alpha\epsilon.
\end{equation}
After phase alignment, $\norm{I_N g-\psi}_{L^2}\leq\sqrt{6\alpha\epsilon}=\tilde\epsilon/2$. The same cutoff gives, by \cref{lem:fourier-interpolation-distance},
\begin{equation}
\norm{I_N\ket{[e^{-\beta E/2}]_N}-\psi}_{L^2}\leq\tilde\epsilon/2.
\end{equation}
The triangle inequality and the isometry of $I_N$ prove the claimed vector bound.
\end{proof}

\subsection{Quantum algorithm for Gibbs sampling}

\subsubsection{Quantum algorithm of \texorpdfstring{\cite{leng2026operator}}{[Leng, Ding, Chen, Lin 2025]} revamped}

In this section we apply the QSVTh algorithm of \cite{leng2026operator} to the potential class \eqref{eq:potential-class}, given a warm start. The previous section proves the needed ground-state accuracy and spectral bounds for the finer-grid factors \eqref{eq:oversampled-factors}, and determines the discretization size and threshold.

The continuous square-root Gibbs function spans the kernel of the following Hamiltonian, which admits the square-root decomposition used in \cite{leng2026operator}:
\begin{equation}\label{eq:spatial-discretize-lj}
\mathcal{H} = \sum^d_{j=1}L^\dagger_j L_j,\quad L_j \coloneqq -i\frac{1}{\sqrt{\beta}} \partial_{x_j} - i\frac{\sqrt{\beta}}{2}\partial_{x_j}E \quad \forall j \in [d]=\{1,2,\dots,d\}.
\end{equation}
Thus, sampling reduces to approximating the kernel of $\mathcal H$. We use the finer-grid factors \eqref{eq:oversampled-factors} and stack them as $\mathbb W_N = [\mathbb L_{1,N}^\top,\dots,\mathbb L_{d,N}^\top]^\top$, so that $\mathbb H_N=\mathbb W_N^\dagger\mathbb W_N$.

\begin{prop}[Quantum singular value thresholding, {\cite[Proposition~7]{leng2026operator}}]
\label{prop:sv-filter}
Let $A \in \mathbb{C}^{2^n\times 2^p}$ be encoded by its $(\gamma,m)$-block-encoding $U_A$. Let $\sigma_1, \sigma_2$ be the first two singular values of $A$ and $0 \le {\sigma_1} \le \mathfrak{s}_1 < \mathfrak{s}_2 \le \sigma_2$. We denote $\mathfrak{s} = \mathfrak{s}_2 -  \mathfrak{s}_1$. Let $f(x)$ be the rectangular filter function
\begin{align}
\label{eq:rectangular-function}
f(x) = \begin{cases}
    1, & x \in [-\mathfrak{s}_1,\mathfrak{s}_1],\\
    0, & x \in [-\infty, -\mathfrak{s}_2]\cup [\mathfrak{s}_2, \infty].
\end{cases}
\end{align}
We can implement a $(1,m+1,\epsilon)$-block-encoding of the matrix function $f^{\triangleright}(A)$ using $U_A$, $U^\dagger_A$, $m$-qubit controlled NOT, and single-qubit rotation gates for $\mathcal{O}\left(\gamma  \mathfrak{s}^{-1}\log(\epsilon^{-1})\right)$ times.
\end{prop}

Since singular values of $\mathbb{W}_N$ are square roots of those of $\mathbb{H}_N$, having an explicit expression for square root of $\mathbb{H}_N$ results in the quadratic speedup in the QSVTh algorithm \cref{prop:sv-filter} in terms of the spectral gap.

For this approach to be effective, one has to know that the ground state of $\mathbb{W}_N$, equivalently, the ground state of $\mathbb{H}_N$ can be made $\tilde\epsilon$-close to the discretized Gibbs state $\ket{[e^{-\beta E/2}]_N}$ with relatively small lattice size. Additionally, the spectral gap of the discretized operator $\mathbb{H}_N$ should be bounded below by a constant independent of $N$, so that the runtime remains well-controlled.

\begin{cor}
\label{cor:threshold-choice}
For a potential $E$ satisfying \cref{def:potential-class} and a constant-overlap warm start, QSVTh (\cref{prop:sv-filter}) applied to $\mathbb W_N$ prepares a normalized grid state of error $0<\tilde\epsilon<1/2$ on heralded success, with constant success probability. The physical singular-value threshold width is $\mathfrak{s}=\left(\frac{\sqrt3}{2}-\frac1{\sqrt2}\right)/\sqrt{\alpha\beta}$ when
\begin{align}
N \in \tilde\Omega\left(\frac{d + \xi}{\rho} \left[ds + \log \frac{\alpha(1 + \xi/\rho)}{\tilde\epsilon}\right]^s\right).
\end{align}
\end{cor}

\begin{proof}
We set $g_0=1/(\alpha\beta)$. The error choice in \cref{thm:overlap-guarantee} and the bounds in \cref{lem:least-eigenvalues-hn} give
\begin{equation}
\lambda_0(\mathbb H_N)<g_0/2<3g_0/4<\lambda_1(\mathbb H_N).
\end{equation}
Thus the physical thresholds for $\mathbb W_N$ are $\mathfrak s_1=\sqrt{g_0/2}$ and $\mathfrak s_2=\sqrt{3g_0/4}$. For its block encoding they are $\mathfrak s_1/\gamma$ and $\mathfrak s_2/\gamma$; this gives the stated width before normalization.

If the initial overlap with the grid Gibbs state is at least $c>0$, we choose the ground-state error at most $\min\{c/2,\tilde\epsilon/2\}$ and the filter's operator error at most $\eta=c\tilde\epsilon/8$. The overlap with $g$ is then at least $c/2$. Normalizing the successful branch changes it from $g$ by at most $2\eta/(c/2)\leq\tilde\epsilon/2$, and its success probability is at least $(c/2-\eta)^2\geq c^2/16$. The triangle inequality proves the grid-state error bound.
\end{proof}

\begin{thm}[Block-encoding of $\mathbb{W}_N$, adapted from {\cite[Proposition~17]{leng2026operator}}]
\label{thm:block-encoding-wn}
Let $N$ and $R := \max_x \|\nabla E(x)\|$ are the same as above, and $\gamma=\pi N\sqrt{d/\beta}+\sqrt{\beta}R/2$.
We can implement a block encoding of $\mathbb W_N$ with normalization $\gamma$, using 2 queries to $O_{\nabla E}$ (or its inverse) and an additional $\tilde{\mathcal O}(d^2)$ elementary gates. Here the tilde suppresses polylogarithmic factors in $d$ and $N$.
\end{thm}

\begin{proof}
We use the derivative and gradient block encodings and their linear combination from \cite[Proposition~17]{leng2026operator}. In the derivative term we compose the $N$-grid derivative with $\iota_N$; its normalization is unchanged because $\iota_N$ is an isometry. In the gradient term we first apply $\iota_N$ and query the gradient on the $2N$-grid, retaining the normalization $\sqrt\beta R/2$. Implementing $\iota_N$ by Fourier transforms and zero-padding requires no oracle queries and $d\,\polylog(N)$ gates. The rectangular input and output grid subspaces are encoded with their respective projections, so unused padded basis states do not enter the filtering subspace. These modifications preserve the stated normalization and query bound.
\end{proof}

Using the block-encoding from \cref{thm:block-encoding-wn} and \cref{cor:threshold-choice}, we can update the runtime for the Gibbs preparation in the next theorem.

\begin{lem}[Resolution booster, {\cite[Lemma~9]{leng2026operator}; cf.~\cite[Proposition~A.7]{motamedi2024gibbs}}]
\label{lem:resolution-booster}
For a fixed $\epsilon > 0$, suppose there is an $N$ such that we have access to a state $\ket{g}\in \mathbb{C}^{(2N+1)^d}$ where
\begin{equation}\label{eq:interpolation-assumption}
\left\|I_N\ket{g} - \sqrt{p}\right\|_{L^2} \le \epsilon/2,
\end{equation}
with ${\sqrt{p}} \in L^2(\mathbb T^d)\propto e^{-\beta E/2}$ being the normalized amplitude function. Then, there is a quantum algorithm that returns a sample from a random variable $X$ following the distribution $\eta$ such that $\rm{TV}(\eta, p) \le \epsilon$ using one copy of the state $\ket{g}$, and an additional $d\cdot \polylog(dN/\epsilon)$ elementary gates.
\end{lem}

\begin{lem}
\label{lem:fourier-interpolation-distance}
Suppose $E$ is a Gevrey potential as in \cref{def:potential-class}, with $s \geq 1$ and $0 < \epsilon < 1$. We have
\begin{align}
\norm{I_N\ket{[e^{-\beta E/2}]_N}-\sqrt{p}}_{L^2}<\epsilon/2,
\end{align}
provided $N \in \tilde\Omega\left(\frac{d + \xi}{\rho} \left[ds + \log \frac{\alpha}{\epsilon}\right]^s\right)$.
\end{lem}

\begin{proof}
We write $u = \sqrt p$ and use \cref{cor:class-reg} with
\begin{align}
\xi_{\sqrt p} = \sqrt\alpha,\qquad \rho_{\sqrt p} = \frac{\rho}{d + \xi/2 + 1}.
\end{align}
By \eqref{eq:interp-map}, Parseval's identity, and the normalized-vector bound of \cite[Corollary~1]{ronagh2026breaking}, the Fourier estimate \eqref{eq:fourier-error-exponential} gives
\begin{align}
\norm{I_N\ket{u_N} - u}_{L^2}
&= \norm{\frac{\tilde u_N}{\norm{\tilde u_N}_2} - \frac{\hat u}{\norm{\hat u}_2}}_2\nonumber\\
&\leq \frac{2\xi_{\sqrt p}}{\norm{\hat u}_2} e^{-\frac12 (\rho_{\sqrt p} N)^{1/s}}
= 2(2\pi)^{d/2} \sqrt\alpha\,e^{-\frac12 (\rho_{\sqrt p} N)^{1/s}},
\end{align}
where $\norm{u}_{L^2} = 1$ implies $\norm{\hat u}_2 = (2\pi)^{-d/2}$. Thus an explicit sufficient cutoff is
\begin{equation}
\label{eq:resolution-n-expression}
N \geq \frac1{\rho_{\sqrt p}} \max\left\{
2\log \frac{4(2\pi)^{d/2} \sqrt\alpha}{\epsilon},
64ds \log\!\left(\frac{64ds}{\rho_{\sqrt p}^{1/s}}\right)
\right\}^s.
\end{equation}
The second term meets the cutoff condition of \cref{prop:dft-distance-bound}, and the first bounds the interpolation error by $\epsilon/2$, proving the claim.
\end{proof}

\begin{lem}
\label{lem:approximate-gibbs-interpolation}
Suppose $\ket{\tilde{[e^{-\beta E/2}]}}$ is such that $\norm{\ket{\tilde{[e^{-\beta E/2}]}} - \ket{[e^{-\beta E/2}]_N}}_2 \leq \epsilon/4$. Then
\begin{equation}
\label{eq:interpolation-error-bound}
\norm{I_N\ket{\tilde{[e^{-\beta E/2}]}}-{\sqrt{p}}}_{L^2}\le \epsilon /2,
\end{equation}
provided $s \geq 1$ and $N \in \tilde\Omega\left(\frac{d + \xi}{\rho} \left[ds + \log \frac{\alpha}{\epsilon}\right]^s\right)$, with the explicit choice \eqref{eq:resolution-n-expression} at precision $\epsilon/2$.
\end{lem}

\begin{proof}
By virtue of $I_N$ being an isometry, the hypothesis implies that
\begin{align}
\norm{I_N\ket{\tilde{[e^{-\beta E/2}]}} - I_N\ket{{\sqrt{p}_N}}}_{L^2} \le \epsilon/4.
\end{align}
By the triangle inequality we have
\begin{align}
\norm{I_N\ket{\tilde{[e^{-\beta E/2}]}}-{\sqrt{p}}}_{L^2}\le \norm{I_N\ket{\tilde{[e^{-\beta E/2}]}}-I_N{\ket{\sqrt{p}_N}}}_{L^2}+\norm{I_N\ket{\sqrt{p}_N}-\sqrt{p}}_{L^2}.
\end{align}
By \cref{lem:fourier-interpolation-distance}, the second term is
\begin{align}
\norm{I_N\ket{[e^{-\beta E/2}]_N} - \sqrt{p}}_{L^2} \leq \epsilon/4,
\end{align}
at the stated cutoff. The sum is at most $\epsilon/2$, proving the claim.
\end{proof}

\begin{cor}
\label{cor:gevrey-sampling}
Suppose $s \geq 1$ and $\ket{\tilde{[e^{-\beta E/2}]}}$ is such that $\norm{\ket{\tilde{[e^{-\beta E/2}]}} - \ket{[e^{-\beta E/2}]_N}}_2 \leq \epsilon/4$. Then, provided
\begin{equation}
\label{eq:gevrey-sampling-n-choice-appendix}
N \in \tilde\Omega\left(\frac{d + \xi}{\rho} \left[ds + \log \frac{\alpha}{\epsilon}\right]^s\right),
\end{equation}
there is a quantum algorithm that returns a sample from a random variable $X$ following the distribution $\eta$ such that $\rm{TV}(\eta, p) \le \epsilon$ using one copy of the state $\ket{\tilde{[e^{-\beta E/2}]}}$, and an additional $d\cdot \polylog(dN/\epsilon)$ elementary gates.
\end{cor}

\begin{proof}
We apply \cref{lem:resolution-booster} to the interpolated-state bound of \cref{lem:approximate-gibbs-interpolation}, proving the claim.
\end{proof}

\begin{thm}
\label{thm:warm-start-query-complexity}
Consider a Gibbs state $p(x)= e^{-\beta E(x)}/Z$ with potential $E(x)$ satisfying \eqref{eq:potential-class}. Let $R := \max_x \|\nabla E(x)\|$. There exists a quantum algorithm with access to a warm-start state $\ket{\phi}$
\begin{equation}
|{\langle{\phi}|{\sqrt{p}_N}\rangle}| = \Omega(1),
\end{equation}
that, on heralded success of constant probability, returns a sample from a random variable $X\sim\eta$ such that $\mathrm{TV}(\eta,p) \le \tilde\epsilon$, after $\mathcal{M}$ quantum queries to the quantum gradient oracle $O_{\nabla E}$ or its inverse, where
\begin{equation}
\label{eq:warm-start-query-complexity}
\mathcal{M}\in   O\left(\sqrt{\alpha}\left(\pi N\sqrt{d}+ \beta R\right ) \log\left(\frac{1}{\tilde\epsilon}\right)\right),
\end{equation}
and
\begin{align}
N \in \tilde\Omega\left(\frac{d + \xi}{\rho} \left[ds + \log \frac{\alpha(1 + \xi/\rho)}{\tilde\epsilon}\right]^s\right).
\end{align}
\end{thm}

\begin{proof}
We apply \cref{prop:sv-filter,cor:threshold-choice} with a constant-fraction allocation of $\tilde\epsilon$ to filtering and ground-state approximation. The required number of calls to the block encoding of $\mathbb W_N$ is $\mathcal O(\gamma \sqrt{\alpha \beta} \log(1/\tilde\epsilon))$, with $\gamma = \pi N \sqrt{d/\beta} + \sqrt\beta R/2$ from \cref{thm:block-encoding-wn}. Each call uses two gradient-oracle queries, giving \eqref{eq:warm-start-query-complexity}. With the same constant-fraction choice in \cref{thm:overlap-guarantee}, the output satisfies
\begin{align}
\norm{\ket{\tilde{[e^{-\beta E/2}]}} - \ket{[e^{-\beta E/2}]_N}}_2 \leq \tilde\epsilon/4.
\end{align}
The stated cutoff also satisfies \eqref{eq:resolution-n-expression} at precision $\tilde\epsilon/2$. Hence \cref{cor:gevrey-sampling} gives total-variation error at most $\tilde\epsilon$, proving the claim.
\end{proof}

\subsubsection{Annealed quantum Gibbs sampling}\label{app:qsvt-anneal}

\begin{lem}[Overlap of adjacent Gibbs states]
\label{lem:annealing-overlap-bound}
Let $\ket{\psi_i}$ and $\ket{\psi_{i+1}}$ be the normalized grid Gibbs states for $E$ on a common grid, at inverse temperatures $\beta_i$ and $\beta_{i+1}=\beta_i+\delta$, with $\delta>0$. Writing $\Delta=\max E-\min E$, we have
\begin{equation}
\label{eq:annealing-overlap-bound}
|\braket{\psi_{i+1}|\psi_i}| \geq \frac{1}{\cosh(\delta\Delta/4)}.
\end{equation}
\end{lem}

\begin{proof}
We write $\pi_i(x)=|\braket{x|\psi_i}|^2$ for the grid Gibbs distribution, and set $E^c=E-(\min E+\max E)/2$, $a=\delta\Delta/4$, $Y=e^{-\delta E^c/2}$, and $\mu=\mathbb E_{\pi_i}Y$. Centering the potential does not change the Gibbs states, and their definitions give
\begin{equation}
\braket{\psi_{i+1}|\psi_i}
= \frac{\mathbb E_{\pi_i}Y}{\sqrt{\mathbb E_{\pi_i}Y^2}}
= \frac{\mu}{\sqrt{\mathbb E_{\pi_i}Y^2}}.
\end{equation}
Since $e^{-a}\leq Y\leq e^a$, we have $(Y-e^{-a})(Y-e^a)\leq0$. Taking expectations yields
\begin{equation}
\mathbb E_{\pi_i}Y^2 \leq 2\cosh(a)\mu-1 \leq \cosh^2(a)\mu^2,
\end{equation}
where the last inequality is equivalent to $(\cosh(a)\mu-1)^2\geq0$. Substituting into the overlap identity gives $\braket{\psi_{i+1}|\psi_i}\geq1/\cosh(a)$, proving the claim.
\end{proof}

\begin{lem}[Success probability of the annealed Gibbs sampler]
\label{lem:success-probability}
Let $\beta>0$, $0<\epsilon<1/2$, and $\Delta=\max E-\min E$, with $e^{\beta\Delta}\leq\alpha$. We use one common grid, $n_{\mathrm{ann}}=\max\{1,\lceil(\log\alpha)^2/4\rceil\}$ steps, and $\beta_k=k\beta/n_{\mathrm{ann}}$. At each step, suppose the ground-state vector error and the filter operator error are at most $\tilde\epsilon=\epsilon/(64 n_{\mathrm{ann}})$. Then one complete annealing attempt succeeds with probability at least
\begin{equation}
p_{\mathrm{success}}\geq (e^{-1/8}-\epsilon/32)^2>1/2.
\end{equation}
Conditional on success, its normalized final state is within $\epsilon/8$ of the grid Gibbs state. Every successful prefix is a constant-overlap warm start for the next filter.
\end{lem}

\begin{proof}
We write $\Pi_k=\ket{\psi_k}\bra{\psi_k}$ for the exact grid Gibbs projector and $F_k$ for the successful branch of the implemented filter. These are contractions, and the hypotheses give $\norm{F_k-\Pi_k}\leq2\tilde\epsilon$. The ideal successful branch is $\Pi_{n_{\mathrm{ann}}}\cdots\Pi_1\psi_0=c_{n_{\mathrm{ann}}}\psi_{n_{\mathrm{ann}}}$, where \cref{lem:annealing-overlap-bound} and $\log\cosh x\leq x^2/2$ imply
\begin{equation}
c_{n_{\mathrm{ann}}}^2\geq\cosh\left(\frac{\beta\Delta}{4 n_{\mathrm{ann}}}\right)^{-2 n_{\mathrm{ann}}}
\geq e^{-(\beta\Delta)^2/(16 n_{\mathrm{ann}})}\geq e^{-1/4}.
\end{equation}
The last inequality follows from $\beta\Delta\leq\log\alpha$ and $n_{\mathrm{ann}}\geq(\log\alpha)^2/4$. The contraction-product identity gives
\begin{equation}
\norm{F_{n_{\mathrm{ann}}}\cdots F_1-\Pi_{n_{\mathrm{ann}}}\cdots\Pi_1}
\leq\sum_{k=1}^{n_{\mathrm{ann}}}\norm{F_k-\Pi_k}
\leq2 n_{\mathrm{ann}}\tilde\epsilon=\epsilon/32.
\end{equation}
Thus the actual unnormalized branch $y=F_{n_{\mathrm{ann}}}\cdots F_1\psi_0$ has success probability $\norm y^2\geq(c_{n_{\mathrm{ann}}}-\epsilon/32)^2$. Its normalized error satisfies
\begin{equation}
\norm{y/\norm y-\psi_{n_{\mathrm{ann}}}}\leq\frac{2\norm{y-c_{n_{\mathrm{ann}}}\psi_{n_{\mathrm{ann}}}}}{c_{n_{\mathrm{ann}}}}
\leq \frac{e^{1/8}\epsilon}{16}<\epsilon/8.
\end{equation}
The same estimates hold for every prefix. Since $\beta\Delta/(4 n_{\mathrm{ann}})\leq\log\alpha/(4 n_{\mathrm{ann}})\leq1/2$, its overlap with the next ground state is at least $1/\cosh(1/2)-\epsilon/8-\tilde\epsilon>1/2$. This proves both the finite success guarantee and the warm-start claim, including $\beta\Delta\leq1$ and $\Delta=0$.
\end{proof}

\begin{thm}[Query complexity of annealed Gibbs sampler]
\label{thm:annealing-main}
Let $p\in\mathscr W^s(\alpha,\xi,\rho)$, and use the common grid and error tolerance of \cref{lem:success-probability}. The annealing schedule is
\begin{equation}
\label{eq:annealing-schedule}
\beta_k = k\delta,
\qquad
k=0,\dots,n_{\mathrm{ann}},
\qquad
\delta=\beta/n_{\mathrm{ann}},
\qquad
n_{\mathrm{ann}}=\max\{1,\lceil(\log\alpha)^2/4\rceil\}.
\end{equation}
One complete attempt uses a total of
\begin{equation*}
\mathcal O\left(\sqrt\alpha\,n_{\mathrm{ann}}(N\sqrt d+\beta R)\log\frac1{\tilde\epsilon}\right)
\end{equation*}
queries to the block encodings of $\mathbb W_{N,\beta_k}$, where $R=\max_x\norm{\nabla E(x)}$. Restarting \cref{alg:cap} from the uniform state on failure, for at most $\lceil\log_2(2/\epsilon)\rceil$ attempts, multiplies this bound by at most that number.
\end{thm}

\begin{proof}
Each intermediate temperature belongs to the same potential class. By \cref{cor:threshold-choice}, its physical singular-value threshold width is a constant times $1/\sqrt{\alpha\beta_k}$. Combining this with the block normalization from \cref{thm:block-encoding-wn} gives the per-step query bound
\begin{equation}
\label{eq:per-step-query-cost}
\mathcal O\left(\gamma_k\sqrt{\alpha\beta_k}\log\frac1{\tilde\epsilon}\right)
=\mathcal O\left(\sqrt\alpha(N\sqrt d+\beta_k R)\log\frac1{\tilde\epsilon}\right).
\end{equation}
Summing over $k\leq n_{\mathrm{ann}}$ and using $\beta_k\leq\beta$ proves the attempt bound. Each block call uses two gradient queries, so the oracle-query order is unchanged. The restart cap gives the final assertion.
\end{proof}

\begin{figure}[t]
\centering
\begin{quantikz}[
    font=\small,
    column sep=0.35cm, row sep={0.65cm,between origins},
    /tikz/meter/.append style={minimum width=0.45cm, minimum height=0.38cm, inner sep=0pt,
        path picture={
            \draw ([shift={(0.05,0.05)}]path picture bounding box.south west)
                to[out=50,in=130] ([shift={(-0.05,0.05)}]path picture bounding box.south east);
            \draw ([yshift=0.05cm]path picture bounding box.south) -- ++(60:0.22cm);
        }}
]
\lstick{$|0\rangle$} & \gate[6][1.35cm]{\mathrm{GbGS}(\beta_1)} & \meter[][0.45cm][0.38cm]{} & \setwiretype{n} &[0.5cm] \lstick{$|0\rangle$}\setwiretype{n} & \gate[6][1.35cm]{\mathrm{GbGS}(\beta_2)}\setwiretype{q} & \meter[][0.45cm][0.38cm]{} &[0.3cm] \setwiretype{n} & &[0.5cm] \lstick{$|0\rangle$}\setwiretype{n} & \gate[6][1.35cm]{\mathrm{GbGS}(\beta_f)}\setwiretype{q} & \meter[][0.45cm][0.38cm]{} & \setwiretype{n} \\
\lstick{$|0\rangle$} & & \meter[][0.45cm][0.38cm]{} & \setwiretype{n} & \lstick{$|0\rangle$}\setwiretype{n} & \setwiretype{q} & \meter[][0.45cm][0.38cm]{} & \setwiretype{n} & & \lstick{$|0\rangle$}\setwiretype{n} & \setwiretype{q} & \meter[][0.45cm][0.38cm]{} & \setwiretype{n} \\
\lstick{$|0\rangle$} & & \meter[][0.45cm][0.38cm]{} & \setwiretype{n} & \lstick{$|0\rangle$}\setwiretype{n} & \setwiretype{q} & \meter[][0.45cm][0.38cm]{} & \setwiretype{n} & & \lstick{$|0\rangle$}\setwiretype{n} & \setwiretype{q} & \meter[][0.45cm][0.38cm]{} & \setwiretype{n} \\
\lstick[3]{$|\psi\rangle$} & & \qw & \qw & & \qw & \qw & |[inner xsep=4pt]| \cdots & \qw & & & \meter[][0.45cm][0.38cm]{} & \setwiretype{n} \\
& & \qw & \qw & & \qw & \qw & |[inner xsep=4pt]| \cdots & \qw & & & \meter[][0.45cm][0.38cm]{} & \setwiretype{n} \\
& & \qw & \qw & & \qw & \qw & |[inner xsep=4pt]| \cdots & \qw & & & \meter[][0.45cm][0.38cm]{} & \setwiretype{n}
\end{quantikz}
\caption{Annealing schedule consisting of repeated applications of the QSVTh protocol \cite{leng2026operator} with increasing target inverse temperatures. Each block GbGS($\beta$) represents application of the algorithm in \cref{alg:GbGS} with target temperature $\beta$ and $\beta_f$ represents the target inverse temperature.}
\label{fig:anneal-schedule}
\end{figure}
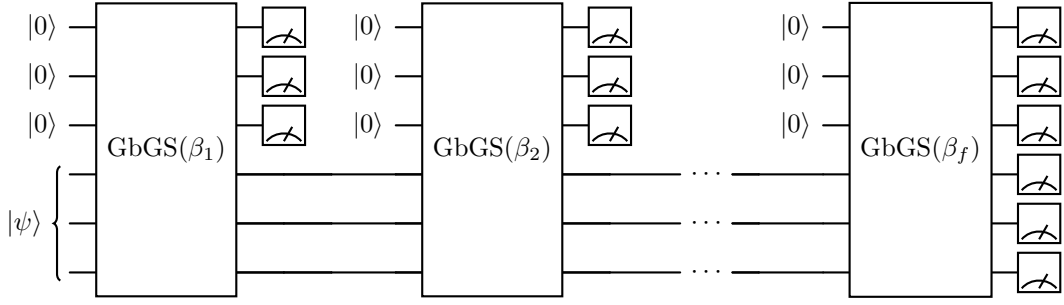

\begin{thm}
\label{thm:main-theorem-quantum-app}
Let $E:\mathbb{T}^{d}\rightarrow\mathbb{R}$ be a Gibbs potential in $\mathscr{W}^s(\alpha,\xi,\rho)$ as defined in \cref{def:potential-class}, with $s \geq 1$, and $p(x)= e^{-\beta E(x)}/Z$ the corresponding Gibbs state.
There exists a quantum algorithm that returns a sample from a random variable $X\sim\eta$ after $N_{q}$ queries to the quantum gradient oracle $O_{\nabla E}$ (or its inverse), such that $\mathrm{TV}(\eta, p) \le \eps$, where
\begin{align}
\begin{aligned}
\label{eq:appendix-query-complexity}
N_q \in \tilde{\mathcal O}\left(\sqrt{\alpha d} \left(1 + \frac{d + \xi}{\rho}
\left[ds + \log \frac{\alpha(1 + \xi/\rho)}{\epsilon}\right]^s\right) \log \frac1\epsilon\right).
\end{aligned}
\end{align}
\end{thm}

\begin{proof}
The schedule of \cref{alg:cap} has $n_{\mathrm{ann}}\in\mathcal O(1+\log^2\alpha)$, and we choose $\tilde\epsilon=\epsilon/(64 n_{\mathrm{ann}})$ on one common grid. By \cref{lem:success-probability}, a successful attempt has grid-state error at most $\epsilon/8$. Applying the resolution booster of \cref{lem:resolution-booster} via \cref{cor:gevrey-sampling} at precision $\epsilon/2$ therefore gives conditional sampling error at most $\epsilon/2$. We restart from the uniform state on failure, for at most $\lceil\log_2(2/\epsilon)\rceil$ attempts. The residual failure probability is at most $\epsilon/2$, whose contribution to the total TV error is at most $\epsilon/2$, giving error at most $\epsilon$. Combining \cref{thm:annealing-main,thm:block-encoding-wn}, and absorbing the restart logarithm into the tilde, gives
\begin{equation}
N_q \in \tilde{\mathcal O}\left(\sqrt\alpha (1 + \log^2\alpha) (N\sqrt d + \beta R) \log(1/\epsilon)\right),
\end{equation}
for
\begin{align}
N \in \tilde\Omega\left(\frac{d + \xi}{\rho} \left[ds + \log \frac{\alpha(1 + \xi/\rho)(1 + \log^2\alpha)}{\epsilon}\right]^s\right).
\end{align}
We retain both the differentiation-validity and accuracy contributions to $N$. Choosing a positive integer cutoff of this order gives
\begin{equation}
\label{eq:appendix-query-complexity-intermediate}
N_q \in \tilde{\mathcal O}\left(\sqrt\alpha (1 + \log^2\alpha)
\left[\sqrt d \left(1 + \frac{d + \xi}{\rho} \left[ds + \log \frac{\alpha(1 + \xi/\rho)}{\epsilon}\right]^s\right) + \beta R\right] \log \frac1\epsilon\right).
\end{equation}
By class membership \eqref{eq:potential-class}, we have
\begin{equation}
\label{eq:r-and-betadelta-bound}
R = \max_{x \in \mathbb{T}^{d}}\Vert{\nabla E} (x)\Vert\leq \sqrt{d}\xi/\beta\rho,\quad \beta \Delta \leq \log \alpha,
\end{equation}
and $\xi/\rho$ is bounded by the cutoff term since $ds \geq 1$. Absorbing the annealing logarithms into the $\tilde{\mathcal O}$ notation gives
\begin{align}
\begin{aligned}
\label{eq:final-query-complexity}
  N_q \in \tilde{\mathcal O}\left(\sqrt{\alpha d} \left(1 + \frac{d + \xi}{\rho}
  \left[ds + \log \frac{\alpha(1 + \xi/\rho)}{\epsilon}\right]^s\right) \log \frac1\epsilon\right).
 \end{aligned}
\end{align}
This is the claimed query bound.
\end{proof}

\section{Details of the Lower Bound Results}
\label{app:lb-sec}

\subsection{Formal query model and proofs of the lower bound}
\label{app:model}

In this appendix, we collect the formal definitions summarized in \cref{ssec:hide-seek}: the oracle and algorithm model (\cref{def:query-model,def:sampling-algorithm}), bump distributions (\cref{def:bump}), and the null-run coupling (\cref{lem:null-coupling}). We also give the proof of \cref{thm:hide-and-seek-lower-bound} and a remark on its validity for general local oracles, beyond derivative queries.

A function $f$ on the $d$-dimensional torus $\T^d=\R^d/(2\pi\Z)^d$ is the same object as a $2\pi$-periodic function $\bar f$ on $\R^d$ (its \emph{periodic lift}), and integration over $\T^d$ is Lebesgue integration over the cube, therefore $\vol(\T^d)=(2\pi)^d$ and $p(A)=\int_{A}{\bar p}(x)\,dx$ for $A\subseteq [0,2\pi)^d$.

\begin{defn}[Gibbs state on $\T^d$]
\label{def:gibbs-state-torus}
A \emph{Gibbs state} on $\T^d$ with potential $E:\T^d\to\R$ and inverse temperature $\beta>0$ is the probability distribution with density
\begin{equation}
\label{eq:gibbs-density-formal}
    p = \frac{e^{-\beta E}}Z,\qquad Z=\int_{[0,2\pi)^d}e^{-\beta\bar{E}(x)}\,dx.
\end{equation}
\end{defn}
\noindent
A sampling algorithm accesses the instance through an \emph{oracle of order $k$}, $k\in\N \cup \{\infty\}$: queried at a point $x\in\T^d$, it answers with the value and the derivatives, up to order $k$, of the log-density at $x$.
\begin{defn}[Oracle of order $k$]
\label{def:query-model}
Let $k\in\N\cup\{\infty\}$ and let $\mathscr{P}$ be a class of probability distributions on $\T^d$ whose densities have positive periodic lifts ${\bar p}\in\mathscr{C}^{k}(\R^d,(0,\infty))$. The \emph{oracle of order $k$} is the map assigning to an instance $p\in\mathscr{P}$ and a query point $x\in\T^d$ the answer
\begin{equation}
\label{eq:response-map}
    \phi_k(p,x) = \bigl(\nabla^{\otimes j}\log{\bar p}(x)\bigr)_{j=0}^{k},
\end{equation}
the value and all derivatives up to order $k$ of the log-density, evaluated at any representative of $x$.
\end{defn}
\noindent
The answer \eqref{eq:response-map} is well defined, since derivatives of $2\pi$-periodic functions are $2\pi$-periodic. Its zeroth-order entry is $-\beta E(x)-\log Z$, which evaluation of a specified potential alone need not supply. Lower bounds consume the property of \eqref{eq:response-map} that on a locally flat landscape the answer degenerates to the flat answer \eqref{eq:null-response},
\begin{equation}
\label{eq:flat-observation}
    p=\kappa\ \text{on an open neighbourhood of $x$}
    \quad\Longrightarrow\quad
    \phi_k(p,x)=(\log\kappa, 0, \dots, 0)={\phi_{\mathrm{null}}},
\end{equation}
because $\log{\bar p}\equiv\log\kappa$ on a neighbourhood of every representative of $x$.

\begin{defn}[Sampling algorithm]
\label{def:sampling-algorithm}
A \emph{sampling algorithm} $\mathcal{A}$ with access to an oracle $\phi$ consists of a random seed $\omega\sim\mathbb{P}$, drawn independently of the instance of sampling, together with \emph{query maps} $Q_n$ and \emph{output maps} $\Psi_n$, $n\geq1$. A run of $\mathcal{A}$ on the instance $p\in\mathscr{P}$ interleaves queries and oracle answers,
\begin{equation}
\label{eq:run}
    X_1=Q_1(\omega),\qquad O_n=\phi(p,X_n),\qquad X_{n+1}=Q_{n+1}(\omega,O_1,\dots,O_n).
\end{equation}
The \emph{output after $n$ queries} is the point $\hat X_n\coloneqq\Psi_n(\omega,O_1,\dots,O_n)\in\T^d$; its distribution over the seed $\omega$ is the \emph{output law} ${\hat p_n}$.
\end{defn}
\noindent
Note that in \cref{def:sampling-algorithm} the initialization $X_1=Q_1(\omega)$ depends only on the seed $\omega\sim\mathbb{P}$. Importantly, any initialization which depends on $p\in\mathscr{P}$, such as a warm start, is excluded.

\begin{defn}[Bump distribution]
\label{def:bump}
Let $\Om\subset\T^d$ be closed with $v_0\coloneqq\vol(\Om)\in{(0,(2\pi)^d)}$. A \emph{bump} on the patch $\Om$ with \emph{concentration} $\delta\in(0,1)$ and \emph{background} $\kappa>0$ is a probability distribution $p$ with density satisfying \eqref{eq:offpatch-density}: $p(\Om)=\delta>{v_0/(2\pi)^d}$ and $p(x)=\kappa$ for every $x\notin\Om$. Necessarily $\kappa={(1-\delta)/\bigl((2\pi)^d-v_0\bigr)}$.
\end{defn}
\noindent
Fix a bump ${p_0}$ on $\Om$ and its translates ${p_g}(A)={p_0}(A-g)$ with patches $\Omega_g=g+\Om$, $g\in\T^d$; all are bumps with the same concentration $\delta$ and background $\kappa$, by translation invariance of the volume. Given an algorithm $\mathcal{A}$ and the flat answer ${\phi_{\mathrm{null}}}$ of \eqref{eq:null-response}, define the \emph{null run} and \emph{null outputs} of $\mathcal{A}$ by feeding it ${\phi_{\mathrm{null}}}$ in place of every oracle answer:
\begin{equation}
\label{eq:null-run}
X^{\mathrm{null}}_{n} = Q_{n}\bigl(\omega,{\phi_{\mathrm{null}}},\dots,{\phi_{\mathrm{null}}}\bigr),\qquad
\hat X^{\mathrm{null}}_{n} = \Psi_{n}\bigl(\omega,{\phi_{\mathrm{null}}},\dots,{\phi_{\mathrm{null}}}\bigr),
\end{equation}
a single instance-independent random process, determined by the seed $\omega$ alone.

For potential-value and derivative oracles, we fix a Gibbs representation of $p_0$ and choose $E_g(x)=E_0(x-g)$. These potentials have a common additive gauge, partition function, and off-patch value, so the null-run coupling below also applies to these oracles on the translated hard family.

\begin{lem}[Null coupling]
\label{lem:null-coupling}
Fix $g\in\T^d$, let $X_n,\hat X_n$ be the run and outputs of $\mathcal{A}$ on the instance ${p_g}$ \eqref{eq:run}, and let
\begin{equation}
    \tau_g(\omega) \coloneqq \inf\{n\geq 1 \,\vert\, X^{\mathrm{null}}_n(\omega) \in \Omega_g\},\qquad\inf\varnothing=\infty,
\end{equation}
be the first time the null run hits the patch $\Omega_g$. Then, for every seed $\omega$,
\begin{equation}
\label{eq:pathwise-coupling}
X_n(\omega) = X^{\mathrm{null}}_n(\omega),\quad \forall\, n \leq \tau_g(\omega);\qquad
\hat X_n(\omega) = \hat X^{\mathrm{null}}_n(\omega),\quad \forall\, n < \tau_g(\omega).
\end{equation}
\end{lem}

\begin{proof}
If $x\notin\Omega_g$ then, the patch being closed, ${p_g}=\kappa$ on a neighbourhood of $x$, hence $\phi_k({p_g},x)={\phi_{\mathrm{null}}}$ by \eqref{eq:flat-observation}. Fix $\omega$ and induct on $n$: $X_1=Q_1(\omega)=X^{\mathrm{null}}_1$. If $X_j=X^{\mathrm{null}}_j$ for all $j\leq n$ and $n<\tau_g(\omega)$, then $X_j\notin\Omega_g$ for $j\leq n$, so the answers received by the two runs coincide, $O_j=\phi_k({p_g},X_j)={\phi_{\mathrm{null}}}$ for $j\leq n$; hence $X_{n+1}=X^{\mathrm{null}}_{n+1}$, proving the first identity in \eqref{eq:pathwise-coupling}. For $n<\tau_g(\omega)$ the same answer gives $\hat X_n=\Psi_n(\omega,O_1,\dots,O_n)=\hat X^{\mathrm{null}}_n$, the second identity in \eqref{eq:pathwise-coupling}.
\end{proof}

\subsection{Proof of the hide-and-seek lower bound}
\label{prf:hide-and-seek-lower-bound}%

\begin{proof}[Proof of \cref{thm:hide-and-seek-lower-bound}]
Fix an algorithm $\mathcal{A}$ with seed $\omega\sim\mathbb{P}$ and query count $N$, and suppose its output law satisfies $\mathrm{TV}({\hat p_{g,N}},{p_g})\leq\eps$ for every $g\in\T^d$; we show $N\geq{(2\pi)^d(\delta-\eps)/v_0}-1$. Write $\hat X^g_N$ for the output of the run on ${p_g}$ and $\tau_g$ for the hitting time of the coupling \cref{lem:null-coupling}.

\smallskip\noindent
\emph{Step 1 (accuracy proxy: mass on the patch).} By the definition of total variation, evaluated on the set $\Omega_g$,
\begin{equation}
\label{eq:eps-to-masses}
    \eps \geq \mathrm{TV}({\hat p_{g,N}},{p_g}) \geq {p_g}(\Omega_g)-{\hat p_{g,N}}(\Omega_g) = \delta-\mathbb{P}\bigl(\hat X^g_N\in\Omega_g\bigr).
\end{equation}

\smallskip\noindent
\emph{Step 2 (null coupling).} Split the event $\{\hat X^g_N\in\Omega_g\}$ according to whether the null run has hit the patch by time $N$,
\begin{equation}
    \{\hat X^g_N\in\Omega_g\} = \{\hat X^g_N\in\Omega_g, \tau_g>N\} \cup \{\hat X^g_N\in\Omega_g, \tau_g\leq N\}.
\end{equation}
On $\{\tau_g>N\}$, \cref{lem:null-coupling} gives $\hat X^g_N=\hat X^{\mathrm{null}}_N$, hence
\begin{equation}
    \{\hat X^g_N\in\Omega_g\} = \{\hat X^{\mathrm{null}}_N\in\Omega_g, \tau_g>N\} \cup \{\hat X^g_N\in\Omega_g, \tau_g\leq N\} \subseteq \{\hat X^{\mathrm{null}}_N\in\Omega_g\} \cup \{\tau_g\leq N\},
\end{equation}
and it follows that
\begin{equation}
    \mathbb{P}\bigl(\hat X^g_N\in\Omega_g\bigr)
    \leq \mathbb{P}\bigl(\tau_g\leq N\bigr)+\mathbb{P}\bigl(\hat X^{\mathrm{null}}_N\in\Omega_g\bigr),
\end{equation}
and by the union bound over $\{\tau_g\leq N\}=\cup_{n\leq N}\{X^{\mathrm{null}}_n\in\Omega_g\}$,
\begin{equation}
\label{eq:eps-to-probs}
    \eps \geq \delta-\sum_{n=1}^N\mathbb{P}\bigl(X^{\mathrm{null}}_n\in\Omega_g\bigr)-\mathbb{P}\bigl(\hat X^{\mathrm{null}}_N\in\Omega_g\bigr).
\end{equation}

\smallskip\noindent
\emph{Step 3 (average over a uniform $g$).} Inequality \eqref{eq:eps-to-probs} holds for every $g$; averaging it over a uniform draw $g$ from $\T^d$ eliminates this dependence. For any fixed point $x\in\T^d$, the probability that $x \in \Omega_g$ follows from identity \eqref{eq:coverage-identity}, that is: $\vol(x-\Om)=v_0${, hence $\mathbb{P}_g(x\in\Omega_g)=v_0/(2\pi)^d$ for $g$ uniform on the torus of volume $(2\pi)^d$}. We denote $\mathbb{P}_g, \mathbb{E}_g$ the probability and expectation value over $g \sim \mathrm{unif}(\T^d)$, to distinguish from the algorithm's randomness $\omega \sim \mathbb{P}$ and the corresponding expectation $\mathbb{E}$. Apply $\mathbb{P}_g(x\in\Omega_g)= {v_0/(2\pi)^d}$ at the random, $g$-independent points $x=X^{\mathrm{null}}_n(\omega)$ and $x=\hat X^{\mathrm{null}}_N(\omega)$,
\begin{align}
\label{eq:avg-hits}
    \begin{aligned}
       &\mathbb{E}_g\Bigl[\sum_{n=1}^N \mathbb{P}\bigl(X^{\mathrm{null}}_n\in \Omega_g\bigr)\Bigr] = \sum_{n=1}^N \mathbb{E}\bigl[\mathbb{P}_g\bigl(X^{\mathrm{null}}_n\in \Omega_g\bigr)\bigr] = {\frac{N\,v_0}{(2\pi)^d}},\\
       &\mathbb{E}_g\bigl[\mathbb{P}\bigl(\hat X^{\mathrm{null}}_N\in \Omega_g\bigr)\bigr] = \mathbb{E}\bigl[\mathbb{P}_g\bigl(\hat X^{\mathrm{null}}_N\in \Omega_g\bigr)\bigr] = {\frac{v_0}{(2\pi)^d}}.
    \end{aligned}
\end{align}
Averaging \eqref{eq:eps-to-probs} therefore yields $\eps\geq\delta-(N+1){v_0/(2\pi)^d}$, i.e.
\begin{equation}
    N \geq {\frac{(2\pi)^d(\delta-\eps)}{v_0}}-1.
\end{equation}

\smallskip\noindent
\emph{Step 4 (monotonicity transfers bounds).} The argument shows that any algorithm that is $\eps$-accurate on \emph{every} translate after $N$ queries obeys the bound; even more so this holds for any algorithm meeting the worst-case guarantee $\sup_{g}\mathrm{TV}({\hat p_{g,N}},{p_g})<\eps$ of \eqref{eq:query-cost}, hence $N_{\mathcal{A}}(\{{p_g}\},\eps)\geq{(2\pi)^d(\delta-\eps)/v_0}-1$ for every $\mathcal{A}$, and $\Ncl(\{{p_g}\},\eps)\geq{(2\pi)^d(\delta-\eps)/v_0}-1$ by \cref{def:classical-query-complexity}. Finally, the monotonicity of query complexity with respect to inclusions transfers the bound to every $\mathscr{P}\supseteq\{{p_g}\}_{g\in\T^d}$. This proves the theorem for the order-$\infty$ oracle. The case of finite $k$ follows a fortiori, since an order-$k$ algorithm is an order-$\infty$ algorithm ignoring the higher entries of each answer.

This completes the proof of \cref{thm:hide-and-seek-lower-bound}.
\end{proof}

We provide a few remarks before concluding the subsection.

\begin{rmk} The property \eqref{eq:flat-observation} consumed by the coupling \cref{lem:null-coupling}, and the lower bound that follows, holds for a wider class of oracles than order-$k$, i.e. \eqref{eq:response-map}. Indeed, it holds for any answers $\phi(p,x)$ to instances and query points which are \emph{local}~\cite{braun2017lower}, i.e. that depend only on the density near that point: for all ${p,p'}\in\mathscr{P}$ and all $x\in\T^d$,
\begin{equation}
\label{eq:locality}
    {p=p'}\ \text{on some neighbourhood}\ B(x,r),\ r>0
    \quad\Longrightarrow\quad
    \phi(p,x)=\phi(p',x).
\end{equation}
\end{rmk}

\begin{rmk}
\label{rem:locally-determined}%
Every derivative oracle \eqref{eq:response-map} is local. Limits of convolutions against kernels supported on neighbourhoods of vanishing radius around $x$ are also local whenever they exist, since densities agreeing near $x$ give eventually identical convolutions. However, returning the entire sequence of finite-radius convolutions need not be local. For a local oracle, the flat answer at $x$ as ${\phi_{\mathrm{null}}}(x)\coloneqq\phi({p_g},x)$ occurs for every translate $g$ such that $x \notin \Omega_g$. Every $x$ admits such $g$, because $v_0<{(2\pi)^d}$ leaves points $y\notin\Om$ and $g=x-y$ is such a translate. The value of the \emph{null} answer does not depend on $g$: since for $x\notin\Omega_g$ and $x\notin\Omega_h$ both densities equal $\kappa$ on a common neighbourhood of $x$, then $\phi({p_g},x)=\phi({p_h},x)$ by \eqref{eq:locality}. With ${\phi_{\mathrm{null}}}(X^{\mathrm{null}}_n)$ in place of ${\phi_{\mathrm{null}}}$ in \eqref{eq:null-run}, \cref{lem:null-coupling} and the proof of \cref{thm:hide-and-seek-lower-bound} hold verbatim: the lower bound \eqref{eq:compact-lower-bound} applies to every local oracle.
\end{rmk}

\begin{rmk}
For comparison, we consider ideal membership search on a tiling of $\T^d$ by $M$ equal-volume cells, exactly one of which is marked. Probing $M-1$ cells identifies the marked cell, matching the formal value of \eqref{eq:compact-lower-bound} at $\delta=1$ and $\eps=0$. This search illustration lies outside the Gibbs family with positive background, for which $\delta<1$.
\end{rmk}

\subsection{Gevrey estimates for the bump potential}
\label{app:gevrey}

In this appendix, we give the constructions and the analysis behind \cref{ssec:bumps}. First, the mollified profile entering \cref{def:bump-torus}, given in~\cref{def:mollified-profile}, and the Gevrey property of the mollifier $\psi_s$ with explicit constants (\cref{prop:psi-gevrey-property}). Then, the elementary properties of the profile $\chi_{\nu,\ell}$ (\cref{lem:bump-profile,lem:fill}) and the proof of \cref{prop:bump-family}. Next, $\mathscr{W}^s$-class membership of the bump's potential is worked out in \cref{thm:bump-potential-gevrey-seminorm}, exploiting \cref{thm:gevrey-constants-f} and related technical lemmas. In conclusion, we give the proof of \cref{cor:lower-bounds-gevrey}, which consumes $\mathscr{W}^s$-class parameters, and prove the exponentiation results \cref{lem:exp-prod}, \cref{cor:exp-abs} and \cref{cor:class-reg}, transferring Gevrey-class parameters of the potential to the density.

\subsubsection{Construction of the bump family}

\begin{defn}[Mollified step and two-level profile]
\label{def:mollified-profile}
For $s>1$, let ${\psi_s}(x)=e^{-(1-x^{2})^{-1/(s-1)}}$ for $\vert x\vert<1$ and ${\psi_s}(x)=0$ otherwise. The mollified step of width $\nu>0$ is
\begin{equation}
\label{eq:mollified-step-def}
    \theta_{\nu}(x) \coloneqq \theta_1(x/\nu),\quad \theta_1(x) \coloneqq \frac{\int_{-1}^{x}\!dy\,{\psi_s}(y)}{Z_s},\quad {Z_s} = \int_{-1}^{1}\! dy \, {\psi_s}(y),
\end{equation}
and its translate $H_{\nu}(x)\coloneqq \theta_{\nu}(x-\nu)=\theta_1(x/\nu-1)$. For $0<\nu\leq\ell/4$ the mollified two-level profile concentrated inside $[0,\ell]^d$ is
\begin{equation}
\label{eq:chi-profile-def}
    \chi_{\nu,\ell}(x) \coloneqq \prod_{k=1}^d H_{\nu}(x_k)\,H_{\nu}(\ell-x_k):\,\R^d\to[0,1].
\end{equation}
\end{defn}
\noindent
The step $\theta_1$ rises smoothly from $0$ to $1$ across $[-1,1]$, so each factor $H_{\nu}(y)H_{\nu}(\ell-y)$ vanishes off $[0,\ell]$, equals $1$ on the plateau $[2\nu,\ell-2\nu]$, and ramps monotonically in between (\cref{lem:bump-profile} below). The construction is based on ${\psi_s}$ in order to exploit its $s$-Gevrey regularity (\cref{prop:psi-gevrey-property}); for the calculations of this appendix it is convenient to write $t\coloneqq1/(s-1)$, so that $s>1$ corresponds to $t>0$. We use the definition of Gevrey class functions on $\T^d$ recalled in \cref{def:gevrey-function}.

\begin{prop}[$s$-Gevrey property of $\psi_s$]
\label{prop:psi-gevrey-property}
    For any $s>1$, the function ${\psi_s}:\mathbb{R}\to\mathbb{R}$
\begin{equation}
\label{eq:psi-def}
{\psi_s}(x) =
\begin{cases}
\exp\!\left(-1/(1-x^{2})^{t}\right), & |x|<1,\\
0, & |x|\geq 1,
\end{cases}
\end{equation}
belongs to $\mathscr{G}^{s}(C,r,\R)$, with uniform constants in $\R$: amplitude $C =1$ and radius
 \begin{equation}
     \label{eq:psi-gevrey-radius}
     r=\frac{\lambda (t\cos(2t \arcsin \lambda))^{\frac{1}t}}{2(\lambda+1)^2},
 \end{equation}
 for any choice $0<\lambda< \sin \left(\min\left \{  \frac{\pi}{2},\frac{\pi}{4t}\right \}\right)$.
\end{prop}
\begin{proof}
On $(-1,1)$, each derivative of ${\psi_s}$ is $e^{-(1-y^2)^{-t}}$ times a finite sum of terms $c\,y^p(1-y^2)^{-q}$ with $p\in\N$ and $q\geq0$. Exponential decay therefore gives
\begin{equation*}
\lim_{y\to\pm 1^{\mp}}{\psi_s}^{(k)}(y)=0,
\end{equation*}
matching ${\psi_s}^{(k)}\equiv0$ off $(-1,1)$ and proving ${\psi_s}\in\mathscr C^\infty(\R)$.

Since $\{{\psi_s}\neq 0\}=(-1,1)$ and ${\psi_s}(y)={\psi_s}(-y)$, it is sufficient to prove the Gevrey property for $y\in[0,1)$. The complex continuation ${\psi_s}(z)$ is holomorphic on $\mathbb{C}\setminus\{(-\infty,-1]\cup[1,+\infty)\}$ (branch points $z=\pm 1$); the Cauchy formula applies on the contour $\gamma=\{y+\eta e^{i\theta},\ \theta\in[0,2\pi)\}$ and for a contour radius $0<\eta<1-y$:
\begin{align}
\begin{aligned}
\label{eq:cauchy-estimate}
\lvert{\psi_s}^{(k)}(y)\rvert
&\leq\left\lvert k!\oint_{\gamma}\frac{dz}{2\pi i}\,\frac{e^{-\frac{1}{(1-z^{2})^{t}}}}{(z-y)^{k+1}}\right\rvert
=\frac{k!}{2\pi \eta^{k}}\left\lvert\int_{0}^{2\pi}d\theta\,\frac{e^{-\frac{1}{(1-(y+\eta e^{i\theta})^{2})^{t}}}}{e^{i\theta k}}\right\rvert\\
&\leq\frac{k!}{2\pi \eta^{k}}\int_{0}^{2\pi}d\theta\,\left\lvert e^{-\frac{1}{(1-(y+\eta e^{i\theta})^{2})^{t}}}\right\rvert
\leq\frac{k!}{\eta^{k}}\,e^{-\inf_{\theta}\Re\left(\frac{1}{[1-(y+\eta e^{i\theta})^{2}]^{t}}\right)}.
\end{aligned}
\end{align}
We lower-bound $\inf_{\theta}\Re\!\left(\tfrac{1}{[1-(y+\eta e^{i\theta})^{2}]^{t}}\right)$, then maximize over $y\in[0,1)$. Parameterize $\eta=\lambda(1-y)$, $0<\lambda<1$, and $w=y+\lambda(1-y)e^{i\theta}$:
\begin{equation}\label{eq:re-expr}
\begin{aligned}
\Re\!\left(\frac{1}{[1-(y+\eta e^{i\theta})^{2}]^{t}}\right)
&=\Re\!\left[\frac{1}{(1-w^{2})^{t}}\right]
=\frac{\cos\!\left(\arg[1/(1-w^{2})^{t}]\right)}{\lvert(1-w)^{t}(1+w)^{t}\rvert}\\
&=\frac{\cos\!\left(-t\arg(1-w^{2})\right)}{\lvert(1-w)^{t}\rvert\,\lvert(1+w)^{t}\rvert}
=\frac{\cos\!\left(t\arg(1-w^{2})\right)}{\lvert1-w\rvert^{t}\,\lvert1+w\rvert^{t}}.
\end{aligned}
\end{equation}
We bound numerator and denominator of \eqref{eq:re-expr} separately. The point $1\mp w$ moves on a circle of radius $\lambda(1-y)$ centered at $1\mp y$; the law of sines in the triangles $\{0, 1-w, 1-y\}$ and $\{0, 1+w, 1+y\}$, maximized at a right angle at $1\mp w$, gives $\sin\lvert\arg(1-w)\rvert\leq\lambda$ and $\sin\lvert\arg(1+w)\rvert\leq\lambda\tfrac{1-y}{1+y}\leq\lambda$. Since $\Re(1\pm w)>0$, with the branch $\arg\in(-\pi,\pi]$, $\lvert\arg(1\pm w)\rvert\leq\arcsin\lambda<\tfrac\pi2$, and $\lvert\arg(1-w^{2})\rvert\leq\lvert\arg(1-w)\rvert+\lvert\arg(1+w)\rvert$ yields
\begin{equation}\label{eq:arg-bound}
\lvert\arg(1-w^{2})\rvert\leq 2\arcsin\lambda.
\end{equation}
The assumed restriction $0<\lambda<\sin(\min\{\pi/2,\pi/(4t)\})$ gives $2t\arcsin\lambda<\pi/2$. Since $\cos$ is decreasing on $[0,\pi/2]$, \eqref{eq:arg-bound} yields
\begin{equation}\label{eq:cos-bound}
\cos(t\,\lvert\arg(1-w^{2})\rvert)=\cos(t\arg(1-w^{2}))\geq\cos(2t\arcsin\lambda).
\end{equation}
For the denominator, $(1-\lambda)(1-y)\leq\lvert1-w\rvert\leq(1+\lambda)(1-y)$ and $(1-\lambda)(1+y)\leq\lvert1+w\rvert\leq(1+\lambda)(1+y)$, by the triangle inequality and $\lambda(1-y)\leq\lambda(1+y)$; hence
\begin{equation}\label{eq:mod-bound}
(1-\lambda)^{2}(1-y^{2})\leq\lvert1-w^{2}\rvert\leq(1+\lambda)^{2}(1-y^{2}).
\end{equation}
Combining \eqref{eq:cos-bound}, \eqref{eq:mod-bound} in \eqref{eq:re-expr} yields
\begin{equation}\label{eq:inf-bound}
\inf_{w=y+\lambda(1-y)e^{i\theta}}\Re\!\left[\frac{1}{(1-w^{2})^{t}}\right]\geq\frac{\cos(2t\arcsin\lambda)}{(1+\lambda)^{2t}(1-y^{2})^{t}}.
\end{equation}
Substituting into \eqref{eq:cauchy-estimate} provides
\begin{equation}\label{eq:pw1}
\lvert{\psi_s}^{(k)}(y)\rvert\leq\frac{k!}{\lambda^{k}(1-y)^{k}}\,e^{-\frac{\cos(2t\arcsin\lambda)}{(1+\lambda)^{2t}(1-y^{2})^{t}}}.
\end{equation}

To uniformize over $y\in[0,1)$, we use $1-y^{2}<2(1-y)$:
\begin{equation}\label{eq:pw2}
\lvert{\psi_s}^{(k)}(y)\rvert\leq\frac{k!}{\lambda^{k}(1-y)^{k}}\,e^{-\frac{\cos(2t\arcsin\lambda)}{2^{t}(1+\lambda)^{2t}(1-y)^{t}}},
\end{equation}
whose right-hand side, for $k\geq1$, is maximized over $y\in[0,1)$ at
\begin{equation}\label{eq:ystar}
(1-y_{\ast})^{t}=\frac{t\cos(2t\arcsin\lambda)}{2^{t}k(1+\lambda)^{2t}}.
\end{equation}
The right-hand side of \eqref{eq:ystar} is positive and at most $t/2^t\leq1/(e\log2)<1$, so $y_{\ast}\in(0,1)$.
The uniform bound following from substitution of $y_{\ast}$ in \eqref{eq:pw2} is
\begin{equation}
\label{eq:psi-deriv-ystar-bound}
\lvert{\psi_s}^{(k)}(y)\rvert\leq k!\left[2\,\frac{(1+\lambda)^{2}}\lambda\left(\frac{k}{t\cos(2t\arcsin\lambda)}\right)^{1/t}\right]^{k}e^{-\frac{k}t}.
\end{equation}
By the symmetry $y\leftrightarrow -y$ it holds for any $\lvert y\rvert<1$, and since all derivatives vanish off $(-1,1)$, it holds in fact for all $y\in\mathbb{R}$. Finally, using $k!=\Gamma(1+k)\geq(k/e)^{k}$, the inequality takes the form of a $s$-Gevrey bound on derivatives
\begin{equation}
\label{eq:psi-deriv-gevrey-bound}
\lvert{\psi_s}^{(k)}(y)\rvert\leq(k!)^{1+\frac{1}t}\left(2\,\frac{(1+\lambda)^{2}}\lambda\left(\frac{1}{t\cos(2t\arcsin\lambda)}\right)^{\frac{1}t}\right)^{k}=(k!)^{s}\,r^{-k},
\end{equation}
with $r$ as in \eqref{eq:psi-gevrey-radius} and amplitude $1$, the case $k=0$ being covered by $0\leq{\psi_s}\leq 1$.
\end{proof}

\begin{lem}[Bump profile]
\label{lem:bump-profile}
$\chi_{\nu,\ell}\in\mathscr{C}^\infty(\R^d,[0,1])$, and its per-axis factor takes the values
\begin{equation}
\label{eq:hh-values}
    H_{\nu}(y)H_{\nu}(\ell-y) = \begin{cases}
        0,&y \in \R \setminus [0,\ell]\\
        \in(0,1),&y \in (0,2\nu)\cup (\ell-2\nu,\ell)\\
        1,&y \in  [2\nu,\ell-2\nu]
    \end{cases}
\end{equation}
\end{lem}
\begin{proof}
$\theta_1\in\mathscr{C}^\infty(\R)$ by the fundamental theorem of calculus and ${\psi_s}\in\mathscr{C}^\infty(\R)$ (\cref{prop:psi-gevrey-property}). Since ${\psi_s}$ is supported and positive on exactly $(-1,1)$, it follows that $\theta_1(x)=0$ for $x \in (-\infty,-1]$, $\theta_1(x)=1$ for $x \in [1,\infty)$, and $0<\theta_1(x)<1$ is strictly increasing on $x\in (-1,1)$. Accordingly, $H_\nu(x)=0$ on $x\in (-\infty,0]$, $H_\nu(x)=1$ on $x\in [2\nu,\infty)$ and $H_\nu(x) \in (0,1)$ for $x\in (0,2\nu)$. The condition $\nu\leq\ell/4$ guarantees that the two ramps of $H_{\nu}(y)H_{\nu}(\ell-y)$ are disjoint, giving \eqref{eq:hh-values}. Finally $\chi_{\nu,\ell}\in\mathscr{C}^\infty$ follows as a finite product of $\mathscr{C}^\infty$ factors.
\end{proof}
\begin{lem}[Filling fraction]\label{lem:fill}
For $\varrho=\nu/\ell\in \left(0,\frac{1}{4}\right]$,
\begin{equation}
    q_1(\varrho)=\int_0^1\!du\, H_1(u/\varrho)\,H_1\left((1-u)/\varrho\right)=1-2\varrho.
\end{equation}
\end{lem}
\begin{proof}
$H_1(u/\varrho)=1$ for $u\geq2\varrho$ and $H_1((1-u)/\varrho)=1$ for $u\leq1-2\varrho$. For $\varrho\leq\tfrac14$ the ramps $[0,2\varrho]$, $[1-2\varrho,1]$ are disjoint and, by the substitution $u\mapsto1-u$, contribute equally:
\begin{equation}
\label{eq:ramp-integral}
\int_0^{2\varrho}H_1(u/\varrho)\,du=\varrho\int_{-1}^{1}\!dv\,\theta_1(v)=\varrho\int_{0}^{1}\!du\,(\theta_1(v)+\theta_1(-v))=\varrho,
\end{equation}
since by construction $\theta_1(v)+\theta_1(-v)=1$. The contribution to $q_1(\varrho)$ of the plateau $u\in [2\varrho,1-2\varrho]$, where both factors $H_1=1$, is trivially $1-4\varrho$. All together, $q_1(\varrho)=\varrho+(1-4\varrho)+\varrho=1-2\varrho$.
\end{proof}

\subsubsection{Proof of the bump family proposition}
\label{prf:bump-family}

\begin{proof}[Proof of \cref{prop:bump-family}]
By \cref{lem:bump-profile}, $\chi_{\nu,\ell}$ and all its derivatives vanish on the boundary of $[0,2\pi]^d$ for $\ell<2\pi$. These matching boundary derivatives give a positive $\mathscr C^\infty(\R^d)$ periodic lift of ${p_{\nu,\ell,\alpha}}$; hence each ${p_g}$ is a Gibbs state with $\mathscr{C}^\infty$ potential $E=-\tfrac1\beta\log{p}-\tfrac1\beta\log Z$. For property~(1) \eqref{eq:z-delta-ms}: $\chi_{\nu,\ell}$ is supported on $[0,\ell]^d$, which gives off-patch density ${p_{\nu,\ell,\alpha}} = 1/Z_{\nu,\ell,\alpha}$ at every point outside the cell. Thus, ${p_{\nu,\ell,\alpha}}$ is a bump (\cref{def:bump}) on $\Om = [0,\ell]^d$, with concentration
\begin{equation}
\label{eq:bump-torus-concentration}
    \delta = {p_{\nu,\ell,\alpha}}(\Om) = \frac{\ell^d\bigl(1+(\alpha-1)(1-2 \varrho)^d\bigr)}{Z_{\nu,\ell,\alpha}},
\end{equation}
by $q_1(\varrho)=1-2\varrho$ (\cref{lem:fill}); with $M={(2\pi/\ell)^{d}}$ and $S=(\alpha-1)(1-2\varrho)^d$, this is \eqref{eq:z-delta-ms}. Note that $Z_{\nu,\ell,\alpha}=(2\pi)^d +\ell^d(\alpha-1)(1-2\varrho)^d$, hence from $0<y=(\ell/2\pi)^d<1$ and $x = (\alpha-1)(1-2\varrho)^d>0$, it follows
\begin{equation}
\label{eq:delta-y-bound}
    \delta = y \frac{1+x}{1+xy}> y={\frac{\ell^d}{(2\pi)^d}} = {\frac{\vol(\Om)}{(2\pi)^d}}.
\end{equation}
Concerning property~(2): ${\bar{p}_{\nu,\ell,\alpha}}$ ranges over $[1/Z_{\nu,\ell,\alpha},\alpha/Z_{\nu,\ell,\alpha}]$, so ${\beta\,\Delta}=\log\alpha$. Any $g$-translate shares both properties by translation invariance of the volume.
\end{proof}

\subsubsection{Gevrey constants of the bump potential}

Certifying the membership of the bump family in the class $\mathscr{W}^{s}(\alpha,\xi,\rho)$ requires the seminorm constants $({\xi_E},{\rho_E})$ of the scaled Gibbs potential $\beta E$, computed in \cref{thm:gevrey-constants-f}--\cref{thm:bump-potential-gevrey-seminorm} below; the proof of \cref{cor:lower-bounds-gevrey} follows \cref{thm:bump-potential-gevrey-seminorm}. The potential of the bump density \eqref{eq:bump-torus-density} has periodic lift given on $[0,2\pi)^d$ by
\begin{equation}
\label{eq:potential-lift}
\beta{\bar{E}_{\nu,\ell,\alpha}}(x)=-F(x)+\log Z_{\nu,\ell,\alpha},\qquad F\coloneqq\log\bigl(1+(\alpha-1)\chi_{\nu,\ell}\bigr).
\end{equation}
By periodicity and invariance of $[\,\cdot\,]_{s,\rho}$ under additive constants, it suffices to bound the derivatives of $F$.

\begin{thm}[Gevrey constants of $F$]\label{thm:gevrey-constants-f}
Let $\alpha\ge e$ and $s>1$. Then $\forall\, a=(a_1,\dots,a_d)\in\N^{d}\setminus\{0\}$ and $\forall\,x\in\R^{d}$,
\begin{equation}\label{eq:main-deriv}
\bigl|\partial^{a}F(x)\bigr| \le \frac12\,\bigl(\onenorm{a}!\bigr)^{s}
\left(\frac{2\,\log^{s} \alpha} {{\mathfrak{r}_s}\,\nu}\right)^{\onenorm{a}},
\end{equation}
with ${\mathfrak{r}_s}\coloneqq\min\Bigl\{r_1,\ \frac{Z_s}{e^{2}{q_s}}\Bigr\}$ depending on $t=1/(s-1)$ via
\begin{equation}
\label{eq:constants}
r_1\coloneqq\min\left\{\frac{1}{8(t+1)4^{s}},\ \frac{\bigl(t\,2^{-(t+1)}\bigr)^{1/t}}{16(t+1)}\right\},\quad
{q_s}\coloneqq\max\left\{4,\ 2^{t+2}t\right\}.
\end{equation}
\end{thm}

\begin{rmk}
\label{rem:rho-s-bounds}
For every $s>1$ one has ${\mathfrak{r}_s}\leq r_1\leq \frac{1}{8(t+1)4^{s}}\leq\frac{1}{32}$, since $t>0$ gives $8(t+1)4^{s}\geq32$.
\end{rmk}

The proof of \cref{thm:gevrey-constants-f} follows assembling the series of technical \crefrange{lem:psi-derivative-bound}{lem:chi-derivative-bound}: derivative bounds for ${\psi_s}$ (\cref{lem:psi-derivative-bound}) and for $\chi_{\nu,\ell}$ (\crefrange{lem:theta-lower-bound}{lem:chi-derivative-bound}) with the Lemma about Gevrey bounds under composition with $\log$ (\cref{lem:log-composition-gevrey-bound}). Throughout we denote $\varphi(u)\coloneqq(1-u^{2})^{-t}$ and ${\psi_s}(u)=e^{-\varphi(u)}$.

\begin{lem}[Derivative bound on ${\psi_s}$]\label{lem:psi-derivative-bound}
For every $\Lambda\ge1$, $k\in\N$, $y\in\R$, with $r_1$ as in \eqref{eq:constants},
\begin{equation}\label{eq:psi-derivative-ineq}
\bigl|{\psi_s}^{(k)}(y)\bigr| \le e\,(k!)^{s}\left(\frac{\Lambda^{s}}{r_1}\right)^{k}
\max\bigl\{{\psi_s}(y), e^{-\Lambda}\bigr\}.
\end{equation}
\end{lem}
\begin{proof}
For $k=0$ the claim is trivial. Since ${\psi_s}$ vanishes on $|y|\geq1$ and is even, it suffices to treat $y\in[0,1)$. By \eqref{eq:pw1}, for all $k\in\N$, $y\in[0,1)$ and $0<\lambda<\sin\bigl(\min\{\tfrac{\pi}{2},\tfrac{\pi}{4t}\}\bigr)$,
\begin{equation}\label{eq:pw1-recall}
\bigl|{\psi_s}^{(k)}(y)\bigr|\leq \frac{k!}{\lambda^{k}(1-y)^{k}}\,
\exp\bigl(-c(\lambda)\,\varphi(y)\bigr),
\qquad
c(\lambda)\coloneqq\frac{\cos(2t\arcsin\lambda)}{(1+\lambda)^{2t}}.
\end{equation}
Fix $\lambda\coloneqq1/(16t\Lambda+2)$, admissible since $\lambda\leq\tfrac12<\sin\tfrac\pi2$ for $t<\tfrac12$, and $\lambda<\tfrac{1}{2t}\leq\sin\tfrac{\pi}{4t}$ for $t\geq\tfrac12$. The inequalities $\arcsin x\leq\tfrac\pi2x$, $\cos x\geq1-\tfrac{x^{2}}{2}$ and $(1+x)^{-2t}\geq1-2tx$ give
\begin{equation}
    c(\lambda) \geq \Bigl(1-\tfrac{\pi^{2}}{2}t^{2}\lambda^{2}\Bigr)(1-2t\lambda) \geq 1-2t\lambda-\tfrac{\pi^{2}}{2}t^{2}\lambda^{2} \geq 1-4t\lambda,
\end{equation}
the last step by $\lambda=1/(16t\Lambda+2)$; hence
\begin{equation}
c(\lambda)\geq 1-4t\lambda\geq 1-\frac{1}{4\Lambda} \ge \frac{3}{4}.
\end{equation}
An upper bound for the rhs of \eqref{eq:pw1-recall} splits into two cases. On $\{y\in[0,1)\mid{\varphi}(y)\leq 4\Lambda\}$, one has $(1-c(\lambda))\varphi(y)\le\tfrac{\varphi(y)}{4\Lambda}\leq 1$, hence
\begin{equation}
   e^{- c(\lambda)\varphi(y)}\leq e^{1- \varphi(y)}= e\,{\psi_s}(y).
\end{equation}
For the denominator in \eqref{eq:pw1-recall}, $(1-y)^{t}\geq(1-y^{2})^{t}/2^{t}=1/(2^{t}\varphi(y))$, which together with the hypothesis $\varphi(y)\leq 4 \Lambda$ give $1-y\geq\tfrac12(4\Lambda)^{-1/t}$. Therefore,
\begin{equation}
\lambda(1-y)=\frac{1-y}{16t\Lambda+2} \geq \frac{(4\Lambda)^{-1/t}}{32(t+1)\Lambda}
=\frac{1}{8(t+1)4^{s}\,\Lambda^{s}}
\geq \frac{r_1}{\Lambda^{s}},
\end{equation}
where the last step uses $r_1\le 1/\bigl(8(t+1)4^{s}\bigr)$ from \eqref{eq:constants}. Hence \eqref{eq:pw1-recall} and $k!\leq(k!)^{s}$ yield $\bigl|{\psi_s}^{(k)}(y)\bigr|\leq e\,(k!)^{s}\left(\Lambda^{s}/r_1\right)^{k}\,{\psi_s}(y)$.
On $\{y\in[0,1)\mid{\varphi}(y)\geq 4\Lambda\}$, with $w\coloneqq1-y$, one has $\varphi(y)\geq(2w)^{-t}$, hence $c(\lambda)\varphi(y)\geq \tfrac{3\varphi(y)}{4}\geq\tfrac{w^{-t}}{2^{t+1}}+\Lambda$, and \eqref{eq:pw1-recall} yields
\begin{equation}
\label{eq:psi-deriv-large-phi-sup}
\bigl|{\psi_s}^{(k)}(y)\bigr|
\leq \sup_{w>0}\Bigl[\frac{k!\, e^{-w^{-t}/2^{t+1}}}{(\lambda w)^{k}}\Bigr]\,e^{-\Lambda}
= k!\,\lambda^{-k}\Bigl(\frac{k\,2^{t+1}}{e\,t}\Bigr)^{k/t}e^{-\Lambda},
\end{equation}
the supremum attained at the unique critical point $w^{t}=t2^{-(t+1)}/k$. Moreover, since $\Lambda\geq 1$,
\begin{equation}
\label{eq:psi-deriv-large-phi-bound}
\lambda^{-1}\bigl(t2^{-(t+1)}\bigr)^{-1/t}
\le 16(t+1)\bigl(t2^{-(t+1)}\bigr)^{-1/t}\Lambda
= 32(t+1)(2/t)^{1/t}\Lambda
\le \Lambda^{s}/r_1,
\end{equation}
the last step by $r_1\le \bigl(t2^{-(t+1)}\bigr)^{1/t}/\bigl(16(t+1)\bigr)$ from \eqref{eq:constants}. Finally, using $k!\ge(k/e)^{k}$ one gets
\begin{equation}
\label{eq:psi-deriv-large-phi-final}
\bigl|{\psi_s}^{(k)}(y)\bigr|
\leq (k!)^{s}\Bigl(\frac{\Lambda^{s}}{r_1}\Bigr)^{k}e^{-\Lambda}.
\end{equation}
This completes the proof.
\end{proof}

\begin{lem}[Lower bound for $\theta_1(u)$]\label{lem:theta-lower-bound}
For $\theta_1(u), {Z_s}$ in \eqref{eq:mollified-step-def} and with ${q_s}$ in \eqref{eq:constants}, it holds:
\begin{equation}\label{eq:theta-lower-bound-ineq}
{Z_s}\,\theta_1(u) \ge \frac{{\psi_s}(u)}{e\,{q_s}\,\varphi(u)^{s}},\quad \forall\, u\in(-1,1).
\end{equation}
\end{lem}
\begin{proof}
Set $\eta \coloneqq\min\bigl\{(1+u)/{2}, \bigl(2^{t+2}t\,\varphi(u)^{s}\bigr)^{-1}\bigr\}$. In particular, this guarantees
\begin{equation}
    u-\eta \geq \tfrac{u-1}{2}>-1 \Longrightarrow [u-\eta,u]\subset(-1,1).
\end{equation}
Moreover, $\eta\ge\bigl({q_s}\varphi(u)^{s}\bigr)^{-1}$ follows from the definition of $\eta$: for the second entry of $\min$ because ${q_s}\ge2^{t+2}t$, for the first because
\begin{equation}
\tfrac{1+u}{2}\ge\tfrac{(1-u)(1+u)}{4}=\tfrac14\,\varphi(u)^{-1/t}\ge\tfrac14\,\varphi(u)^{-s}\ge\bigl({q_s}\varphi(u)^{s}\bigr)^{-1},
\end{equation}
using $s\geq 1/t$ and $q_s \geq 4$. Furthermore, $\forall v\in[u-\eta,u]$, $1-v\geq1-u$ and $1+v\geq\tfrac{1+u}{2}$ give $1-v^{2}\ge\tfrac12(1-u^{2})$, and accordingly
\begin{equation}
|\varphi'(v)| = \frac{2t|v|}{(1-v^{2})^{t+1}} \le 2t\Bigl(\frac{2}{1-u^{2}}\Bigr)^{t+1} = 2^{t+2}t\,\varphi(u)^{s} \leq \frac{1}\eta,
\end{equation}
so by the mean value theorem $\varphi(v)\leq{\varphi}(u)+1$, for any $v\in [u-\eta,u]$. All together,
\begin{equation}
{Z_s}\,\theta_1(u) \geq \int_{u-\eta}^{u}e^{-\varphi(v)}\,dv \geq \eta\,e^{-\varphi(u)-1} \geq \frac{{\psi_s}(u)}{e\,{q_s}\,\varphi(u)^{s}}.
\end{equation}
This completes the proof.
\end{proof}

\begin{lem}[Pointwise bound on derivatives of $\chi_{\nu,\ell}$]\label{lem:chi-derivative-bound}
Let $\Lambda\coloneqq\log \alpha \geq 1$. The following pointwise bounds hold for $0<\nu\leq\ell/4$:
\begin{equation}
    \label{eq:chi-derivative-bound}
\bigl|\partial^{a}\chi(x)\bigr|\le(a!)^{s}\,(\Lambda^{s}/{\mathfrak{r}_s}\nu)^{\onenorm{a}}\,\max\{\chi(x),e^{-\Lambda}\}, \forall\,a\in\N^{d}, x\in\R^{d}.
\end{equation}
\end{lem}
\begin{proof}
By the fundamental theorem of calculus, $\theta_1^{(k)}(u)={\psi_s}^{(k-1)}(u)/{Z_s}$ for $k\geq1$; it suffices to treat $u\in(-1,1)$, since ${\psi_s}^{(k-1)}$ vanishes elsewhere. \cref{lem:theta-lower-bound} gives ${\psi_s}(u)\le e\,{q_s}\,\varphi(u)^{s}\,{Z_s}\theta_1(u)$. For $\varphi(u)\le\Lambda$, it guarantees ${\psi_s}(u)\le e\,{q_s}\Lambda^{s}{Z_s}\theta_1(u)$, while for $\varphi(u)>\Lambda$ instead ${\psi_s}(u)=e^{-\varphi(u)}<e^{-\Lambda}$, and $e\,{q_s}\Lambda^{s}\geq 1$. In either case
\begin{equation}
\label{eq:chi-bound-intermediate}
\max\{{\psi_s}(u),e^{-\Lambda}\} \le e\,{q_s}\Lambda^{s}\,\max\{{Z_s}\theta_1(u), e^{-\Lambda}\}\leq e\,{q_s}\Lambda^{s}\,\max\{\theta_1(u), e^{-\Lambda}\},
\end{equation}
the last inequality by ${Z_s}\leq2e^{-1}<1$. Define $\mathfrak{r}_s \coloneqq \min\{r_1, {Z_s}/(e^{2}{q_s})\}$; it satisfies $\tfrac{e^{2}{q_s}r_1}{Z_s}(\Lambda^{s}/r_1)^{k}\leq(\Lambda^{s}/{\mathfrak{r}_s})^{k}$ for $k\geq1$, hence \cref{lem:psi-derivative-bound} and \eqref{eq:chi-bound-intermediate} give
\begin{align*}
\begin{aligned}
    \bigl|\theta_1^{(k)}(u)\bigr| = \frac{\bigl|{\psi_s}^{(k-1)}(u)\bigr|}{Z_s}
&\leq \frac{e\,\bigl((k-1)!\bigr)^{s}}{Z_s}\Bigl(\frac{\Lambda^{s}}{r_1}\Bigr)^{k-1}\max\{{\psi_s}(u),e^{-\Lambda}\}\\
&\leq\frac{e^{2}{q_s}r_1}{Z_s}\,(k!)^{s}\Bigl(\frac{\Lambda^{s}}{r_1}\Bigr)^{k}\max\{\theta_1(u),e^{-\Lambda}\}\\
&\leq (k!)^{s}\Bigl(\frac{\Lambda^{s}}{{\mathfrak{r}_s}}\Bigr)^{k}\max\{\theta_1(u),e^{-\Lambda}\}.
\end{aligned}
\end{align*}
The same bound is now transferred to $g(u)\coloneqq H_{\nu}(u)H_{\nu}(\ell-u)$ and then to $\chi_{\nu,\ell}=\prod_{k}g(x_k)$. By the chain rule $H_\nu^{(k)}(u)=\nu^{-k}\theta_1^{(k)}(u/\nu-1)$, supported in $(0,2\nu)$ and likewise $\tfrac{d^{k}}{du^{k}}H_\nu(\ell-u)= (-1)^k \nu^{-k}\theta_1^{(k)}((\ell-u)/\nu-1)=(-1)^k H_{\nu}^{(k)}(\ell-u)$,  supported in $(\ell-2\nu,\ell)$. For $\nu\le\ell/4$ the two supports are disjoint, therefore
\begin{equation}
    g^{(k)}(u)= \begin{cases}
        H_\nu^{(k)}(u)H_\nu(\ell-u),&u\in(0,2\nu)\\
        (-1)^k H_\nu^{(k)}(\ell-u)H_\nu(u),& u\in(\ell-2\nu,\ell)\\
        0,&\mathrm{otherwise}.
    \end{cases}
\end{equation}
From the bound on $\vert\theta_1^{(k)}(u)\vert$ it follows that:
\begin{equation}
    \bigl|H_{\nu}^{(k)}(u)\bigr|
 \leq (k!)^{s}\Bigl(\frac{\Lambda^{s}}{\nu {\mathfrak{r}_s}}\Bigr)^{k}\max\{H_{\nu}(u),e^{-\Lambda}\},
\end{equation}
and therefore, for any $u$ and since $H_{\nu}(u) \leq 1$:
\begin{equation}
    \bigl|g^{(k)}(u)\bigr|
 \leq (k!)^{s}\left(\frac{\Lambda^{s}}{\nu {\mathfrak{r}_s}}\right)^{k}\max\{g(u),e^{-\Lambda}\}.
\end{equation}
Finally $\partial^{a}\chi(x)=\prod_{j=1}^d g^{(a_j)}(x_j)$, so that
\begin{equation}
\bigl|\partial^{a}\chi(x)\bigr|\leq(a!)^{s}\left(\frac{\Lambda^{s}}{{\mathfrak{r}_s}\nu}\right)^{\onenorm{a}}\prod_{j=1}^{d}\max\{g(x_j),e^{-\Lambda}\}\leq(a!)^{s}\left(\frac{\Lambda^{s}}{{\mathfrak{r}_s}\nu}\right)^{\onenorm{a}}\max\{\chi(x),e^{-\Lambda}\},
\end{equation}
the last step by $g\leq1$ and $e^{-\Lambda}\leq1$.
\end{proof}

\begin{lem}[Gevrey-like bound for composition with $\log$]\label{lem:log-composition-gevrey-bound}
Let $\Omega\subseteq\R^{d}$ be open, $h\in\mathscr C^{\infty}(\Omega,\R^+)$, such that for $s\geq 1$, $K>0$, $R>0$, it satisfies
\begin{equation}\label{eq:log-composition-hypothesis}
\bigl|\partial^{a}h(x)\bigr| \le K\,\bigl(\onenorm{a}!\bigr)^{s}\,R^{\onenorm{a}}\,h(x)
\qquad\forall \,a\in\N^{d}\setminus\{0\},\ x\in\Omega.
\end{equation}
Then $F\coloneqq\log h$ satisfies
\begin{equation}\label{eq:log-composition-conclusion}
\bigl|\partial^{a}F(x)\bigr| \le \frac{K}{1+K}\,\bigl(\onenorm{a}!\bigr)^{s}\,\bigl((1+K)R\bigr)^{\onenorm{a}}\qquad\forall \,a\in\N^{d}\setminus\{0\},\quad \forall\, x\in\Omega.
\end{equation}
\end{lem}

\begin{proof}
For $n=1$, $|\partial_jF|=|\partial_jh|/h\le KR$. For $n\geq2$, we assume \eqref{eq:log-composition-conclusion} at every positive order less than $n$. We first apply the Leibniz formula to $h(x)=e^{F(x)}$, giving
\begin{equation}
\label{eq:leibniz-expansion}
\partial_j^{n} h(x) =\partial_j^{n-1} \left(\partial_j F(x) e^{F(x)}\right)  = \sum_{m=0}^{n-1} \binom{n-1}m \partial_j^{m+1}F(x) \partial_j^{n-m-1} h(x),
\end{equation}
and we isolate the top term $m=n-1$ to get
\begin{equation}
\label{eq:faa-di-bruno-1d}
    \partial_j^n F(x) = \frac{\partial_j^nh(x)}{h(x)} - \sum_{m=0}^{n-2}\binom{n-1}m \frac{\partial_j^{n-m-1} h(x)}{h(x)} \partial_j^{m+1} F(x).
\end{equation}
In general dimension, we take $a\in\N^d$ with $\onenorm{a}=n$ and choose $j$ with $a_j\geq1$. We set $a'=a-e_j$ and, as in the derivation of \eqref{eq:faa-di-bruno-1d}, isolate the top term $\partial^a F(x)$ to get
\begin{equation}
\label{eq:derivative-f-expansion}
\partial^{a}F = \frac{\partial^{a}h}h - \sum_{0\le b\lneq a'}\binom{a'}b\,\frac{\partial^{a'-b}h}h\, \partial^{b+e_j}F,
\qquad
\binom{a'}b\coloneqq\prod_{i=1}^{d}\binom{a_i'}{b_i},
\end{equation}
where $b\lneq a'$ means that $\forall i$, $0\leq b_i\leq a'_i$ and $b \neq a'$. Denote $m\coloneqq\onenorm{b}\le n-2$; since $\onenorm{a'-b}=n-1-m\geq 1$ the hypothesis \eqref{eq:log-composition-hypothesis} applies:
\begin{equation}
\label{eq:intermediate-bound-h}
    \left \vert{\partial^{a'-b}h}(x) \right \vert \leq K ((n-1-m)!)^{s}  R^{n-1-m} h(x).
\end{equation}
We apply the inductive hypothesis to $\partial^{b+e_j}F$, whose order is $m+1<n$. Grouping the summations on $b$ by $m$ in \eqref{eq:derivative-f-expansion}, using \eqref{eq:intermediate-bound-h} and then the multivariate identity (see e.g.~\cite[Identity~2.5]{mestrovic2018several})
\begin{equation}
\label{eq:multinomial-sum-identity}
\sum_{0\le b\lneq a'}\binom{a'}b=\sum_{m=0}^{n-2}\,\sum_{\onenorm{b}=m, b\le a'}\binom{a'}b =\sum_{m=0}^{n-2}
\binom{n-1}m,
\end{equation}
one gets
\begin{align*}
    \begin{aligned}
\bigl|\partial^{a}F\bigr|&\leq K(n!)^{s}R^{n}
+\sum_{m=0}^{n-2}\binom{n-1}m\,K\bigl((n-1-m)!\bigr)^{s}R^{n-1-m} \!\times \! \frac{K}{1+K}\bigl((m+1)!\bigr)^{s}\bigl((1+K)R\bigr)^{m+1}\\
&= K(n!)^{s}R^{n}
+\frac{K^{2}R^{n}}{1+K}\sum_{m=0}^{n-2}\binom{n-1}m\left((n-1-m)!\,(m+1)!\right)^{s}(1+K)^{m+1}.
    \end{aligned}
\end{align*}
Finally $(n-1-m)!\,(m+1)!\le n!$ and $s\geq 1$ imply
\begin{equation}
\label{eq:binomial-factorial-bound}
\binom{n-1}m\bigl((n-1-m)!(m+1)!\bigr)^{s}\leq(n!)^{s-1}(n-1)!(m+1)\le(n!)^{s},
\end{equation}
and accordingly the summation in the previous formula is bounded by
\begin{equation}
\label{eq:geometric-sum-bound}
\begin{aligned}
\sum_{m=0}^{n-2}\!\binom{n-1}m\!\left((n-1-m)!\,(m+1)!\right)^{s}(1+K)^{m+1}
&\leq (n!)^{s}\sum_{m=0}^{n-2} (1+K)^{m+1}\\
&=(n!)^{s} \frac{(1+K)^n\!-\!(1+K)}K.
\end{aligned}
\end{equation}
Substitution gives, by induction,
\begin{equation}
\label{eq:log-composition-conclusion-proof}
\bigl|\partial^{a}F\bigr|\leq \frac{K}{K+1}\,(n!)^{s}\bigl((1+K) R\bigr)^{n},
\end{equation}
which is \eqref{eq:log-composition-conclusion}, completing the proof.
\end{proof}

\begin{proof}[Proof of \cref{thm:gevrey-constants-f}]
Set $h(x) \coloneqq1+(\alpha-1)\chi_{\nu,\ell}(x)$, so $F=\log h$, $h\in\mathscr C^{\infty}(\R^{d},[1,\infty))$.
Let $a\neq 0$; by \cref{lem:chi-derivative-bound} with $\Lambda=\log\alpha\ge1$, and using $a!\leq\onenorm{a}!$,
\begin{equation}
\label{eq:h-derivative-bound}
\bigl|\partial^{a}h(x)\bigr|=(\alpha-1)\bigl|\partial^{a}\chi(x)\bigr|
\leq \left(\onenorm{a}!\right)^{s}\left(\frac{\Lambda^{s}}{{\mathfrak{r}_s}\,\nu}\right)^{\onenorm{a}}(\alpha-1)\max\{\chi(x),e^{-\Lambda}\},
\end{equation}
and, since $e^{-\Lambda}=1/\alpha$,
\begin{equation}
\label{eq:h-alpha-bound}
(\alpha-1)\max\{\chi(x),e^{-\Lambda}\}
=\max\Bigl\{(\alpha-1)\chi(x), \tfrac{\alpha-1}\alpha\Bigr\}
\le 1+(\alpha-1)\chi(x) = h(x).
\end{equation}
Thus $h$ satisfies \eqref{eq:log-composition-hypothesis} with $K=1$, $R=\Lambda^{s}/({\mathfrak{r}_s}\nu)$, and \cref{lem:log-composition-gevrey-bound} yields \eqref{eq:main-deriv}.
\end{proof}

\begin{thm}[Gevrey seminorm of the bump potential]\label{thm:bump-potential-gevrey-seminorm}
Let $s>1$, $\alpha\ge e$, and set $\Lambda\coloneqq\log\alpha\geq1$, with ${\mathfrak{r}_s}$ the constant of \eqref{eq:constants}. The Gibbs potential of the bump,
\begin{equation}
{E_{\nu,\ell,\alpha}} = -\tfrac{1}\beta\,\log{\bar{p}_{\nu,\ell,\alpha}} = -\tfrac{1}\beta\,F+\tfrac{1}\beta\,\log Z_{\nu,\ell,\alpha},
\end{equation}
has scaled potential $\beta {E_{\nu,\ell,\alpha}}\in\dot{\mathscr{G}}^{s}({\xi_E},{\rho_E})$ (as in \cref{def:gevrey-seminorm}) with
\begin{equation}
\label{eq:bump-potential-gevrey-constants}
{\xi_E} = \frac{1}{2},\qquad {\rho_E} = \frac{{\mathfrak{r}_s}\,\nu}{2\,d^{s}\,(\log\alpha)^{s}},
\end{equation}
and barrier amplitude $\alpha=e^{\beta\,\Delta}$.
\end{thm}

\begin{proof}[Proof of \cref{thm:bump-potential-gevrey-seminorm}]
For $a\neq0$, \cref{thm:gevrey-constants-f} gives
\begin{equation}
\label{eq:f-derivative-bound-nonzero}
    \bigl|\partial^{a}F(x)\bigr| \leq \tfrac{1}{2}\,\bigl(\onenorm{a}!\bigr)^{s}\left(\frac{2\,\Lambda^{s}}{{\mathfrak{r}_s}\,\nu}\right)^{\onenorm{a}}.
\end{equation}
The additive constant $\log Z_{\nu,\ell,\alpha}$ in the affine relation $\beta{\bar E_{\nu,\ell,\alpha}}=-F+\log Z_{\nu,\ell,\alpha}$ leaves all $a\neq0$ derivatives unchanged, so $\partial^{a}(\beta E)=-\partial^{a}F$ and
\begin{equation}
\label{eq:potential-deriv-bound-sum-convention}
\bigl|\partial^{a}(\beta E)(x)\bigr| \leq \tfrac{1}{2}\,\bigl(\onenorm{a}!\bigr)^{s}\left(\frac{2\,\Lambda^{s}}{{\mathfrak{r}_s}\,\nu}\right)^{\onenorm{a}}.
\end{equation}
Converting from the sum-convention $\onenorm{a}!$ to the product-convention $a!$ of \cref{def:gevrey-function} through $\onenorm{a}!\leq d^{\onenorm{a}}\,a!$ multiplies the inverse radius by $d^{s}$:
\begin{equation}
\label{eq:potential-deriv-bound-product-convention}
\bigl|\partial^{a}(\beta E)(x)\bigr| \leq \tfrac{1}{2}\,(a!)^{s}\left(\frac{2\,d^{s}\,\Lambda^{s}}{{\mathfrak{r}_s}\,\nu}\right)^{\onenorm{a}},\qquad {\rho_E}^{-1} = \frac{2\,d^{s}\,(\log\alpha)^{s}}{{\mathfrak{r}_s}\,\nu}.
\end{equation}
The prefactor $\tfrac{1}{2}$ bounds the nonzero-order seminorm, $[\beta E]_{s,{\rho_E}}\leq {\xi_E}$, read off the display above; the additive constant $\log Z_{\nu,\ell,\alpha}$ affects only the order-zero term, which is gauge and enters solely through ${\beta\,\Delta}=\osc F=\log\alpha$, i.e.\ the invariant $\alpha=e^{\beta\,\Delta}$.
\end{proof}

\subsubsection{Proof of the classical lower bound corollary}
\label{prf:lower-bounds-gevrey}

\begin{proof}[Proof of \cref{cor:lower-bounds-gevrey}]
By \cref{thm:bump-potential-gevrey-seminorm} the potentials of the translates $\{{p_g}\}_{g\in\T^d}$ of \cref{prop:bump-family} satisfy $[\beta {E_{p_g}}]_{s,{\rho_E}}\leq\tfrac{1}{2}$ and $e^{\beta{\Delta}}=\alpha$; from $\xi\geq\tfrac{1}{2}$ and $\rho\leq {\rho_E}$, $[\beta {E_{p_g}}]_{s,\rho}\leq[\beta {E_{p_g}}]_{s,{\rho_E}}\leq\tfrac12\leq \xi$. Moreover, the ceiling condition of \eqref{eq:potential-class} holds: indeed ${\mathfrak{r}_s}\leq\tfrac{1}{32}$ (see remark below~\eqref{eq:constants}) and ${\nu\leq\ell/4<\pi/2}$. Therefore one has ${\mathfrak{r}_s\,\nu\leq\tfrac{\pi}{64}\leq\pi\leq\pi\, d^{1+s}\log^{s-1}\alpha}$ for $\alpha\geq e$, that guarantees the ceiling condition $\log\alpha\leq{\pi d\,\xi_E/\rho_E}$. Hence $\{{p_g}\}\subset{\mathscr{W}^{s}(\alpha,\xi,\rho)}$ (\cref{def:potential-class}).

Note that \cref{thm:hide-and-seek-lower-bound} with $v_0=\vol(\Om)=\ell^d$ bounds the query complexity of the subfamily by $M(\delta-\eps)-1$, where $M=(2\pi/\ell)^d$. Finally, the bound transfers by monotonicity with respect to inclusions to the family ${\mathscr{W}^{s}(\alpha,\xi,\rho)} \supset \{{p_g}\}$. This completes the proof.
\end{proof}

\end{document}